\documentclass[aps,pra,10pt,reprint,superscriptaddress,showpacs]{revtex4-2}

\usepackage{xcolor}
\usepackage{physics}
\usepackage{amssymb}
\usepackage{graphicx}
\usepackage{mathtools}
\usepackage{amsmath}
\usepackage{hyperref}
\usepackage{tikz}
\usetikzlibrary{calc}
\usetikzlibrary{arrows.meta, decorations.pathreplacing}
\usepackage{xifthen}
\usepackage{xcolor} 
\usepackage{xparse}

\usepackage{amsfonts}
\usepackage{centernot}
\usepackage{dsfont}
\usepackage[linesnumbered,ruled,vlined]{algorithm2e}

\usepackage{amsthm}
\newtheorem{theorem}{Theorem}
\newtheorem{definition}{Definition}
\newtheorem{lemma}{Lemma}
\newtheorem{corollary}{Corollary}
\newtheorem{proposition}{Proposition}
\newtheorem{observation}{Observation}

\DeclarePairedDelimiter\floor{\lfloor}{\rfloor}
\DeclarePairedDelimiter\ceil{\lceil}{\rceil}

\newcommand\equalhat{
\let\savearraystretch\arraystretch
\renewcommand\arraystretch{0.3}
\begin{array}{c}
\stretchto{
    \scalerel*[\widthof{=}]{\wedge}
    {\rule{1ex}{3ex}}
}{0.5ex}\\ 
=
\end{array}
\let\arraystretch\savearraystretch
}

\hypersetup{
colorlinks = true,
urlcolor = blue,
linkcolor = blue,
citecolor = blue
}

\newcommand\circuitResetGlobalParameters{

        \definecolor{circuitcolorbg}{RGB}{255,255,255}
        \definecolor{circuitcolorfg}{RGB}{0,0,0}

        \def\circuitlinewidth{1pt}
        \def\circuitCNOTControl{2.2pt}
        \def\circuitCNOTTarget{3.7pt}
        \def\circuitRoundedCorners{2pt}
        \def\circuitMultiLabelMargin{0.1}
        \def\circuitMultiLabelLineWidth{0.5pt}
        \def\circuitMeasureInnerLineWidth{0.7pt}
        
        \def\circuitlinespacing{0.5}
        
        \def\circuitGateStandardWidth{\circuitlinespacing*1.2} 
        \def\circuitGateHeight{\circuitlinespacing*0.8}
        \def\circuitMultiGateStandardWidth{\circuitlinespacing*2.2}
        \def\circuitMultiGateHeight{\circuitlinespacing*0.4}
        \def\circuitStandardGateDistance{\circuitlinespacing*0.4}
        \def\circuitMeasureSize{0.4}
        \def\circuitMeasureArcOffset{0.85} 
        \def\circuitMeasureArrowLength{1.05} 
        \def\circuitMeasureArrowAngle{37}
        \def\circuitMeasureArrowHeadLength{4pt}
        \def\circuitMeasureArrowHeadWidth{3pt}
}
\circuitResetGlobalParameters

\newcommand\circuitInitXAt[1]{\coordinate (circuitTheX) at (#1,0);}

\NewDocumentCommand\circuitAdvanceXBy{O{0} m}{
	\coordinate (circuitTheX) at ($(circuitTheX) + (#2, 0)$);
	\ifthenelse{\isempty{#1} \OR #1=0}{}{\coordinate (circuitTheX) at ($(circuitTheX) + (\circuitStandardGateDistance, 0)$);}
}

\NewDocumentCommand\circuitCNOT{O{circuitcolorfg} O{myCenterNode} m m m}{ 
	\ifthenelse{\isempty{#1}}{\def\mycolorfg{circuitcolorfg}}{\def\mycolorfg{#1}}
	\ifthenelse{\isempty{#3}}{\coordinate (myXCoord) at (circuitTheX);}{\coordinate (myXCoord) at (#3, 0);}

	\ifthenelse{\isempty{#2}}{\def\mynodename{myCenterCoord}}{\def\mynodename{#2}}
	
	\coordinate (myControlCoord) at ($(myXCoord) + (0, #4*\circuitlinespacing)$);
	\coordinate (myTargetCoord) at ($(myXCoord) + (0, #5*\circuitlinespacing)$);

	\coordinate (\mynodename) at ($0.5*(myControlCoord) + 0.5*(myTargetCoord)$);
	
	\draw[draw=\mycolorfg, fill=\mycolorfg] (myControlCoord) circle (\circuitCNOTControl); 
	\draw[line width=\circuitlinewidth, draw=\mycolorfg] (myTargetCoord) circle (\circuitCNOTTarget); 
	\draw[line width=\circuitlinewidth, draw=\mycolorfg] (myControlCoord) -- (myTargetCoord); 
	\draw[line width=\circuitlinewidth, draw=\mycolorfg] ($(myTargetCoord) + (0,-\circuitCNOTTarget)$) -- ($(myTargetCoord) + (0,\circuitCNOTTarget)$); 
	\draw[line width=\circuitlinewidth, draw=\mycolorfg] ($(myTargetCoord) + (-\circuitCNOTTarget,0)$) -- ($(myTargetCoord) + (\circuitCNOTTarget,0)$);
}

\NewDocumentCommand\circuitCZ{O{circuitcolorfg} O{myCenterNode} m m m}{ 
	\ifthenelse{\isempty{#1}}{\def\mycolorfg{circuitcolorfg}}{\def\mycolorfg{#1}}
	\ifthenelse{\isempty{#3}}{\coordinate (myXCoord) at (circuitTheX);}{\coordinate (myXCoord) at (#3, 0);}

	\ifthenelse{\isempty{#2}}{\def\mynodename{myCenterCoord}}{\def\mynodename{#2}}
	
	\coordinate (myControlCoord) at ($(myXCoord) + (0, #4*\circuitlinespacing)$);
	\coordinate (myTargetCoord) at ($(myXCoord) + (0, #5*\circuitlinespacing)$);

	\coordinate (\mynodename) at ($0.5*(myControlCoord) + 0.5*(myTargetCoord)$);
	
	\draw[draw=\mycolorfg, fill=\mycolorfg] (myControlCoord) circle (\circuitCNOTControl);
	\draw[draw=\mycolorfg, fill=\mycolorfg] (myTargetCoord) circle (\circuitCNOTControl); 
	\draw[line width=\circuitlinewidth, draw=\mycolorfg] (myControlCoord) -- (myTargetCoord); 
}

\NewDocumentCommand\circuitSingleGate{O{\circuitGateStandardWidth} O{circuitcolorbg} O{circuitcolorfg} O{myCenterCoord} m m m}{
	\ifthenelse{\isempty{#1}}{\def\mywidth{\circuitGateStandardWidth}}{\def\mywidth{#1}}
	\ifthenelse{\isempty{#2}}{\def\mycolorbg{circuitcolorbg}}{\def\mycolorbg{#2}}
	\ifthenelse{\isempty{#3}}{\def\mycolorfg{circuitcolorfg}}{\def\mycolorfg{#3}}
	\ifthenelse{\isempty{#5}}{\coordinate (myXCoord) at (circuitTheX);}{\coordinate (myXCoord) at (#5, 0);}
	\ifthenelse{\isempty{#4}}{\def\mynodename{myCenterCoord}}{\def\mynodename{#4}}
	
	\coordinate (\mynodename) at ($(myXCoord) + (0,#6*\circuitlinespacing)$);

	\draw[line width=\circuitlinewidth, fill=\mycolorbg, draw=\mycolorfg, rounded corners=\circuitRoundedCorners]
		($(\mynodename) - 0.5*(\mywidth, \circuitGateHeight)$) 
		rectangle ($(\mynodename) + 0.5*(\mywidth, \circuitGateHeight)$) node[midway, \mycolorfg] {#7};
}

\NewDocumentCommand\circuitMultiGate{O{\circuitMultiGateStandardWidth} O{circuitcolorbg} O{circuitcolorfg} O{myCenterCoord} m m m m}{
	\ifthenelse{\isempty{#1}}{\def\mywidth{\circuitMultiGateStandardWidth}}{\def\mywidth{#1}}
	\ifthenelse{\isempty{#2}}{\def\mycolorbg{circuitcolorbg}}{\def\mycolorbg{#2}}
	\ifthenelse{\isempty{#3}}{\def\mycolorfg{circuitcolorfg}}{\def\mycolorfg{#3}}
	\ifthenelse{\isempty{#4}}{\def\mynodename{myCenterCoord}}{\def\mynodename{#4}}
	\ifthenelse{\isempty{#5}}{\coordinate (myXCoord) at (circuitTheX);}{\coordinate (myXCoord) at (#5, 0);}
	
	\ifthenelse{#6 > #7}{\def \yUpper{#6*\circuitlinespacing} \def\yLower{#7*\circuitlinespacing}}{\def \yUpper{#7*\circuitlinespacing} \def\yLower{#6*\circuitlinespacing}}
	\coordinate (myLowerCoord) at ($(myXCoord) + (0,\yLower)$);
	\coordinate (myUpperCoord) at ($(myXCoord) + (0,\yUpper)$);
	
	\coordinate (\mynodename) at ($0.5*(myLowerCoord) + 0.5*(myUpperCoord)$);
	
	\draw[line width=\circuitlinewidth, fill=\mycolorbg, draw=\mycolorfg, rounded corners=\circuitRoundedCorners]
		($(myLowerCoord) - 0.5*(\mywidth, \circuitMultiGateHeight)$) 
		rectangle ($(myUpperCoord) + 0.5*(\mywidth, \circuitMultiGateHeight)$) node[midway, \mycolorfg] {#8};
	
}

\NewDocumentCommand\circuitMeasure{O{circuitcolorbg} O{circuitcolorfg} O{1,0} O{myCenterCoord} m m m}{
	\ifthenelse{\isempty{#1}}{\def\mycolorbg{circuitcolorbg}}{\def\mycolorbg{#1}}
	\ifthenelse{\isempty{#2}}{\def\mycolorfg{circuitcolorfg}}{\def\mycolorfg{#2}}
	\ifthenelse{\isempty{#3}}{\coordinate (myLabelOffset) at (\circuitMeasureSize,0);}{\coordinate (myLabelOffset) at ($\circuitMeasureSize*(#3)$);}
	\ifthenelse{\isempty{#5}}{\coordinate (myXCoord) at (circuitTheX);}{\coordinate (myXCoord) at (#5, 0);}
	\ifthenelse{\isempty{#4}}{\def\mynodename{myCenterCoord}}{\def\mynodename{#4}}
	
	\coordinate (\mynodename) at ($(myXCoord) + (0,#6*\circuitlinespacing)$);
	
	\coordinate (myLowerCoord) at ($(\mynodename) - 0.5*(\circuitMeasureSize,\circuitMeasureSize)$);
	
	\draw[line width=\circuitlinewidth, fill=\mycolorbg, draw=\mycolorfg] (myLowerCoord)
		rectangle ($(\mynodename) + 0.5*(\circuitMeasureSize,\circuitMeasureSize)$);
	
	\draw[line width=\circuitMeasureInnerLineWidth, draw=\mycolorfg] 
		($(myLowerCoord) + (0, \circuitMeasureSize*\circuitMeasureArcOffset)$) to[out=0, in=90]
		($(myLowerCoord) + (\circuitMeasureSize*\circuitMeasureArcOffset,0)$);
	
	\draw[-{Latex[length=\circuitMeasureArrowHeadLength, width=\circuitMeasureArrowHeadWidth]}, line width=\circuitMeasureInnerLineWidth, draw=\mycolorfg] 
		(myLowerCoord) -- +(\circuitMeasureArrowAngle:\circuitMeasureArrowLength*\circuitMeasureSize);
		
	\ifthenelse{\isempty{#7}}{}{\node[\mycolorfg] at ($(\mynodename) + (myLabelOffset)$) {#7};}

}

\NewDocumentCommand\circuitQubitLine{O{} O{} O{circuitcolorfg} m m m}{
	\ifthenelse{\isempty{#1}}{\def\myincommand{}}{\def\myincommand{node[left] {#1}}}
	\ifthenelse{\isempty{#2}}{\def\myoutcommand{}}{\def\myoutcommand{node[right] {#2}}}
	\ifthenelse{\isempty{#3}}{\def\mycolorfg{circuitcolorfg}}{\def\mycolorfg{#3}}
	
	\draw[line width=\circuitlinewidth, \mycolorfg] (#4, #6*\circuitlinespacing) \myincommand -- (#5,#6*\circuitlinespacing) \myoutcommand;
}

\NewDocumentCommand\circuitMultiInput{O{circuitcolorfg} m m m m}{
	\ifthenelse{\isempty{#1}}{\def\mycolorfg{circuitcolorfg}}{\def\mycolorfg{#1}}
	
	\ifthenelse{#3 > #4}{\def \yUpper{#3*\circuitlinespacing} \def\yLower{#4*\circuitlinespacing}}{\def \yUpper{#4*\circuitlinespacing} \def\yLower{#3*\circuitlinespacing}}
	
	\draw[line width=\circuitMultiLabelLineWidth, \mycolorfg] ($(#2, \yUpper) + (\circuitMultiLabelMargin, \circuitMultiLabelMargin)$)
		 -- ($(#2, \yUpper) + (- \circuitMultiLabelMargin, \circuitMultiLabelMargin)$)
		  -- node[left] {#5} ($(#2, \yLower) + (-\circuitMultiLabelMargin, -\circuitMultiLabelMargin)$)
		 -- ($(#2, \yLower) + (\circuitMultiLabelMargin, -\circuitMultiLabelMargin)$);
}

\NewDocumentCommand\circuitMultiOutput{O{circuitcolorfg} m m m m}{
	\ifthenelse{\isempty{#1}}{\def\mycolorfg{circuitcolorfg}}{\def\mycolorfg{#1}}
	
	\ifthenelse{#3 > #4}{\def \yUpper{#3*\circuitlinespacing} \def\yLower{#4*\circuitlinespacing}}{\def \yUpper{#4*\circuitlinespacing} \def\yLower{#3*\circuitlinespacing}}
	
	\draw[line width=\circuitMultiLabelLineWidth, \mycolorfg] ($(#2, \yUpper) + (-\circuitMultiLabelMargin, \circuitMultiLabelMargin)$)
		 -- ($(#2, \yUpper) + ( \circuitMultiLabelMargin, \circuitMultiLabelMargin)$)
		  -- node[right] {#5} ($(#2, \yLower) + (\circuitMultiLabelMargin, -\circuitMultiLabelMargin)$)
		 -- ($(#2, \yLower) + (-\circuitMultiLabelMargin, -\circuitMultiLabelMargin)$);
}

\colorlet{mgcolorfg}{blue!100}
\colorlet{mgcolorbg}{blue!20}

\colorlet{othergatecolorfg}{red!80!black}
\colorlet{othergatecolorbg}{red!80!black!10}

\colorlet{entangledstatecolor}{cyan!80!black}

\colorlet{entcolfg}{blue!100}
\colorlet{entcolbg}{blue!20}

\colorlet{discolfg}{red!80!black}
\colorlet{discolbg}{red!10}

\colorlet{gencolfg}{cyan!75!black}
\colorlet{gencolbg}{cyan!10}

\def\flatgatesize{0.15}
\def\flatgateadvance{0.3}

\def\circuitRoundedCorners{1.5pt}
\def\qubitBallRadius{1.6pt}

\newcommand{\introfiguretextsize}{\footnotesize}

\newcommand{\helpdrawEntangled}[3]{ 
\foreach\p in {#1,...,#2}{
    \draw[line width  =4pt, line cap = round] ($0.5*(#3,-\p)$) -- ++(0,-0.5);
    }
}

\newcommand{\helpdrawCircuitLines}[4]{
    \foreach\p in {#1,...,#2}{\circuitQubitLine{#3}{#4}{-\p}\draw[fill] ($0.5*(2*#3, -\p)$)
    circle[radius=\qubitBallRadius];}
}

\newcommand{\helpdrawCircuitLinesWithoutQubits}[4]{
    \foreach\p in {#1,...,#2}{\circuitQubitLine{#3}{#4}{-\p}\draw[fill] ($0.5*(2*#3, -\p)$);}
}

\newcommand{\helpdrawRSFdiagonal}[2]{
\coordinate (rsfStartX) at (circuitTheX);
    \foreach[count=\i, evaluate=\i as \pp using int(-\p-1)] \p in {#1,...,\number\numexpr #1+#2-1} {
	   \circuitMultiGate[\flatgatesize][mgcolorbg][mgcolorfg][rsfnode\i]{}{-\p}{\pp}{}
        \circuitAdvanceXBy{\flatgateadvance} }
\coordinate (rsfEndX) at (circuitTheX); 
\coordinate (circuitTheX) at (rsfStartX);
}

\newcommand{\helpdrawMGColorToFaint}{
        \colorlet{t1}{mgcolorfg}
        \colorlet{t2}{mgcolorbg}
        \colorlet{mgcolorfg}{lightmgcolorfg}
        \colorlet{mgcolorbg}{lightmgcolorbg}
}

\newcommand{\helpdrawMGColorReset}{
        \colorlet{mgcolorfg}{t1}
        \colorlet{mgcolorbg}{t2}
}

\newcommand{\helpdrawInnerProdSevenlines}[1]{
        \helpdrawCircuitLines{0}{6}{0}{#1}
        \foreach\p in {0,...,6}{\draw[fill] ($0.5*(2*#1, -\p)$) circle[radius=\qubitBallRadius];}
}

\newcommand{\helpdrawInnerProduct}{
\node at (4.8,2.9/2) {$=$};
\node at (9.3,2.9/2) {$=$};
\node at (13.8,2.9/2) {$=$};
\node at (18.2,2.9/2) {$=$};
\node at (4.8,3.8/2) {(i)};
\node at (9.3,3.8/2) {(ii)};
\node at (13.8,3.8/2) {(iii)};
\node at (18.2,3.8/2) {(iv)};
\node at (-4.5,2.9) {(a)};
\node at (0.5,2.9) {(b)};

\begin{scope}[xshift=-4.0cm]
    \begin{scope}[rotate = 90]
        \circuitInitXAt{0.25}
        \def\flatgateadvance{0.375} 

        \helpdrawInnerProdSevenlines{2.9}

        \helpdrawMGColorToFaint
        \helpdrawRSFdiagonal{0}{6}
        \coordinate (line1node1) at (rsfnode3);
        \coordinate (line1node1e) at (rsfnode4);
        \coordinate (line2node1) at (rsfnode5);;
        \coordinate (line2node1e) at (rsfnode6);
        \helpdrawRSFdiagonal{2}{4}
        \coordinate (line1node2) at (rsfnode2);
        \coordinate (line1node2e) at (rsfnode3);
        \coordinate (line2node2) at (rsfnode4);
        \helpdrawRSFdiagonal{4}{2}
        \coordinate (line1node3) at (rsfnode1);
        \coordinate (line1node3e) at (rsfnode2);
        \helpdrawMGColorReset
        \circuitAdvanceXBy{0.5}
        \circuitAdvanceXBy{\flatgateadvance}
        \circuitAdvanceXBy{\flatgateadvance}
        \circuitAdvanceXBy{\flatgateadvance}        \circuitAdvanceXBy{\flatgateadvance}

        \circuitMultiGate[\flatgatesize][mgcolorbg][mgcolorfg][addgate1]{}{-2}{-3}{}
        
        \draw[absorbalgarrow] 
        ($(addgate1) - (0.15,0)$) -- 
            ($(line1node1) + (\flatgatesize+0.2*\flatgateadvance, 0)$) --
            ($(line1node1e) - (\flatgatesize+0.2*\flatgateadvance, 0)$) -- 
            ($(line1node2) + (\flatgatesize+0.2*\flatgateadvance, 0)$) --
            ($(line1node2e) - (\flatgatesize+0.2*\flatgateadvance, 0)$) -- 
            ($(line1node3) + (\flatgatesize+0.2*\flatgateadvance, 0)$) --
            ($(line1node3e) - (\flatgatesize+0.2*\flatgateadvance, -0.07)$) -- (0,-2.68) -- (0,-2.85) -- ($(line1node3e) - (0.12,0.1)$);

        \circuitAdvanceXBy{\flatgateadvance}
        
        \circuitMultiGate[\flatgatesize][mgcolorbg][mgcolorfg][addgate2]{}{-4}{-5}{}

        \draw[absorbalgarrow] 
        ($(addgate2) - (0.15,0)$) -- 
            ($(line2node1) + (\flatgatesize+0.2*\flatgateadvance, 0)$) --
            ($(line2node1e) - (\flatgatesize+0.2*\flatgateadvance, 0)$) -- 
            ($(line2node2) + (0.15, 0)$);

    \end{scope}
\end{scope}

\begin{scope}[xshift=1cm]
    \begin{scope}[rotate = 90]
        \circuitInitXAt{0.25}

        \helpdrawCircuitLines{0}{6}{0}{2.9}
        \foreach\p in {0,...,6}{\draw[fill] ($0.5*(2*2.9, -\p)$) circle[radius=\qubitBallRadius];}

        \helpdrawMGColorToFaint
        \helpdrawRSFdiagonal{2}{4}
        \helpdrawRSFdiagonal{4}{2}
        \helpdrawRSFdiagonal{0}{1}
        \helpdrawMGColorReset

        \circuitAdvanceXBy{\flatgateadvance}
        \helpdrawRSFdiagonal{1}{5}

        \draw[absorbalgarrow] ($(rsfnode4) +(0.15,0)$) -- ++ (0.5,0);
        \draw[absorbalgarrow] ($(rsfnode3) +(0.15,0)$) --
            ++ ($(0.5,0) + (\flatgateadvance+\flatgatesize,0)$);
        \draw[absorbalgarrow] ($(rsfnode2) +(0.15,0)$) --
            ++ ($(0.5,0) + 2*(\flatgateadvance+\flatgatesize,0)$);
        \draw[absorbalgarrow] ($(rsfnode1) +(0.15,0)$) --
            ++ ($(0.5,0) + 3*(\flatgateadvance+\flatgatesize,0)$);
        
    \end{scope}
\end{scope}

\begin{scope}[xshift=5.5cm]
    \begin{scope}[rotate = 90]

        \circuitInitXAt{0.25}
        \helpdrawInnerProdSevenlines{2.9}

        \helpdrawMGColorToFaint
        \helpdrawRSFdiagonal{0}{1}
        \helpdrawRSFdiagonal{4}{2}
        \helpdrawRSFdiagonal{2}{4}
        \helpdrawMGColorReset
        \coordinate (circuitTheX) at (rsfEndX);
        \helpdrawRSFdiagonal{5}{-2}
        
        \draw[absorbalgarrow] ($(rsfnode2) -(0.15,0)$) --
            ++ (-0.5,0);
         \draw[absorbalgarrow] ($(rsfnode3) -(0.15,0)$) --
            ++ ($(-0.5,0) - (\flatgateadvance+2*\flatgatesize,0)$);
         \draw[absorbalgarrow] ($(rsfnode4) -(0.15,0)$) --
            ++ ($(-0.5,0) - (2*\flatgateadvance+4*\flatgatesize,0)$);

        \coordinate (circuitTheX) at (rsfEndX);
        \helpdrawMGColorToFaint
        \helpdrawRSFdiagonal{1}{1}
        \helpdrawMGColorReset
        
    \end{scope}
\end{scope}

\begin{scope}[xshift=10cm, yshift=0.45cm]
    \begin{scope}[rotate = 90]
        \circuitInitXAt{0.25}
        \helpdrawInnerProdSevenlines{2.0}

        \helpdrawMGColorToFaint
        \helpdrawRSFdiagonal{4}{2}
        \helpdrawRSFdiagonal{0}{1}
        \helpdrawRSFdiagonal{2}{1}
        \coordinate (cxforpainteddiag) at (rsfEndX);
        \foreach \x in {1,...,5} {\circuitAdvanceXBy{\flatgateadvance}}
        \helpdrawRSFdiagonal{1}{1}
        \helpdrawMGColorReset

        \coordinate (circuitTheX) at (cxforpainteddiag);
        \helpdrawRSFdiagonal{3}{3}
        \draw[absorbalgarrow] ($(rsfnode2) +(0.15,0)$) -- ++ (0.5,0);
        \draw[absorbalgarrow] ($(rsfnode1) +(0.15,0)$) --
            ++ ($(0.5,0) + (\flatgateadvance+\flatgatesize,0)$);
       
    \end{scope}
\end{scope}

\begin{scope}[xshift=14.5cm, yshift=0.45cm]
    \begin{scope}[rotate = 90]
        \circuitInitXAt{0.25}
        \helpdrawInnerProdSevenlines{2.0}

        \helpdrawMGColorToFaint
        \helpdrawRSFdiagonal{4}{2}
        \coordinate (cxforpainteddiag) at (rsfEndX);
        \helpdrawRSFdiagonal{0}{1}
        \helpdrawRSFdiagonal{2}{1}
        \foreach \x in {1,...,5} {\circuitAdvanceXBy{\flatgateadvance}}
        \helpdrawRSFdiagonal{1}{1}
        \helpdrawRSFdiagonal{3}{1}
        \helpdrawMGColorReset

        \coordinate (circuitTheX) at (cxforpainteddiag);
        \helpdrawRSFdiagonal{5}{0}
        \draw[absorbalgarrow] ($(rsfnode2) -(0.15,0)$) --
            ++ (-0.5,0);
    
    \end{scope}
\end{scope}

\begin{scope}[xshift=19cm, yshift=1.05cm]
    \begin{scope}[rotate = 90]
        \circuitInitXAt{0.25}
        
        \helpdrawInnerProdSevenlines{0.8}

        \helpdrawRSFdiagonal{0}{2}
        \helpdrawRSFdiagonal{2}{2}
        \helpdrawRSFdiagonal{4}{2}
        
    \end{scope}
\end{scope}
}

\newcommand{\helpdrawtdopedMGC}{

\begin{scope}[xshift=1cm]
    \begin{scope}[rotate = 90]
        \circuitInitXAt{0.25}
        \helpdrawCircuitLinesWithoutQubits{0}{5}{0}{5.3}

        \helpdrawRSFdiagonal{4}{1}
        \circuitAdvanceXBy{\flatgateadvance}
        \helpdrawRSFdiagonal{3}{2}
        \circuitAdvanceXBy{\flatgateadvance}
        \helpdrawRSFdiagonal{2}{3}
        \circuitAdvanceXBy{\flatgateadvance}
        \helpdrawRSFdiagonal{1}{4}
        \circuitAdvanceXBy{\flatgateadvance}
        \helpdrawRSFdiagonal{0}{5}
        \circuitAdvanceXBy{2*\flatgateadvance}
        \circuitMultiGate[\flatgatesize][othergatecolorbg][othergatecolorfg]{}{0}{-1}{}
        \circuitAdvanceXBy{\flatgateadvance}
        \helpdrawRSFdiagonal{1}{4}
        \circuitAdvanceXBy{\flatgateadvance}
        \helpdrawRSFdiagonal{0}{5}
        \circuitAdvanceXBy{2*\flatgateadvance}
        \circuitMultiGate[\flatgatesize][othergatecolorbg][othergatecolorfg]{}{0}{-1}{}
        \circuitAdvanceXBy{\flatgateadvance}
        \helpdrawRSFdiagonal{1}{4}
        \circuitAdvanceXBy{\flatgateadvance}
        \helpdrawRSFdiagonal{0}{5}
        
    \end{scope}
\end{scope}

}

\newcommand{\helpdrawtdopedFGS}{

\begin{scope}[xshift=1cm]
    \begin{scope}[rotate = 90]
        \circuitInitXAt{0.25}
        \helpdrawCircuitLines{0}{5}{0}{4.1}

        \helpdrawRSFdiagonal{4}{1}
        \helpdrawRSFdiagonal{2}{3}
        \helpdrawRSFdiagonal{0}{5}
        \circuitAdvanceXBy{2*\flatgateadvance}
        \circuitMultiGate[\flatgatesize][othergatecolorbg][othergatecolorfg]{}{0}{-1}{}
        \circuitAdvanceXBy{\flatgateadvance}
        \helpdrawRSFdiagonal{1}{4}
        \circuitAdvanceXBy{\flatgateadvance}
        \helpdrawRSFdiagonal{0}{5}
        \circuitAdvanceXBy{2*\flatgateadvance}
        \circuitMultiGate[\flatgatesize][othergatecolorbg][othergatecolorfg]{}{0}{-1}{}
        \circuitAdvanceXBy{\flatgateadvance}
        \helpdrawRSFdiagonal{1}{4}
        \circuitAdvanceXBy{\flatgateadvance}
        \helpdrawRSFdiagonal{0}{5}
        
    \end{scope}
\end{scope}

}

\newcommand{\helpdrawtdopedInner}{

\begin{scope}[xshift=1cm]
    \begin{scope}[rotate = 90]
        \circuitInitXAt{0.25}
        \helpdrawCircuitLines{0}{5}{0}{2.9}
        \foreach\p in {0,...,5}{\draw[fill] ($0.5*(2*2.9, -\p)$) circle[radius=\qubitBallRadius];}
        \helpdrawRSFdiagonal{4}{1}
        \helpdrawRSFdiagonal{2}{3}
        \helpdrawRSFdiagonal{0}{5}
        \circuitAdvanceXBy{2*\flatgateadvance}
        \circuitMultiGate[\flatgatesize][othergatecolorbg][othergatecolorfg]{}{0}{-1}{}
        \circuitAdvanceXBy{\flatgateadvance}
        \helpdrawRSFdiagonal{1}{4}
        \circuitAdvanceXBy{\flatgateadvance}
        \helpdrawRSFdiagonal{0}{5}
        \circuitAdvanceXBy{2*\flatgateadvance}
        \circuitMultiGate[\flatgatesize][othergatecolorbg][othergatecolorfg]{}{0}{-1}{}
        \circuitAdvanceXBy{\flatgateadvance}
        \helpdrawRSFdiagonal{1}{2}
        \circuitAdvanceXBy{\flatgateadvance}
        \helpdrawRSFdiagonal{0}{1}
        
    \end{scope}
\end{scope}

}

\newcommand{\helpdrawtdoped}{

\node at (0.3,8) {(a)};
    \node at (5.2,8) {(b)};
    \node at (5.2,3) {(c)};

    \node[anchor=west] at (2.,1.45) {Matchgate};
    \node[anchor=west] at (2.,0.7) {Resourceful};
    \node[anchor=west] at (2.,0.25) {gate};
    
    \begin{scope}[xshift = 0.9cm]
    \begin{scope}[rotate = 90]
    \circuitInitXAt{1.5}
    \helpdrawCircuitLinesWithoutQubits{0}{1}{1.25}{1.75}
    \circuitMultiGate[\flatgatesize][mgcolorbg][mgcolorfg]{}{0}{-1}{}
    \end{scope}
    \end{scope}
    
    \begin{scope}[xshift = 0.9cm]
    \begin{scope}[rotate = 90]
    \circuitInitXAt{0.5}
    \helpdrawCircuitLinesWithoutQubits{0}{1}{0.25}{0.75}
    \circuitMultiGate[\flatgatesize][othergatecolorbg][othergatecolorfg]{}{0}{-1}{}
    \end{scope}
    \end{scope}
    
        \begin{scope}[xshift=0cm, yshift=2.65cm]
        \helpdrawtdopedMGC
        \end{scope}
        \begin{scope}[xshift=5cm, yshift=3.8cm]
        \helpdrawtdopedFGS
        \end{scope}
        \begin{scope}[xshift=5cm]
        \helpdrawtdopedInner
        \end{scope}

}

\newcommand{\helpdrawMainfigureSED}{
\node at (3,0.8) {$=$};
\node at (6.5,0.8) {$=$};
\node at (3,-0.8) {$=$};
\node at (6.5,-0.8) {$=$};
\begin{scope}
    \node[anchor=south] at (1.25, 1.7) {\introfiguretextsize$\ket\psi$};
    \begin{scope}[rotate = 90]
        \circuitInitXAt{0.25}
        \helpdrawEntangled{0}{4}{0}
        \helpdrawCircuitLines{0}{5}{0}{1.7}
    \end{scope}
    \begin{scope}[yshift=-0.8cm, xshift=1.25cm, scale=0.5]
        \draw[draw = none, fill=black!40] (-1,-1) rectangle (1,1);
        \draw[draw = black, line width=0.6pt] (-1.04,-1.04) to[out=100,in=260] node[left] {\introfiguretextsize $\Gamma = $}(-1.04,1.04);
        \draw[draw = black, line width=0.6pt] (1.04,-1.04) to[out=80,in=280] (1.04,1.04);
    \end{scope}
\end{scope}
\begin{scope}[xshift=3.5cm]
\node[anchor=south] at (1.25, 1.7) {\introfiguretextsize$D_1 \ket{00}\ket{\psi'}$};
    \begin{scope}[rotate=90]
        \circuitInitXAt{0.25}
        \helpdrawEntangled{2}{4}{0}
        \helpdrawCircuitLines{0}{5}{0}{1.7}
        \helpdrawRSFdiagonal{0}{5}
    \end{scope}
    \begin{scope}[yshift=-0.8cm, xshift=1.25cm, scale=0.5]
        \draw[draw = none, fill=black!40] (-1,1) rectangle (-0.666,0.666);
        \draw[draw = none, fill=black!40] (-0.666,0.666) rectangle (-0.333,0.333);
        \draw[draw = none, fill=black!40] (-0.333,0.333) rectangle (1,-1);
        \draw[draw = black, line width=0.6pt] (-1.04,-1.04) to[out=100,in=260] node[left] {\introfiguretextsize $R_1$}(-1.04,1.04);
        \draw[draw = black, line width=0.6pt] (1.04,-1.04) to[out=80,in=280] node[right] {\introfiguretextsize $R_1^\transpose$}(1.04,1.04);
    \end{scope}
\end{scope}

\begin{scope}[xshift=7cm]
\node[anchor=south] at (1.25, 1.7) {\introfiguretextsize$\prod_i D_i\ket{0^n}$};
    \begin{scope}[rotate=90]
        \circuitInitXAt{0.25}
        \helpdrawCircuitLines{0}{5}{0}{1.7}
        \helpdrawRSFdiagonal{0}{5}
        \helpdrawRSFdiagonal{2}{3}
        \helpdrawRSFdiagonal{4}{1}
    \end{scope}
    \begin{scope}[yshift=-0.8cm, xshift=1.25cm, scale=0.5]
        \draw[draw = none, fill=black!40] (-1,1) rectangle (-0.666,0.666);
        \draw[draw = none, fill=black!40] (-0.666,0.666) rectangle (-0.333,0.333);
        \draw[draw = none, fill=black!40] (-0.333,0.333) rectangle (0,0);
        \draw[draw = none, fill=black!40] (0,0) rectangle (0.333,-0.333);
        \draw[draw = none, fill=black!40] (0.333,-0.333) rectangle (0.666,-0.666);
        \draw[draw = none, fill=black!40] (0.666,-0.666) rectangle (1,-1);
        \draw[draw = black, line width=0.6pt] (-1.04,-1.04) to[out=100,in=260] node[left] {\introfiguretextsize $\prod_i\! R_i\! \!\:\!$}(-1.04,1.04);
        \draw[draw = black, line width=0.6pt] (1.04,-1.04) to[out=80,in=280] node[right] {\introfiguretextsize $\!\:\!\!\:\!\!\Big(\!\!\prod_i\! R_i\!\Big)^{\!\:\!\!\!\transpose}$} (1.04,1.04);
    \end{scope}
\end{scope}
}

\newcommand{\helpdrawMainfigureCircuitDepth}{
\node at (4.5,0.8) {$\Leftrightarrow$};
\begin{scope}[yshift=1.5cm, xshift=0.25cm]
    \node[anchor=south] at (1.75, 0.2) {\introfiguretextsize$\ket\psi$};
    \begin{scope}[rotate = 90]
        \circuitInitXAt{0.25}
        \helpdrawEntangled{0}{6}{0}
        \helpdrawCircuitLines{0}{7}{0}{0.2}
    \end{scope}
    \begin{scope}[yshift=-0.9cm, xshift=2.25cm, scale=0.7]
        \draw[draw = none, fill=black!40] (-1,1) -- (-0.2,1) -- (1,-0.2)-- (1,-1) -- (0.2,-1) -- (-1,0.2) -- cycle;
        \draw[dotted, black] (-1,1)--(1,-1);
        
        \draw[line width=0.6pt, decorate,decoration={brace,amplitude=3pt,mirror,raise=0.5pt}] (0,0) -- node[midway, anchor = north west, yshift=2pt, xshift=-1pt] {\introfiguretextsize $d$} (0.4,0.4);
        
        \draw[draw = black, line width=0.6pt] (-1.04,-1.04) to[out=100,in=260] node[left] {\introfiguretextsize $\Gamma = $}(-1.04,1.04);
        \draw[draw = black, line width=0.6pt] (1.04,-1.04) to[out=80,in=280] (1.04,1.04);
    \end{scope}
\end{scope}
\begin{scope}[xshift=5.25cm, yshift=0.05cm]
    \node[anchor=south] at (1.75, 1.65) {\introfiguretextsize$\ket\psi = U_\text{MG} \ket{0^n}$};
    \begin{scope}[rotate=90]
        \circuitInitXAt{0.25}
        \helpdrawCircuitLines{0}{7}{0}{0.8}
            \foreach\p in {0,...,7}{
            \draw[dotted, line width=1pt] ($(0.74, -0.5*\p)$) -- ++(0.5,0);
            \circuitQubitLine{1.25}{1.65}{-\p}}
        \helpdrawRSFdiagonal{0}{2} 
        \helpdrawRSFdiagonal{2}{2} 
        \helpdrawRSFdiagonal{4}{2}
        \helpdrawRSFdiagonal{6}{1}
        \circuitAdvanceXBy{\flatgateadvance}
        \circuitAdvanceXBy{\flatgateadvance}
        \circuitAdvanceXBy{\flatgateadvance}
        \circuitAdvanceXBy{\flatgateadvance}
        \helpdrawRSFdiagonal{0}{1}
        \helpdrawRSFdiagonal{2}{1} 
        \helpdrawRSFdiagonal{4}{1}
        \helpdrawRSFdiagonal{6}{1}
    \end{scope}
    \draw[line width=0.6pt, decorate,decoration={brace,amplitude=4pt,mirror,raise=0.2cm}] (3.45,0.1) -- node[midway, anchor = west, xshift=8pt] {\introfiguretextsize $\mathcal{O}(d)$} (3.45,1.6);
\end{scope}
}

\newcommand{\helpdrawMainfigureEntanglementCutting}{
\node at (4.5,1.2) {$=$};
\begin{scope}[yshift=1.cm, xshift=0.25cm]
    \node[anchor=south] at (1.75, 0.5) {\introfiguretextsize$\ket\psi$};
    \begin{scope}[rotate = 90]
        \circuitInitXAt{0.25}
        \helpdrawEntangled{0}{6}{0}
        \helpdrawCircuitLines{0}{7}{0}{0.5}
    \end{scope}
\end{scope}
\begin{scope}[xshift=5.25cm, yshift=0.6cm]
    \node[anchor=south] at (1.75, 1.1) {\introfiguretextsize$\ket\psi = U_\text{MG} \ket{\psi_\partyA}\ket{\psi_\partyB}$};
    \begin{scope}[rotate=90]
        \circuitInitXAt{0.55}
        \helpdrawEntangled{0}{2}{0}
        \helpdrawEntangled{4}{6}{0}
        \helpdrawCircuitLines{0}{7}{0}{1.1}
        \circuitMultiGate[0.6][mgcolorlightbg][mgcolorfg]{}{-2}{-5}{\introfiguretextsize \color{black} $U_\text{MG}$}
    \end{scope}
    \draw[line width=0.6pt, decorate,decoration={brace,amplitude=4pt,mirror,raise=0.2cm}] (0.9,0.1) -- node[midway, anchor = north, yshift=-0.3cm] {\introfiguretextsize $\mathcal{O}(d)$} (2.6,0.1);
\end{scope}
}

\newcommand{\helpdrawMainfigureImportedTools}{
\begin{scope}[xshift=0.5cm, yshift=1cm]

    \begin{scope}[yshift=0cm]
    \node[align=left, anchor = north west] (N1) at (-0.2,1.9){\introfiguretextsize GYB relation${}^\text{\tiny\cite{CaKo22}}$};
    \begin{scope}[yshift=0.2cm]
    \node at (1.5,0.45){$=$};
    
        \begin{scope}[rotate=90]
            \circuitInitXAt{0.2}
            \foreach\p in {0,...,2}{\circuitQubitLine{0}{1}{-\p}}
            \helpdrawRSFdiagonal{0}{2}
            \coordinate (circuitTheX) at (rsfEndX);
            \helpdrawRSFdiagonal{0}{1}
        \end{scope}
    
        \begin{scope}[xshift=2cm, rotate=90]
            \circuitInitXAt{0.2}
            \foreach\p in {0,...,2}{\circuitQubitLine{0}{1}{-\p}}
            \helpdrawRSFdiagonal{1}{0}
            \coordinate (circuitTheX) at (rsfEndX);
            \helpdrawRSFdiagonal{1}{1}
        \end{scope}
    \end{scope}
    \end{scope}

    \begin{scope}[yshift=-1.75cm]
    \node[align=left, anchor = north west] (N1) at (-0.2,1.6){\introfiguretextsize LR relation${}^\text{\tiny\cite{MoLa25}}$};
    \begin{scope}[yshift=0.15cm]
    \node at (1.5,0.35){$=$};
    
        \begin{scope}[rotate=90]
            \circuitInitXAt{0.25}
            \helpdrawCircuitLines{0}{2}{0}{0.75}
            \helpdrawRSFdiagonal{0}{2}
        \end{scope}
    
        \begin{scope}[xshift=2cm, rotate=90]
            \circuitInitXAt{0.25}
            \helpdrawCircuitLines{0}{2}{0}{0.75}
            \helpdrawRSFdiagonal{1}{0}
        \end{scope}
    \end{scope}
    \end{scope}

    \begin{scope}[yshift=-4.2cm]
    \node[anchor = north west] (N1) at (-0.2,2.2) {\introfiguretextsize Absorption alg.${}^\text{\tiny\cite{MoLa25}}$};
        \begin{scope}[xshift=-5.3cm, yshift=0.1cm]
            \node[anchor=west] (N1) at (6.75,1.1){\introfiguretextsize $U_\textbf{MG}\ket{0^n}$};
            \node[anchor=west] (N2) at (6.75,-1){\introfiguretextsize $U_\textbf{RSF}\!\ket{0^n}$};
            \draw[-latex, line width=1pt] (N1.south) -- (N2.north);
            
            \begin{scope}[xshift=5.3cm, rotate=90]
                \helpdrawCircuitLines{0}{3}{0}{1.4}
                \circuitInitXAt{0.25}
                \helpdrawRSFdiagonal{0}{3}
                \helpdrawRSFdiagonal{2}{1}
                \circuitAdvanceXBy{\flatgateadvance}
                \circuitAdvanceXBy{\flatgateadvance}
                \helpdrawRSFdiagonal{0}{2}
            \end{scope}
        
            \begin{scope}[xshift=5.3cm, yshift=-1.5cm, rotate=90]
                \helpdrawCircuitLines{0}{3}{0}{1.1}
                \circuitInitXAt{0.25}
                \helpdrawRSFdiagonal{0}{3}
                \helpdrawRSFdiagonal{2}{1}
            \end{scope}
        \end{scope}
    \end{scope}

\end{scope}
}

\newcommand{\helpdrawMainfigureOptimalVtwo}{
\node at (3.1,0.8) {$\neq^{*}$};

\begin{scope}[yshift=0.3cm]
\node[anchor=south] at (1.25, 1.05) {\introfiguretextsize $U_\text{RSF} \ket{0^n}$};
    \begin{scope}[rotate=90]
        \circuitInitXAt{0.25}
        \helpdrawCircuitLines{0}{5}{0}{1.1}
        \helpdrawRSFdiagonal{0}{3}
        \helpdrawRSFdiagonal{2}{3}
        \helpdrawRSFdiagonal{4}{1}
    \end{scope}
\node[anchor=north] at (1.25,0) {\introfiguretextsize $k$ matchgates};
\end{scope}

\begin{scope}[xshift=3.5cm, yshift=0.3cm]
\node[anchor=south] at (1.25, 1.05) {\introfiguretextsize any $U'_\text{MG} \ket{0^n}$};
    \begin{scope}[rotate = 90]
        \circuitInitXAt{0.25}
        \helpdrawCircuitLines{0}{5}{0}{1.1}
        \helpdrawRSFdiagonal{1}{2}
        \helpdrawRSFdiagonal{3}{2}
        \circuitAdvanceXBy{\flatgateadvance}
        \helpdrawRSFdiagonal{0}{2}
    \end{scope}
\node[anchor=north] at (1.25,0) {\introfiguretextsize $<k$ matchgates};
\end{scope}
}

\newcommand{\helpdrawMainfigureInnerProduct}{
\node at (3.5,0.9) {$=$};
\begin{scope}[xshift=1cm]
    \begin{scope}[rotate = 90]
        \circuitInitXAt{0.25}
        \helpdrawCircuitLines{0}{4}{0}{1.9}
        \foreach\p in {0,...,4}{\draw[fill] ($0.5*(2*1.9, -\p)$) circle[radius=\qubitBallRadius];}
        \helpdrawRSFdiagonal{0}{3}
        \helpdrawRSFdiagonal{2}{2}
        \circuitAdvanceXBy{\flatgateadvance}
        \circuitAdvanceXBy{\flatgateadvance}
        \helpdrawRSFdiagonal{0}{1}
        \circuitAdvanceXBy{\flatgateadvance}
        \circuitAdvanceXBy{0.2cm}
        \helpdrawRSFdiagonal{1}{0}
        \helpdrawRSFdiagonal{3}{0}
        \circuitAdvanceXBy{\flatgateadvance}
    \end{scope}
    \draw[line width=0.6pt, decorate,decoration={brace,amplitude=4pt,mirror,raise=0.2cm}] (0,1) -- node[midway, anchor = east, align=right, xshift=-8pt] {\introfiguretextsize $\ket{\psi}$} (0,-0.1);
    \draw[line width=0.6pt, decorate,decoration={brace,amplitude=4pt,mirror,raise=0.2cm}] (0,2.0) -- node[midway, anchor = east, align=right, xshift=-8pt] {\introfiguretextsize $\bra{\phi}$} (0,1.2);
\end{scope}

\begin{scope}[xshift=4cm, yshift=0.2]
    \node[] at (1, 1.7) {\introfiguretextsize $\braket{\phi}{\psi}$};
    \begin{scope}[rotate=90]
        \circuitInitXAt{0.7}
        \helpdrawCircuitLines{0}{4}{0.4}{1.3}
        \foreach\p in {0,...,4}{\draw[fill] ($0.5*(2*1.3, -\p)$) circle[radius=\qubitBallRadius];}
        \helpdrawRSFdiagonal{0}{2}
        \helpdrawRSFdiagonal{2}{2}
    \end{scope}
\end{scope}
}

\newcommand{\helpdrawMainfigureSWAPStandardForm}{

\begin{scope}[scale = 0.6, yshift=3.1cm]
    \def\circuitlinewidth{0.7pt}
    \def\circuitRoundedCorners{0.9pt}

    \node[align=center] (N1) at (2.5,0.25) {\introfiguretextsize $U_{t\text{-doped}}$};
    \node[align=center] (N2) at (9.05,0.25) {\introfiguretextsize $U_{t\text{-doped}}\ket{0^n}$};
    \node[align=center] (N3) at (14.25,0.25) {\introfiguretextsize $\bra{0^n}U_{t\text{-doped}}\ket{0^n}$};

    \draw[-latex, line width=1pt] ($(N1.south) + (0,0.15)$) -- ++(0,-0.65);
    \draw[-latex, line width=1pt] ($(N2.south) + (0,0.15)$) -- ++(0,-0.65);
    \draw[-latex, line width=1pt] ($(N3.south) + (0,0.15)$) -- ++(0,-0.65);

    \draw[line width=0.6pt, decorate,decoration={brace,amplitude=4pt,mirror,raise=0.2cm}] (0.1,-3.2) -- node[midway, anchor = north, align=center, yshift=-0.25cm] {\introfiguretextsize $\mathcal{O}(n+t)$} (5.2,-3.2);
    
    \draw[line width=0.6pt, decorate,decoration={brace,amplitude=4pt,mirror,raise=0.2cm}] (7,-3.2) -- node[midway, anchor = north, align=center, yshift=-0.25cm] {\introfiguretextsize $\mathcal{O}(n+t)$} (10.9,-3.2);

    \draw[line width=0.6pt, decorate,decoration={brace,amplitude=4pt,mirror,raise=0.2cm}] (12.9,-3.2) -- node[midway, anchor = north, align=center, yshift=-0.25cm] {\introfiguretextsize $\mathcal{O}(t)$} (15.6,-3.2);
    
    \begin{scope}[xshift=0cm, yshift=0cm, rotate=-90]
    \helpdrawtdopedMGC
    \end{scope}
    \begin{scope}[xshift=6.9cm, yshift=0cm, rotate=-90]
    \helpdrawtdopedFGS
    \end{scope}
    \begin{scope}[xshift=12.8cm, yshift=0cm, rotate=-90]
    \helpdrawtdopedInner
    \end{scope}
    
\end{scope}
}

\newcommand\helpdrawBox[5]{
    \draw[draw=#3, line width=1.5pt, fill=#4, rounded corners=8pt,
            path picture = 
                {\node[anchor=north west, align=left, xshift=0pt, yshift=-1pt] at (path picture bounding box.north west)
                    {#5};} ]
        #1 rectangle #2;
}

\newcommand{\helpdrawSEDfigureMerging}{
    \def\circuitGateHeight{\circuitlinespacing*0.4}
    \def\flatgateadvance{0.225}
    \def\biggeradvance{0.4}
    \colorlet{xxcolfg}{mgcolorfg}
    \colorlet{xxcolbg}{mgcolorbg}
    \colorlet{zcolbg}{othergatecolorbg}
    \colorlet{zcolfg}{othergatecolorfg}
    \begin{scope}[rotate=90]
        \helpdrawCircuitLines{0}{3}{0}{2.35}
        \circuitInitXAt{0.25}
        \circuitMultiGate[\flatgatesize][xxcolbg][xxcolfg]{}{-1}{-2}{}
        \circuitAdvanceXBy{\flatgateadvance}
        \circuitSingleGate[\flatgatesize][zcolbg][zcolfg]{}{-2}{}
        \circuitAdvanceXBy{\flatgateadvance}
        \circuitMultiGate[\flatgatesize][xxcolbg][xxcolfg]{}{-2}{-3}{}
        \circuitAdvanceXBy{\flatgateadvance}
        \circuitSingleGate[\flatgatesize][zcolbg][zcolfg]{}{-3}{}
        \circuitInitXAt{0.25}
        \circuitAdvanceXBy{\flatgateadvance}
        \circuitAdvanceXBy{\biggeradvance}
    
        \circuitMultiGate[\flatgatesize][xxcolbg][xxcolfg]{}{-0}{-1}{}
        \circuitAdvanceXBy{\flatgateadvance}
        \circuitSingleGate[\flatgatesize][zcolbg][zcolfg]{}{-1}{}
        \circuitAdvanceXBy{\flatgateadvance}
        \circuitMultiGate[\flatgatesize][xxcolbg][xxcolfg]{}{-1}{-2}{}
        \circuitAdvanceXBy{\flatgateadvance}
        \circuitSingleGate[\flatgatesize][zcolbg][zcolfg]{}{-2}{}
        \circuitAdvanceXBy{\flatgateadvance}
        \circuitMultiGate[\flatgatesize][xxcolbg][xxcolfg]{}{-2}{-3}{}
        \circuitAdvanceXBy{\flatgateadvance}
        \circuitSingleGate[\flatgatesize][zcolbg][zcolfg]{}{-3}{}
    \end{scope}
    
    \node at (2.25,1.15) {$=$};
    
    \begin{scope}[xshift=3cm, rotate=90]
        \circuitInitXAt{0.25}
        \helpdrawCircuitLines{0}{3}{0}{2.35}
        \circuitMultiGate[\flatgatesize][xxcolbg][xxcolfg]{}{-0}{-1}{}
        \circuitAdvanceXBy{\biggeradvance}
        \circuitMultiGate[\flatgatesize][xxcolbg][xxcolfg]{}{-1}{-2}{}
        \circuitAdvanceXBy{\flatgateadvance}
        \circuitSingleGate[\flatgatesize][zcolbg][zcolfg]{}{-2}{}
        \circuitSingleGate[\flatgatesize][zcolbg][zcolfg]{}{-1}{}
        \circuitAdvanceXBy{\flatgateadvance}
        \circuitMultiGate[\flatgatesize][xxcolbg][xxcolfg]{}{-1}{-2}{}
        \circuitAdvanceXBy{\biggeradvance}
        \circuitMultiGate[\flatgatesize][xxcolbg][xxcolfg]{}{-2}{-3}{}
        \circuitAdvanceXBy{\flatgateadvance}
        \circuitSingleGate[\flatgatesize][zcolbg][zcolfg]{}{-3}{}
        \circuitSingleGate[\flatgatesize][zcolbg][zcolfg]{}{-2}{}
        \circuitAdvanceXBy{\flatgateadvance}
        \circuitMultiGate[\flatgatesize][xxcolbg][xxcolfg]{}{-2}{-3}{}
        \circuitAdvanceXBy{\flatgateadvance}
        \circuitSingleGate[\flatgatesize][zcolbg][zcolfg]{}{-3}{}
    \end{scope}
    
    \begin{scope}[xshift=5.5cm, yshift=1.3cm, rotate=90]
        \circuitInitXAt{0}
        \circuitQubitLine{-0.2}{0.2}{0}
        \circuitQubitLine{-0.2}{0.2}{-1}
        \circuitMultiGate[\flatgatesize][xxcolbg][xxcolfg]{}{-0}{-1}{}
        \circuitAdvanceXBy{0.7}
        \circuitQubitLine{0.5}{0.9}{-1}
        \circuitSingleGate[\flatgatesize][zcolbg][zcolfg]{}{-1}{}
        \node[anchor=west] at (0.7,-0.6) {\footnotesize $\exp(\ii \alpha_i Z)$};
        \node[anchor=west] at (0,-0.6) {\footnotesize $\exp(\ii \beta_i X\!\otimes\! X)$};
    \end{scope}
}

\newcommand{\helpdrawSEDfigureAlgAction}{

    \newcommand\tsize\footnotesize

    \node (N1) at (1.25, -2.5) {$\cdots$};
    \draw[-latex, line width=1pt] (0.5, -2.5) -- (N1.west);
    \draw[-latex, line width=1pt] (N1.east) -- (2, -2.5);

    \draw[-latex, line width=1pt] (4.1,-2.5) -- node[above]{\tsize $e^{\ii \beta\:\!\! X_{\:\!\!1} \:\!\! X_{\:\!\!2}}$} (4.9,-2.5);

    \newcommand{\helpdrawCM}{
        \foreach \c in {0,...,7}{\draw[xshift=-1cm, yshift=-1cm, scale = 0.25, draw = none, fill=white] ($(\c - 0.03,7-\c - 0.03)$) rectangle ($(\c+1 + 0.03 ,7-\c+1 + 0.03)$);}
    
        \draw[draw = black, line width=0.6pt] (-1.04,-1.04) to[out=100,in=260] (-1.04,1.04);
        \draw[draw = black, line width=0.6pt] (1.04,-1.04) to[out=80,in=280] (1.04,1.04);
    }

    \begin{scope}[yshift=-0.5cm, xshift=1.25cm]
        \node[anchor = east, align = right] at (-1,0) {\tsize $\Gamma = $};
        \draw[-latex, line width=1pt] (1.1,0) -- node[above]{\tsize $e^{\ii \alpha \:\!\! Z_{\:\!\!4}}$} (1.9,0);
        \draw[-latex, line width=1pt] (4.1,0) -- node[above]{\tsize $e^{\ii \beta\:\!\! X_{\:\!\!3} \:\!\! X_{\:\!\!4}}$} (4.9,0);

        \begin{scope}[yshift=0cm, xshift=0cm, scale=0.8]
            \draw[xshift=-1cm, yshift=-1cm, scale = 0.25, draw = none, fill=black!40] (0,0) rectangle (8,8);
            \helpdrawCM
        \end{scope}
        
        \begin{scope}[yshift=0cm, xshift=3cm, scale=0.8]
            \draw[xshift=-1cm, yshift=-1cm, scale = 0.25, draw = none, fill=black!40] (0,8) --  (7,8) -- (7, 7) -- (8,7) -- (8,0) -- (1,0) -- (1,1) -- (0,1) -- cycle;
            \helpdrawCM
        \end{scope}

    \end{scope}
    
    \begin{scope}[yshift=-2.5cm, xshift=-0.5cm, scale=0.8]
        \draw[xshift=-1cm, yshift=-1cm, scale = 0.25, draw = none, fill=black!40] (0,8) --  (6,8) -- (6, 7) -- (8,7) -- (8,0) -- (1,0) -- (1,2) -- (0,2) -- cycle;
        \helpdrawCM
    \end{scope}

    \begin{scope}[yshift=-2.5cm, xshift=3cm, scale=0.8]
        \draw[xshift=-1cm, yshift=-1cm, scale = 0.25, draw = none, fill=black!40] (0,8) --  (3,8) -- (3, 7) -- (8,7) -- (8,0) -- (1,0) -- (1,5) -- (0,5) -- cycle;
        \helpdrawCM
    \end{scope}
    \begin{scope}[yshift=-2.5cm, xshift=6cm, scale=0.8]
        \draw[xshift=-1cm, yshift=-1cm, scale = 0.25, draw = none, fill=black!40] (0,8) --  (2,8) -- (2, 6) -- (8,6) -- (8,0) -- (2,0) -- (2,6) -- (0,6) -- cycle;
        \helpdrawCM
    \end{scope}
}

\newcommand{\helpdrawCuttingFigureTheorem}{
\node at (6.5,1.2) {$=$};
\begin{scope}[yshift=1.cm, xshift=0.25cm]
    \node[anchor=south] at (2.75, 0.5) {\introfiguretextsize$\ket\psi$ with $\bdns$-banded CM};
    \begin{scope}[rotate = 90]
        \circuitInitXAt{0.25}
        \helpdrawEntangled{0}{10}{0}
        \helpdrawCircuitLines{0}{11}{0}{0.5}
    \end{scope}
\end{scope}
\begin{scope}[xshift=7.25cm, yshift=0.6cm]
    
    \draw[line width=\circuitlinewidth] (1, 0) to[out=-30,in=210] (2.5,0);
    \draw[line width=\circuitlinewidth] (1.5, 0) to[out=-30,in=210] (2,0);
    \draw[line width=\circuitlinewidth] (3.5, 0) to[out=-30,in=210] (5,0);
    \draw[line width=\circuitlinewidth] (4, 0) to[out=-30,in=210] (4.5,0);
    \begin{scope}[rotate=90]
        \circuitInitXAt{0.55}
        \helpdrawCircuitLines{0}{11}{0}{1.1}
        \circuitMultiGate[0.6][mgcolorlightbg][mgcolorfg]{}{-0}{-3}{\introfiguretextsize \color{black} $U_\partyA$}
        \circuitMultiGate[0.6][mgcolorlightbg][mgcolorfg]{}{-4}{-8}{\introfiguretextsize \color{black} $U_\partyB$}
        \circuitMultiGate[0.6][mgcolorlightbg][mgcolorfg]{}{-9}{-11}{\introfiguretextsize \color{black} $U_\partyC$}
    \end{scope}
    \draw[line width=0.6pt, decorate,decoration={brace,amplitude=4pt,raise=0.2cm}] (1.9,0.8) -- node[midway, anchor = south, yshift=0.3cm] {\introfiguretextsize $\mathcal{O}(\bdns)$} (4.1,0.8);
\end{scope}
}

\newcommand{\helpdrawCuttingFigureAlg}{
\node at (4.75,1.0) {$=$};
\node at (9.75,1.0) {$=$};
\begin{scope}[yshift=0.8cm, xshift=0.25cm]
    \draw[line width=1.5*\circuitlinewidth, dotted] (0,0) -- (0.5,0);
    \draw[line width=1.5*\circuitlinewidth, dotted] (3.5,0) -- (4,0);
    \begin{scope}[rotate = 90]
        \circuitInitXAt{0.25}
        \helpdrawEntangled{1}{6}{0}
        \helpdrawCircuitLines{1}{7}{0}{0.5}
    \end{scope}
\end{scope}
\begin{scope}[yshift=.5cm, xshift=5.25cm]
    \draw[line width=1.5*\circuitlinewidth, dotted] (0,0) -- (0.5,0);
    \draw[line width=1.5*\circuitlinewidth, dotted] (3.5,0) -- (4,0);
    \draw[line width=\circuitlinewidth] (0.7, -0.5) to[out=0,in=240] (1,0);
    \draw[line width=\circuitlinewidth] (0.7, -0.25) to[out=0,in=210] (2.5,0);
    \draw[line width=\circuitlinewidth] (1.5, 0) to[out=-30,in=180] (3.3,-0.5);
    \draw[line width=\circuitlinewidth] (2, 0) to[out=-40,in=180] (3.3,-0.25);
    \draw[line width = 4pt, line cap = round] (0.5,0) -- (0.3,0);
    \draw[line width = 4pt, line cap = round] (3.5,0) -- (3.7,0);
    
    \begin{scope}[rotate = 90]
        \circuitInitXAt{0.55}
        \helpdrawCircuitLines{1}{7}{0}{1.1}
        \circuitMultiGate[0.6][mgcolorlightbg][mgcolorfg]{}{-2}{-6}{\introfiguretextsize \color{black} $U_\text{diag}$}
    \end{scope}
\end{scope}
\begin{scope}[yshift=.cm, xshift=10.25cm]
    \draw[line width=1.5*\circuitlinewidth, dotted] (0,0) -- (0.5,0);
    \draw[line width=1.5*\circuitlinewidth, dotted] (3.5,0) -- (4,0);
    \draw[line width=\circuitlinewidth] (0.72, -0.3) to[out=10,in=220] (1.5,0);
    \draw[line width=\circuitlinewidth] (0.72, -0.15) to[out=0,in=210] (1,0);
    \draw[line width=\circuitlinewidth] (2.5, 0) to[out=-30,in=180] (3.28,-0.3);
    \draw[line width=\circuitlinewidth] (3, 0) to[out=-40,in=170] (3.28,-0.15);
    \draw[line width = 4pt, line cap = round] (0.5,0) -- (0.3,0);
    \draw[line width = 4pt, line cap = round] (3.5,0) -- (3.7,0);
    
    \begin{scope}[rotate = 90]
        \circuitInitXAt{0.55}
        \helpdrawCircuitLines{1}{7}{0}{1.9}
        \circuitMultiGate[0.6][mgcolorlightbg][mgcolorfg]{}{-2}{-6}{\introfiguretextsize \color{black} $U_\text{perm}$}
        \circuitAdvanceXBy{0.8}
        \circuitMultiGate[0.6][mgcolorlightbg][mgcolorfg]{}{-2}{-6}{\introfiguretextsize \color{black} $U_\text{diag}$}
    \end{scope}
\end{scope}
}

\ExplSyntaxOn
\NewDocumentCommand{\concat}{m}
 {
  \concat_process:n { #1 }
 }

\cs_new_protected:Nn \concat_process:n
 {
  \clist_use:nn { #1 } { \Vert }
 }
\ExplSyntaxOff

\definecolor{mgcolorfg}{HTML}{120ACF}
\definecolor{mgcolorbg}{HTML}{322AFF}
\colorlet{mgcolorlightbg}{mgcolorfg!35}

\colorlet{lightmgcolorfg}{mgcolorfg!30}
\colorlet{lightmgcolorbg}{mgcolorbg!30}

\definecolor{othergatecolorfg}{HTML}{EA4044}
\definecolor{othergatecolorbg}{HTML}{FF7036}

\colorlet{entangledstatecolor}{cyan!80!black}

\definecolor{mainfigurealgcolor}{HTML}{743834}
\definecolor{mainfigureresultcolor}{HTML}{388697}

\colorlet{mainfigurealgcolorbg}{mainfigurealgcolor!18}
\colorlet{mainfigurealgcolorfg}{mainfigurealgcolor!11}

\colorlet{mainfigureresultcolorbg}{mainfigureresultcolor!15}
\colorlet{mainfigureresultcolorfg}{mainfigureresultcolor!10}

\colorlet{mainfigureimportedcolorbg}{black!7}
\colorlet{mainfigureimportedcolorfg}{black!3}

\tikzset{absorbalgarrow/.style={-{Stealth[length=5pt]}, line width=1.2pt, rounded corners = 1.5pt}}

\def\flatgatesize{0.15}
\def\flatgateadvance{0.3}

\def\circuitRoundedCorners{1.5pt}

\usepackage[normalem]{ulem}

\newcommand{\TUM}{\affiliation{Technical University of Munich, TUM School of Natural Sciences, Physics Department, 85748 Garching, Germany}}
\newcommand{\MCQST}{\affiliation{Munich Center for Quantum Science and Technology (MCQST), Schellingstr. 4, 80799 M{\"u}nchen, Germany}}
\newcommand{\Nottingham}{\affiliation{School of Physics and Astronomy, University of Nottingham, Nottingham, NG7 2RD, UK}}
\newcommand{\CQNE}{\affiliation{Centre for the Mathematics and Theoretical Physics of Quantum Non-Equilibrium Systems, University of Nottingham, Nottingham, NG7 2RD, UK}}

\newcommand{\majo}[1]{\gamma_{#1}}
\newcommand{\cbs}{b} 
\newcommand{\depth}{d}
\newcommand{\bdns}{\beta}
\def\ii{\mathrm{i}}
\def\ee{e}

\newcommand{\partyA}{\mathbf{A}}
\newcommand{\partyB}{\mathbf{B}}
\newcommand{\partyC}{\mathbf{C}}
\newcommand{\lsr}{\operatorname{LSR}}
\newcommand{\sr}{\operatorname{SR}}
\newcommand{\covblock}{J_2}

\makeatletter
\newcommand{\raisemath}[1]{\mathpalette{\raisem@th{#1}}}
\newcommand{\raisem@th}[3]{\raisebox{#1}{$#2#3$}}
\makeatother
\newcommand{\transpose}{ {\raisemath{2pt}{\intercal}} }

\newcommand{\id}{\mathds{1}}
\newcommand\ID\id

\begin{document}

\title{Matchgate circuit representation of fermionic Gaussian states: optimal preparation, approximation, and classical simulation}

\author{Marc Langer} \TUM \MCQST
\author{Ra{\'u}l Morral-Yepes} \TUM \MCQST \author{Adam Gammon-Smith} \Nottingham \CQNE
\author{Frank Pollmann} \TUM \MCQST
\author{Barbara Kraus} \TUM \MCQST

\date{\today}

\begin{abstract}
Fermionic Gaussian states (FGSs) and the associated matchgate circuits play a central role in quantum information theory and condensed matter physics. Despite being possibly highly entangled, they can still be efficiently simulated on classical computers. We address the question of how to optimally create such states when using matchgate circuits acting on product states. To this end, we derive lower bounds on the number of gates required to prepare an arbitrary pure FGS: We establish both an asymptotic bound on the minimal gate count over general nearest-neighbor gate sets and an exact bound for circuits composed solely of matchgates. We present explicit algorithms whose constructions saturate these bounds, thereby proving their optimality. We furthermore determine when an FGS can be prepared with a circuit of any given depth, and derive an algorithm that constructs such a circuit whenever this condition is satisfied, either exactly or approximately. Our results have direct applications to (approximate) state preparation and to disentangling procedures. Moreover, we introduce a new classical simulation algorithm for matchgate circuits, based entirely on manipulating the generating circuits of the FGSs. Finally, we briefly study an extension of our framework for $t$-doped Gaussian states and circuits.
\end{abstract}

\maketitle

\section{Introduction}

Fermionic Gaussian states (FGSs) serve as central analytical and numerical tools in the study of free fermion systems. They have widespread applications in condensed matter physics and quantum chemistry~\cite{McEn20}, for example in the Hartree-Fock approximation~\cite{EcAl07} and in the BCS theory of superconductivity~\cite{BCS57}, as well as in quantum information and quantum simulation. On qubit systems, they correspond to those states which can be generated by acting on computational basis states with matchgate circuits (MGCs) -- a class of quantum circuits that can be efficiently simulated on a classical computer in a variety of settings~\cite{Va01,TeDi02,Kn01,JoMi08,Br05Grassmann,BrKo12}. On a quantum computer, they can be efficiently learned~\cite{GlKl18,Br22,AaGe23,MeHe25,LeMe25} and protocols based on their simulability for verifying quantum computations have been proposed~\cite{CaLa25}. 
Furthermore, supplementing MGCs with non-Gaussian resourceful gates~\cite{JoMi08,BrGa11} or states~\cite{bravyi,HeJo19} is sufficient for realizing universal quantum computations.

One standard approach of representing FGSs classically is via their covariance matrix (CM), which is composed of two-point correlation functions. Given the CM, one can compute higher order correlation functions via Wick's theorem~\cite{Wi50, BaLi94}, entanglement entropies~\cite{PeCh99, Pe03, BoRe04}, and the absolute value of the inner product between FGSs~\cite{Loe55,TeDi02,BrGo17}. Algorithms have been developed to update the CM when acting with a matchgate circuit, or performing a computational basis measurement~\cite{TeDi02,Br05Grassmann,BrKo12}. Combining these tools leads to efficient algorithms for simulating MGCs. Furthermore, classical simulation of MGCs supplemented with non-Gaussian resources has been studied, with simulation costs typically scaling exponential in the number of added resources~\cite{HeJo20,HaTa22,MoLa24,MiFa25,DiKo24,ReOs24,GuHa24}.

It is known that any FGS on $n$ qubits can be prepared by some MGC acting on a product state, and that any MGC can be decomposed into $\mathcal{O}(n^2)$ gates arranged into $\mathcal{O}(n)$ layers~\cite{JiSu18,KiMc18,DaDe19,OsDa22}. Beyond this generic approach, there exist algorithms for preparing some special classes of FGSs~\cite{VeCi09}. Alternative descriptions based on matrix product state (MPS) representations yield matchgate circuits of linear depth, which contain far fewer than $\mathcal{O}(n^2)$ gates for some FGSs~\cite{FiWh15, ThLe23,WuKl25,WoPo25}. Related tensor network techniques such as the Gaussian multiscale entanglement renormalization ansatz (GMERA) can produce very shallow circuits that require long-range gates or mid-circuit insertion of qubits in product states~\cite{FiWh15}. For generating arbitrary states with low entanglement using circuits of, e.g., logarithmic depth composed only of nearest-neighbor gates, methods based on MPS representations are known~\cite{MaSt24}.

Recently, we introduced a special representation of pure FGSs via matchgate circuits in the so-called \emph{right standard form} (RSF)~\cite{MoLa25}, which serves as a useful tool for studying entanglement phase transitions in unitary circuit games~\cite{MoSm24}. In particular, we used the RSF both for simulating the evolution of the states, and to derive a disentangling strategy for FGSs. Central to this approach are some algebraic relations that matchgates satisfy, among them a generalized Yang-Baxter relation~\cite{CaKo22,KoCa22}, which enable systematic simplifications of MGCs. However, relevant quantities in the simulation were still calculated using the CM. Furthermore, to substantiate our claim that the disentangling strategy presented in Ref.~\cite{MoLa25} is optimal, it is necessary to study the circuit complexity of FGSs, i.e., the minimum number of gates required to generate a given state~\cite{Aa16}. Some results in this direction are known: In the context of field theories, some measures of complexity of FGSs have been defined and analyzed~\cite{HaMy18}. Moreover, the complexity of random MGCs has been studied~\cite{BrDi25}.

In contrast, when focusing on preparing a particular given FGS, it its not known what is the exact minimal number of matchgates necessary. Likewise, the problem of minimal required circuit depth has not been studied beyond the known $\mathcal{O}(n)$ bound for generic MGCs. A related problem is then how to efficiently generate circuits that attain such theoretical lower bounds and hence -- potentially -- undercut the upper bounds known from the generic case. Regarding the classical simulation of matchgate circuits, it is crucial to know how the simulation can be performed in the RSF framework. Given that the circuit description of FGSs can be updated efficiently when acting with matchgates, an important question is whether a classical simulator for MGCs can be entirely based on this framework without needing to resort to the CM or other tools. These are precisely the problems we address in this work.

To begin with, we study the problem of efficiently constructing a circuit representation of an FGS, i.e., a circuit that prepares it when acting on a product state. We put a particular emphasis on two classes of such circuit representations. The first class are MGCs in RSF, which can represent arbitrary pure FGSs. We derive simple and efficient algorithms to generate such an MGC from the CM. We then prove that our constructions have minimal gate count in an asymptotic sense when comparing to arbitrary circuits comprised of nearest neighbor gates and exactly when comparing to other MGCs. Second, we focus on a subset of FGSs that can be (approximately) prepared by a circuit of bounded depth, which can be an arbitrary function of the system size. We determine which depth is necessary to prepare a given state and show that our methods achieve such depths.

Next, we derive a classical simulation method based entirely on the circuit representation of FGSs that does not make use of the CM. In particular, we introduce an algorithm for computing the inner product between two FGSs that are generated via two known MGCs. Finally, we extend the circuit representation framework to states that are prepared by MGCs and some non-Gaussian resources. We introduce here standard forms and simulation algorithms for MGCs interleaved with $t$ resourceful gates and show how any circuit can be transformed to such a form.  Circuits in this form have a depth of $\mathcal{O}(n+t)$ as opposed to circuits of depth $\mathcal{O}(nt)$ which one may obtain with standard methods.

The remainder of this paper is organized as follows. We start presenting a summary of our results in section~\ref{sec:summary}. In section~\ref{sec:preliminaries}, we explain the necessary concepts, and review both the algebraic identities matchgates satisfy, as well as the RSF of FGS. In Section~\ref{sec:algorithms}, we review the absorption algorithm for deriving RSF matchgate circuits~\cite{MoLa25}, and present an algorithm for determining the circuit when given the CM. In section~\ref{sec:application_lr_yb} we derive the following applications of the matchgate identities: (i) we show that matchgate circuits in RSF generically have optimal gate count among all other matchgate circuits, (ii) we present the algorithm for evaluating the inner product between FGSs, and (iii), we generalize our results to $t$-doped matchgate circuits. Finally, we turn our attention towards matchgate circuits of bounded depth in section~\ref{sec:shallow}, where we first investigate the exact case, and then show some numerical results for physically relevant states. The appendices contain detailed proofs of our statements, as well as further numerical examples and details.

\begin{figure*}[htb!]
    \centering
    \begin{tikzpicture}[scale=0.8]
    \begin{scope}[yshift=-0.5cm]
        \helpdrawBox{(-0.8,-1.5)}{(10,3.5)}{mainfigurealgcolorfg}{mainfigurealgcolorbg}{\textbf{(a)} Symmetric Euler Decomposition (Sec.~\ref{subsec:sed})}
        \begin{scope}[yshift=0.25cm, xshift=-0.15cm]
            \helpdrawMainfigureSED
        \end{scope}
    \end{scope}

    \begin{scope}[yshift=-5.5cm]
        \helpdrawBox{(-0.8,-4.0)}{(10,3)}{mainfigurealgcolorfg}{mainfigurealgcolorbg}{\textbf{(b)} Depth-$d$ matchgate circuits (Sec.~\ref{sec:shallow})}
        \begin{scope}[yshift=-0.1cm, xshift=-0.15cm]
            \helpdrawMainfigureCircuitDepth
        \end{scope}
        \draw[draw = mainfigurealgcolorfg, line width=1pt] (-0.8,-0.5)
        node[black, anchor=north west, align=left, xshift=0.2cm, yshift=-1pt] {Entanglement cutting alg. (exact \& approximate)} -- (10,-0.5);
        \begin{scope}[yshift=-3.65cm, xshift=-0.15cm]
        \helpdrawMainfigureEntanglementCutting
        \end{scope}
    \end{scope}

    \begin{scope}[xshift=11cm]
        \helpdrawBox{(-0.5,-1.2)}{(6.3,3)}{mainfigureresultcolorfg}{mainfigureresultcolorbg}{\textbf{(c)} Optimal gate count (Sec.~\ref{subsec:rsf_optimal})}
        \begin{scope}[yshift=0.2cm, xshift=0.15cm]
            \begin{scope}[xshift=-0.25cm]
                \helpdrawMainfigureOptimalVtwo
            \end{scope}
            \node[anchor=north west] (N1) at (-0.55,-0.4) {\introfiguretextsize ${}^*$also impossible with $<k/2$ arb.};
            \node[anchor=north west, xshift=0.4em, yshift=-0.3cm] at (N1.north west) {\introfiguretextsize gates in 1d (Sec.~\ref{subsec:sed})};
        \end{scope}
    \end{scope}

    \begin{scope}[xshift=11cm, yshift=-4.7cm]
        \helpdrawBox{(-0.5,-0.3)}{(6.3,3)}{mainfigureresultcolorfg}{mainfigureresultcolorbg}{\textbf{(d)} Inner product alg. (Sec.~\ref{subsec:inner_product_alg})}
        \begin{scope}[yshift=0.1cm, xshift=-0.25cm]
        \helpdrawMainfigureInnerProduct
        \end{scope}
    \end{scope}

    \begin{scope}[xshift=18.3cm]
        \helpdrawBox{(-0.5,-5.0)}{(3,3)}{mainfigureimportedcolorfg}{mainfigureimportedcolorbg}{}
        \begin{scope}[xshift=-0.75cm]
            \helpdrawMainfigureImportedTools
        \end{scope}
    \end{scope}

    \begin{scope}[yshift=-8.5cm, xshift=11cm]
        \helpdrawBox{(-0.5,-1.)}{(10.3,3)}{mainfigureresultcolorfg}{mainfigureresultcolorbg}{\textbf{(e)} Generalized standard circuit forms (Sec.~\ref{subsec:generalization_doped})}
        \helpdrawMainfigureSWAPStandardForm
    \end{scope}
    \end{tikzpicture}
    
    \caption{
    Results and algorithms. Panels \textbf{(a)} and \textbf{(b)} (lower panel) summarize the algorithms to determine matchgate circuits (MGCs) that generate a given fermionic Gaussian state (FGS): \textbf{(a)} In the symmetric Euler decomposition, matchgates are selected to eliminate entries in the covariance matrix (CM). Pairs of qubits are successively removed and the circuit produced is in right standard form (RSF). \textbf{(b)} FGSs whose CM is $\depth$-banded can be prepared with MGCs of depth $\mathcal{O}(\depth)$. For determining low-depth circuit for states with (approximately) banded CM, the entanglement cutting algorithm produces small MGCs that cut the state into two subsystems. Panels \textbf{(c)}, \textbf{(d)} and \textbf{(e)} show applications of RSF circuits and the algebraic matchgate identities (GYB and LR relation in the top right panel): \textbf{(c)} Since any FGS admits a representation as an RSF circuit, comparing RSF circuits for the same FGS allows one to identify the circuit with minimal gate count. \textbf{(d)} Efficient evaluation of the inner product of two FGS (including phase information) by rewriting the inner product as a simple tensor network contraction. \textbf{(e)} Standard forms similar to the RSF can be derived for circuits, states and inner product evaluations, where the initial circuit is an MCG interleaved with $t$ resourceful gates (resourceful gates are depicted in red, whereas matchgates are blue).
    }
    \label{fig:summary_figure}
\end{figure*}
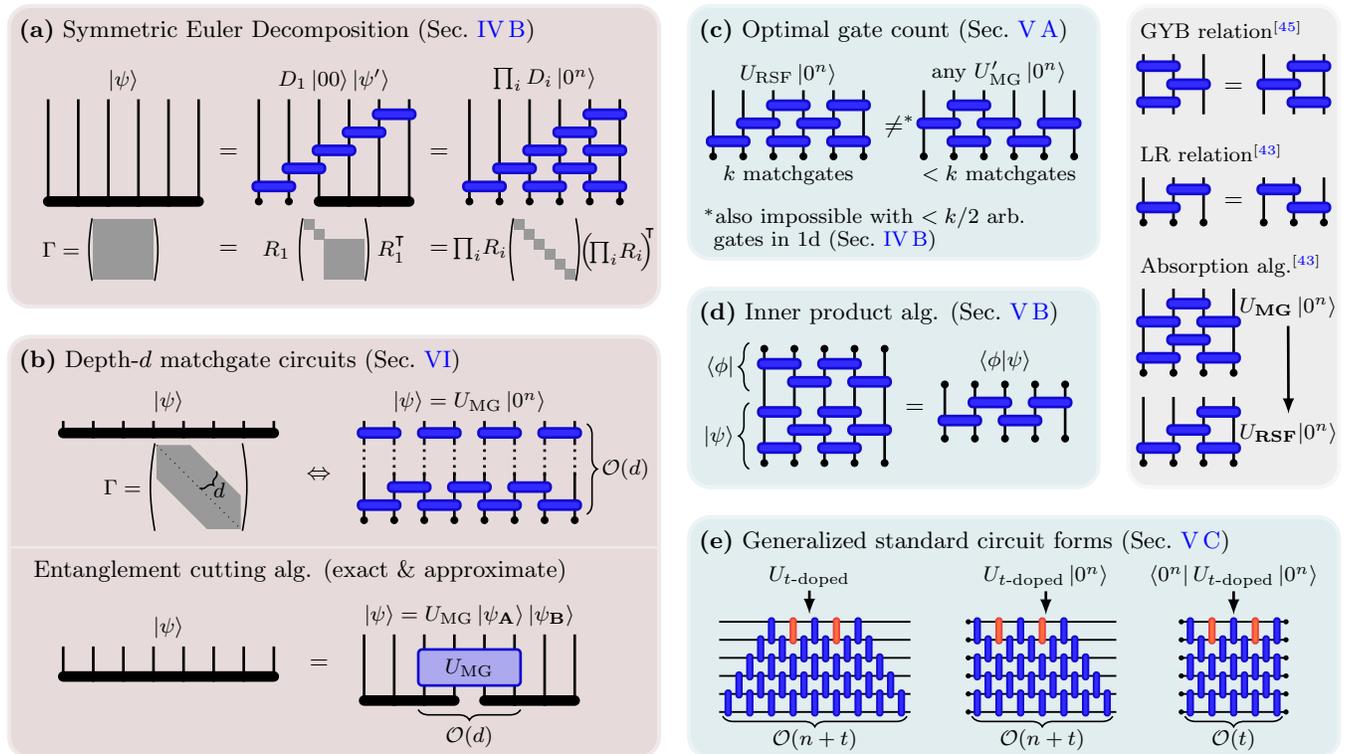

\section{Summary of results}
\label{sec:summary}

In this section, we summarize the main results of our work (see  Fig.~\ref{fig:summary_figure} for an illustration). Broadly speaking, we pursue two different directions. First, we introduce two classes of efficient algorithms for determining (RSF) circuits that prepare a given FGS when provided its CM (Fig.~\ref{fig:summary_figure}a and~b) and study their applications. Second, we explore further applications of representing FGSs via (RSF) circuits and the algebraic matchgate identities. The statement on exact minimal matchgate count can be obtained by comparing MGCs in RSF (Fig.~\ref{fig:summary_figure}c). Beyond that, we study new classical simulation algorithms and further circuit standard forms (Fig.~\ref{fig:summary_figure}d and~e). Below, we state the main contributions of this paper and explain how different aspects of our work are related.

The first class of algorithms we introduce aims to decompose the CM of a given FGS directly into a series of Givens rotations which correspond to elementary matchgates (Fig.~\ref{fig:summary_figure}a), without the need to diagonalize the CM first. The rotations are selected such that when acting on the CM, entries are eliminated in a column by column fashion. The idea is very similar to a decomposition of a rotation matrix into Euler angles. However, we decompose the covariance matrix directly by acting with Givens rotations and their transpose from the left and right respectively. Owing to this structure, we term this class of algorithms \emph{symmetric Euler decompositions}. The algorithm obtained eliminating each column with a sequence of Givens rotations is particularly simple to implement (see Sec.~\ref{subsec:sed}). However, the elementary matchgates corresponding to Givens rotations utilize only one of two possible nonlocal parameters. Instead, a modified version (in Appendix~\ref{app:sed_two_columns_version}) outputs matchgates in which both nonlocal parameters are utilized. This is done by eliminating two columns in the CM at once instead of a single one. The output circuit is then directly in RSF. Finally, a more careful selection of which columns to eliminate in the CM leads to a third version of the symmetric Euler decomposition (see Appendix~\ref{app:enhaced_decomposition}). In this case, when viewed as a disentangling algorithm, we show that every gate maximally reduces the Schmidt rank in the bipartition where it is applied.

Building on this, a first connection to circuit complexity can be drawn: Denote by $\lsr_k(\ket{\psi})$ the logarithm of the Schmidt rank of a state $\ket\psi$ in a bipartition $1,\ldots,k\vert k+1,\ldots,n$. By applying the algorithm presented in Appendix~\ref{app:enhaced_decomposition}, we show that a given FGS $\ket\psi$ can be prepared by an MGC containing at most
\[
    K = \sum_{k=1}^{n-1} \lsr_k( \ket{\psi})
\]
matchgates. Note that generating $\ket{\psi}$ with any gate set under the nearest neighbor restriction requires at least $K/2$ gates~\footnote{To see this, consider the application of a gate $U$ to qubits $i$ and $i+1$ of some state $\ket\varphi$. The difference $\lsr_i(U \ket{\phi}) - \lsr_i(\ket\phi)$ can be at most two, since $|\lsr_{i}(\ket\chi) - \lsr_{i\pm 1}(\ket\chi)| \leq 1$ due to subadditivity for both $\ket\chi = \ket\varphi$ and $\ket\chi =U \ket{\varphi}$. Hence, even if all gates in a circuit generating $\ket{\psi}$ increase the Schmidt rank maximally, this circuit needs to contain at least $K/2$ gates.}. Thus, MGCs achieve asymptotically optimal gate count to generate arbitrary pure FGSs.

Next, we restrict our attention to the exact gate count comparison between matchgate circuits. Specifically, we ask under which conditions a matchgate circuit uses the minimal number of gates among all possible matchgate circuits. We show the following: If an FGS is generated by an MGC in RSF, under some mild conditions on the gates, the state cannot be generated by any other matchgate circuit with even a single gate fewer (Fig.~\ref{fig:summary_figure}c). The conditions on the gates rule out that, e.g., the gates factorize. Gates which do not satisfy these conditions form a set of measure zero. Our argument relies on the gate absorption algorithm~\cite{MoLa25}, which transforms an arbitrary MGC acting on a computational basis state into an RSF circuit. This enables us to reduce the problem to the much more tractable comparison of MGCs in RSF. We remark that the circuits produced by the symmetric Euler decomposition algorithms typically satisfy the demanded conditions.

While the previous results focus on gate count, they do not address bounds on the required circuit depth. Such bounds can be obtained using a different approach: Given the covariance matrix $\Gamma$ of $\ket\psi$, define $\depth$ as the largest integer such that matrix entries $\Gamma_{i,i+\depth}$ are nonzero (we say that $\Gamma$ is $\depth$-banded). Note that when $\ket\psi$ is prepared by a circuit of depth $\depth$, its CM is $\mathcal{O}(\depth)$-banded. Here, we show that the reverse also holds: Any FGS whose CM is $\depth$-banded can be prepared by an MGC of depth $\mathcal{O}(\depth)$ (Fig.~\ref{fig:summary_figure}b). For proving this, it is sufficient to consider a particular algorithm for computing a symmetric Euler decomposition and show that in each step, the intermediate CMs remain banded. It then follows that the circuit produced by this algorithm has the desired depth.

The existence of a circuit of depth $\depth$ clearly implies that for any bipartition $1,\ldots,k\vert k+1,\ldots,n$, there exists a circuit acting on $\mathcal{O}(\depth)$ qubits neighboring qubit $k$ that, when applied, completely disentangles the state in this bipartition. A question one might now have is whether this circuit can be determined efficiently from local properties of the state, i.e., without knowledge of the global circuit. We show that this is indeed the case and present an algorithm that determines such a circuit (lower panel of Fig.~\ref{fig:summary_figure}b). Applying this $\mathcal{O}(n/k)$ times in different regions cuts the state into $\mathcal{O}(n/k)$ states of size $\mathcal{O}(d)$ that are not entangled to each other. These can be independently disentangled with small MGCs. In total, this gives a different algorithm for determining MCGs, which we term \emph{entanglement cutting algorithm}. Since this algorithm only uses local properties of the state, it is numerically much more stable compared to the symmetric Euler decomposition algorithms, where numerical errors could propagate through the entire system. This renders such algorithms potentially useful in preparing states with approximately banded CM, such as ground states of local Hamiltonians. To demonstrate this, we study the ground state of the one-dimensional Ising model in the disordered phase (Sec.~\ref{sec:numerical_examples}), as well as a model with algebraically decaying correlations and a two-dimensional model (Appendix~\ref{app:further_numerical_examples}). In all examples, we find a tradeoff between circuit depth and achieved fidelity when using our algorithm.

We study furthermore the problem of classically simulating FGSs that are represented by some MGCs. Our focus is on methods that are entirely based on the circuit representation framework and do not make use of, e.g., the CM formalism. We present an algorithm that makes use of the algebraic matchgate identities to compute the inner product $\braket{\phi}{\psi}$ between two FGSs $\ket{\psi}=U\ket{\cbs_1}$ and $\ket{\psi} = V\ket{\cbs_2}$ that are generated by known MGCs $U$ and $V$ acting on computational basis states $\ket{\cbs_1}$ and $\ket{\cbs_2}$ respectively (Fig.~\ref{fig:summary_figure}d). Similar to the absorption algorithm~\cite{MoLa25}, the idea is to repeatedly apply the GYB and LR identities to decrease the number of gates. This way, the problem of computing an inner product can be reduced to the contraction of a depth-2 matchgate circuit. Let us point out that this approach is not only limited to matchgates, but to any group of gates satisfying similar algebraic identities. To show that there exist non-matchgate solutions, in Appendix~\ref{app:other_GYB_LR} we give two trivial examples of gate sets for which this is the case. Additionally, given that there is some similarity between free fermionic and bosonic systems, we also briefly study the latter in Appendix~\ref{app:bosonic_gaussian_systems}. We show that almost all bosonic Gaussian unitaries satisfy the algebraic identities. However, unlike in the fermionic case there is a set of measure zero for which the relations do not hold.

Since our algorithm computes not only the absolute value of inner products, but also includes phase information, it can be used to simulate circuits in which a small number $t$ of resourceful gates are inserted. This has been studied before in Refs.~\cite{DiKo24,ReOs24} using different algorithms to evaluate inner products between FGSs. Consider a so-called $t$-doped circuit of the form
\[
U = G_{t+1} C_t \ldots C_1 G_1,
\]
where $G_i$ are MGCs and $C_i$ are non-Gaussian gates, such as the SWAP gate. Using gadgetization techniques~\cite{bravyi,HeJo19,ReOs24} or suitable decompositions of $C_i = \sum_{j} M_{ij}$ into sums of MGCs,  $\bra{x} U \ket{0^n}$ can be computed by summing over exponentially in $t$ many inner products between FGSs. When aiming to evaluate $\bra{x} U \ket{0^n}$ only approximately, better algorithms are known that scale with the so-called fermionic extent~\cite{DiKo24,ReOs24,CuSt25}.

Let us remark that the ability to efficiently compute inner products already implies a range of other simulation capabilities. In the literature, one distinguishes between strong and weak simulation~\cite{JoVa14}. A strong simulator computes arbitrary output probabilities, including marginals on subsets of qubits, whereas a weak simulator efficiently samples from the output distribution $p(x) = \vert\bra{x}U\ket{0^n}\vert^2$. By employing the sequential sampling procedure of~\cite{BrGo22} and combining it with an efficient inner product evaluation algorithm, one obtains an efficient weak simulator for MGCs in the computational basis~\footnote{This approach requires that the local Hilbert space dimension is finite.}. Another commonly encountered simulation task is the calculation of expectation values. In the case of MGCs, expectation values of a Pauli string $P$ can be computed by interpreting $P$ itself as a small MGC and then evaluating the expression $\bra{0} U^\dagger P U \ket{0}$.

Finally, we report here also on an alternative way to evaluate inner products of the form $\bra{x} U \ket{0^n}$, where $U$ is a $t$-doped MCG. Instead of decomposing each resourceful gate into a linear combination of matchgates, we show that one can use the algebraic matchgate identities to reduce the calculation of the inner product to calculating the contraction of a brickwall circuit with depth $\mathcal{O}(t)$ (Fig.~\ref{fig:summary_figure}e). The idea here is to use the fact that the resourceful gates can first be moved to act only on the first qubits lines~\cite{JoMi08,MeHe25}. After doing so, the circuit depth can be reduced using the algebraic matchgate identities. However, since those identities do not apply to the resourceful gates, and thus one obtains a total depths of $\mathcal{O}(t)$ instead of $2$ as is the case with undoped MGCs. More generally, we also consider the problem of defining circuit standard forms for $t$-doped MGCs and FGSs. We show that any $t$-doped MGC can be transformed into a $t$-doped MGC of depth $O(n+t)$. A similar result can, of course, be obtained for $t$-doped FGS.

\section{Preliminaries and definitions} \label{sec:preliminaries}

In this section we review matchgates and fermionic Gaussian states. We begin by establishing notation and recalling basic results on matchgates that will be used throughout the paper. We then review key algebraic properties satisfied by matchgates. Finally, we summarize the right standard form (RSF) for matchgate circuits introduced in~\cite{MoLa25}, which provides a representation of pure FGSs as circuits of matchgates with a characteristic structure.

\subsection{Matchgates and FGS}

We denote the standard Pauli operators as $X_i$, $Y_i$, and $Z_i$, with the optional subscript indicating on which qubit the operators act. With $\id_k$, we denote the $k\times k$ identity matrix, where we may drop the subscript if the context is clear.

A \emph{matchgate}~\cite{Va01} is a two-qubit gate of the form $G(A,B) = A\oplus B$, with the unitary matrix $A$ ($B$) acting on the subspace spanned by $\ket{00},\ket{11}$ ($\ket{01},\ket{10}$) and satisfying the constraint $\det A = \det B$. Up to a global phase, any matchgate $G$ can be decomposed into local and non-local parts
\begin{equation} \label{eq:parametrization_of_mg_local_nonlocal}
    G = U_1 \!\otimes\! U_2 \ \exp(\ii \alpha \, X \!\otimes\! X + \ii \beta \, Y \!\otimes\! Y) \ U_3 \!\otimes\! U_4,
\end{equation}
with single qubit rotations $U_k = \exp(\ii \phi_k Z)$~\cite{BrGa11, SoAk14}. We refer to $\alpha$ and $\beta$ as the \emph{nonlocal parameters} of~$G$.

Matchgate circuits (MGCs) are circuits composed out of matchgates acting on neighboring qubits arranged in a line. Any MGC $U$ preserves the span of the $2n$ Majorana operators
\[
    \majo{2k-1} = Z_1 \ldots Z_{k-1} X_k, \quad
    \majo{2k} = Z_1 \ldots Z_{k-1} Y_k,
\]
which are used to map fermionic to qubit systems~\cite{JoWi28}. In particular, the action of a MGC  $U$ on the Majorana operators is given by
\begin{equation}
    \label{eq:mg_conj_majo}
    U^\dagger \majo{k} U = \sum_{l=1}^{2n} R_{kl} \majo{l},
\end{equation}
where $R$ is a special orthogonal matrix~\cite{TeDi02}. Allowing instead for arbitrary orthogonal matrices in Eq.~\eqref{eq:mg_conj_majo} corresponds to circuits of gates of the form $G(A,B) (P \otimes \ID_2)$, where $G(A,B)$ is a matchgate and $P\in\{\ID_2, X\}$. Our results can easily be generalized to this gate set.

Given an orthogonal matrix $R$, one can conversely find a MGC $U$ composed of $\mathcal{O}(n^2)$ many gates such that Eq.~\eqref{eq:mg_conj_majo} holds~\cite{JiSu18,KiMc18,DaDe19,OsDa22}. This can be achieved with a \emph{generalized Euler decomposition} (see also Refs.~\cite{RaRu69, HoRa72}) of $R$ into \emph{Givens rotations}, i.e., rotations of the form \[E_{k,k+1}(\phi) = \id_{k-1} \oplus \begin{pmatrix}
    \cos\phi & \sin\phi \\ -\sin\phi & \cos\phi
\end{pmatrix} \oplus \id_{2n-k-1}, \]
which act non-trivially only on the $k$ and $k+1$-th standard basis vector. For completeness, we provide an explicit construction thereof in Appendix~\ref{appendix:euler_decomp_orth}. Alternatively to Eq.~\eqref{eq:parametrization_of_mg_local_nonlocal}, one can express a single matchgate $G$ as
\[
    \id\!\otimes\! U_1  \exp(\ii \varphi_2 \,X \!\otimes\! X) U_3\!\otimes \!U_4   \exp(\ii \varphi_5 \,X \!\otimes\! X) \id\!\otimes\!U_6,
\]
with some local gates $U_k = \exp(\ii \phi_k Z)$, using the Euler decomposition.

Mixed \emph{fermionic Gaussian states} (FGSs) are defined to be the thermal states $\exp(-\beta H)/\Tr(\exp(-\beta H))$ of quadratic Hamiltonians $H=\sum_{kl}h_{kl}\majo{k}\majo{l}$. Pure FGS are limits of such expressions. For instance they can be obtained by taking the limit $\beta\to\infty$. Alternatively, any FGS $\rho$ can be written as 
\[
    \rho = U \bigotimes_{k=1}^n \begin{pmatrix}
        \chi_k & 0 \\ 0 & 1-\chi_k
    \end{pmatrix} U^\dagger
\]
where $U$ is an MGC, and $0\leq \chi_k\leq 1$. Pure FGS can correspondingly be written as $U\ket{\cbs}$, where here and in the following $\ket{\cbs} = \ket{\cbs_1,\ldots,\cbs_n}$ denotes an arbitrary computational basis state.

The \emph{covariance matrix} (CM) $\Gamma$ of any state $\rho$ is defined via 
\[
    \Gamma_{kl} = \frac{\ii}{2} \Tr(\rho [\majo{k},\majo{l}]),
\]
where $[A,B] = AB - BA$ denotes the commutator. The CM is always antisymmetric, i.e., $\Gamma = - \Gamma^\transpose$. Furthermore, it holds that $\Gamma\Gamma^\transpose \leq \id$ (meaning that $\ID - \Gamma\Gamma^\transpose$ is positive semidefinite), with equality if and only if the corresponding state is a pure FGS.

In case $\rho$ is an FGS, then it is completely characterized by its CM via Wick's theorem~\cite{Wi50, BaLi94}, which states that the expectation value of a Majorana string operator $\majo{k_1} \ldots \majo{k_L}$ is given by
\[
    i^L\Tr(\rho \majo{k_1} \ldots \majo{k_L}) = \operatorname{Pf} \left(\Gamma_{\{k_1,\ldots, k_L\},\{k_1,\ldots, k_L\}}\right),
\]
where $\Gamma_{\{a_1,\ldots a_M\}, \{b_1,\ldots,b_N\}}$ is the matrix obtained by selecting rows and columns of $\Gamma$ with indices $\{a_1,\ldots a_M\}$ and $\{b_1,\ldots,b_N\}$ respectively, and $\operatorname{Pf}$ denotes the \emph{Pfaffian} of a matrix.

When applying a MGC $U$ to a state $\rho$ with CM $\Gamma$, the CM of the updated state $U\rho U^\dagger$ is given by $R \Gamma R^\transpose$. Likewise, update rules for the CM can be derived when measuring in the computational basis~\cite{BrKo12}, and lead to efficient simulation algorithms for MGC~\cite{TeDi02}. For any CM $\Gamma$, there exists a rotation $R$ s.t. 
\begin{equation}
R \Gamma R^\transpose = \bigoplus_{i=1}^{n} \lambda_i \covblock, \label{eq:fermionic_williamson_form}
\end{equation}
with $\vert\lambda_i\vert \leq 1$ and
\[
\covblock = \begin{pmatrix} 0 & -1 \\ 1 & 0 \end{pmatrix}.
\]
Such a decomposition is called the (fermionic) Williamson normal form, and $\lambda_i$ are called the Williamson eigenvalues of $\Gamma$~\cite{BoRe04, SuTa20}. Computing an Euler decomposition of $R$ in Eq.~\eqref{eq:fermionic_williamson_form} is part of a straightforward way for constructing an MGC that generates an FGS $\ket{\psi}$.

A particularly useful result in our context on entanglement of pure FGS is due to Botero and Reznik~\cite{BoRe04}: Given two parties $\partyA$ and $\partyB$, comprising qubits $1,\ldots,k$ and $k+1,\ldots,n$ respectively, and a pure FGS $\ket{\psi}$ shared between the parties, there exist matchgates circuits $U_\partyA$ and $U_\partyB$, acting on $\partyA$ and $\partyB$, such that $U_\partyA \otimes U_\partyB \ket{\psi}$ factorizes into three constituents: $n_\partyA$ and $n_\partyB$ computational basis state contained in parties~$\partyA$ and~$\partyB$ respectively, and $n_\text{e}$ entangled pairs of the form
\[
    \ket{\psi_i} = \cos{\chi_i}\ket{00} - \sin{\chi_i}\ket{11}
\]
shared between both parties. The values $\chi_i$ are related to those Williamson eigenvalues $\lambda_i$ of the reduced CM on any party with $|\lambda_i| \neq 1$ via $\cos(2\chi_i)=\lambda_i$. The entanglement shared between parties $\partyA$ and $\partyB$ is thus fully described by the numbers $\chi_i$. In particular, the Schmidt rank $\sr(\ket{\psi}, \partyA, \partyB)$ in the given splitting is always a power of $2$.

\subsection{Algebraic matchgate circuit identities}

Several of our results rely on two algebraic identities that small MGCs satisfy. The first identity can be understood as a generalization of the Yang-Baxter relation \cite{Ya67,Ba72}, which has originally been introduced in the context of statistical physics and is strongly linked to integrability~\cite{Sm92, KoBo93}. This \emph{generalized Yang-Baxter} (GYB) relation for matchgates has been identified and used to compress the depth of MGCs in Refs.~\cite{CaKo22,KoCa22}. Together with the second identity, the \emph{left-right} (LR) identity, matchgate circuits to represent pure FGS can be compressed even further~\cite{MoLa25}. Below, we review both these identities and comment briefly on how to prove them. Both identities are represented graphically in the top right panel of Fig.~\ref{fig:summary_figure}.

We start with the GYB relation~\cite{CaKo22,KoCa22}: Given three matchgates $U_{1},$ $U_{2}$, and $U_{3}$, one can find another three matchgates, $V_{1}$, $ V_{2}$, and $ V_{3}$, s.t.
\begin{equation}
    \label{eq:yb}
    (U_{1}\otimes \id) (\id\otimes U_{2}) (U_{3}\otimes\id ) =
    (\id\otimes V_{1}) (V_{2}\otimes\id) (\id\otimes V_{3}).
\end{equation}
The converse direction, i.e., when given the $V$ matrices, holds analogously. Equation~\eqref{eq:yb} can be proven for matchgates by computing an Euler decomposition of an arbitrary $6\times6$ orthogonal matrix, which corresponds to an MGC on three qubits (see Appendix~\ref{appendix:euler_decomp_orth} and Appendix~A of Ref.~\cite{MoLa25}). This gives Eq.~\eqref{eq:yb} up to a complex phase $\ee^{\ii\varphi}$, which can be easily determined since the left and right-hand-side are matrices of size $8\times 8$.

The second identity we review here, the LR identity, can be applied when representing FGS on three qubits via MGCs. Given a computational basis state on three qubits, $\ket{\cbs}$, and two matchgates, $U_{1}$ and $ U_{2}$, there exist another two matchgates, $ V_{1}$ and $V_{2}$, s.t.
\begin{equation}
    \label{eq:lr}
    (U_{1} \otimes \id) (\id\otimes U_{2}) \ket{\cbs} = 
    (\id\otimes V_{1}) (V_{2}\otimes\id ) \ket{\cbs}.
\end{equation}
In fact, any FGS $\ket{\psi}$ on three qubits can be decomposed into the left or right-hand-side of Eq.~\eqref{eq:lr} for suitable matchgates and~$\ket{\cbs}$. Up to a phase, this can be seen by an application of Algorithm~\ref{alg:symmetric_euler_decomposition}, which we present later. Alternatively, one can parametrize an arbitrary FGS on three qubits and easily derive both decompositions (see Appendix~A of Ref.~\cite{MoLa25})

A natural question to ask is whether these algebraic relations are only a property of matchgates, or other gate sets exist that satisfy them. We answer this question in Appendix~\ref{app:other_GYB_LR}, where we discuss two trivial non-matchgate examples for which this is the case. Moreover, we also examine the bosonic case and show that all bosonic Gaussian unitaries except a set of measure zero satisfy a GYB and LR relation.

\subsection{RSF circuits}

In Ref.~\cite{MoLa25}, we introduced the \emph{right standard form} (RSF) as a representation of an arbitrary FGS via an MGC acting on a product state. We also provided an algorithm to efficiently modify such representation upon the application of a matchgate, and used it to propose a disentangling method for pure FGS. Here, we briefly review the definition of RSF circuits, as this will be crucial for showing that RSF circuits are optimal with respect to gate count under certain mild conditions. Later, in Sec.~\ref{sec:algorithms}, we present algorithms to obtain the RSF circuit representation both from a generic MGC and from a covariance matrix.

\def\numdiag{\ensuremath{{m}}}

As defined in Ref.~\cite{MoLa25}, an RSF is a special layout a circuit can have. Formally, it is defined by a list of $0 \leq \numdiag \leq \floor{n/2}$ pairs of integers $((k_i, l_i))_{i=1}^\numdiag$, satisfying $1\leq k_i \leq k_{i+1} -2 \leq n-1$, and $1\leq l_i \leq n-k_i$, called the \emph{parameters} of the RSF. A circuit $C$ is in RSF, if it is a product
\[
    C = D_1 \ldots D_\numdiag
\]
of so-called \emph{diagonals}. Each of those is given as a sequence of gates
\[
    D_i = U_{il_i} \ldots U_{i 1},
\]
where the gate $U_{ij}$ acts on the qubits with indices $k_i + j -1$ and $k_i + j$. We call $k_i$ the \emph{position} of the diagonal $D_i$, as it is the index of the first qubit on which the diagonal acts, and $l_i$ is the \emph{length} of the diagonal. An illustration of a particular RSF layout is depicted, e.g., in the rightmost circuit of Fig.~\ref{fig:summary_figure}a.

Typically, we rearrange RSF circuits into layers such that all gates $U_{i1}$ are in the first layer, $U_{i2}$ are in the second one, and so on. This is possible since the conditions on RSF circuits ensure that, e.g., the first gates of each diagonal act on distinct qubits and can thus act in parallel within a layer. In this way, the depth, i.e., the number of layers, of an RSF circuit is determined by its longest diagonal.

Even though (as we explain next) RSF circuits can generate any pure FGS, the layout of such circuits is quite restrictive. For instance, no gate can act `after' the first diagonal. This implies that the qubit with index $1$ is acted on by at most a single gate. More generally, any gate $G$ in an RSF circuit will either act directly on the input state (in which case it is in the first layer), or is part of a local structure $(\id_2 \otimes G)(G' \otimes \id_2)$ on three qubits with another gate $G'$. For a given $n$, the number of different RSF circuit layouts is thus finite. In fact, it is given by the so-called telephone number $T(n)$~\cite{MoLa25, origin_of_telephone}. Furthermore, there is a \emph{maximal} RSF with $\numdiag= \floor{n/2}$ and parameters $((2i, n-2i + 1))_{i=1}^\numdiag$, having the largest number of gate among all RSF circuits on $n$ qubits, $\floor{n^2/4}$. As we will argue below, the maximal RSF can represent any pure FGS.

\section{Algorithms to derive RSF circuits}
\label{sec:algorithms}

As established in Ref.~\cite{MoLa25}, RSF matchgate circuits acting on a product state can represent any arbitrary pure FGS $\ket{\psi}$. The argument for this is as follows: Firstly, one can clearly represent $\ket{\psi}$ as some MGC acting on $\ket{0^n}$. One can then systematically use the GYB and LR properties to rewrite the circuit into RSF. We call this procedure the \emph{absorption algorithm}. In this section, we first review this algorithm as it plays a crucial role in subsequent results. We then present another algorithm, which takes as an input the covariance matrix of a pure FGS $\ket{\psi}$, and produces as an output an MGC  in RSF which generates $\ket{\psi}$. This algorithm is very different to, and simpler than, computing an Euler decomposition of the orthogonal matrix which brings the CM to the Williamson normal form, and then applying the absorption algorithm to each gate resulting from this procedure. Instead, the RSF can be computed directly from the CM by virtue of a `symmetric Euler decomposition', i.e., an Euler decomposition in which each of the constituent rotations is applied from the left and right instead of only from one side.

\subsection{Absorption algorithm}
\label{subsec:absorption}
The absorption algorithm takes as input an FGS $\ket{\psi} = U \ket \cbs$, represented by a matchgate circuit $U$ in RSF, and a matchgate $A$ to be applied on two neighboring qubits.
As an output, it provides another MGC representing the state $A\ket{\psi}$. The output circuit is again in RSF and contains at most one gate more than the initial circuit. If it contains the same number of gates, the output circuit has the same RSF layout as the initial one. However, several gates that compose it might be modified. This algorithm can be applied iteratively to turn any MGC acting on a product state into RSF. To do so, we start from a product state, which by definition is in RSF, and then we apply gates sequentially, each time using the absorption algorithm.

We refer the reader to Ref.~\cite{MoLa25}, in particular Section~III~D and Appendix~B for a comprehensive description of the absorption algorithm. Here, we summarize the main concept. The starting point of the algorithm is a circuit in RSF and a gate $A$ applied to qubits $i$ and $i+1$. The following is now repeatedly performed until a termination condition is reached. We check whether $A$ together with some gates in its vicinity matches one of four patterns (e.g., $(A\otimes \id) (\id\otimes U) (U'\otimes\id)$ for some gates $U$ and $U'$). If one of these patterns is encountered, a corresponding GYB or LR relation is applied, and a specific one of the output gates is relabeled to $A$. For technical reasons, at most one of these patterns can be encountered in each step (see Appendix~B of Ref.~\cite{MoLa25}). If none of them are encountered, the algorithm terminates. This happens after at most $\mathcal{O}(n)$ steps. The gate $A$ can then either be multiplied with another gate, or be appended to one diagonal in the RSF. The absorption algorithm can be formulated as follows.

\begin{algorithm}[H]
\caption{Absoption algorithm}
\label{alg:absorption}
\KwIn{ 
    A matchgate $A$ acting on qubits $i$, $i+1$, an RSF circuit $C$ and a product state $\ket{\cbs}$ describing the FGS $C\ket\cbs$.
}
\KwOut{
    An RSF circuit describing $AC\ket\cbs$.
}

\While{true}{
    $G \gets $ gates in vicinity of $A$;
    
    \If{$G$ matches $(A\otimes \id) (\id\otimes U) (U'\otimes\id)$}{
        $(\id\otimes V) (V'\otimes\id) (\id\otimes A) \gets G$; $\quad\,\,\,$\texttt{(GYB)}}
    \ElseIf{$G$ matches $(A \otimes \id) (\id\otimes U) \ket{\cbs}$}{
        $(\id\otimes V) (A\otimes\id) \ket{\cbs} \gets G$;  $\quad\quad\quad\quad$\texttt{(LR)}}
    \ElseIf{$G$ matches $(U\otimes \id) (\id\otimes A) \ket{\cbs}$}{
        $(\id\otimes A) (V\otimes \id) \ket{\cbs} \gets G$; $\quad\quad\quad\quad$\texttt{(LR)}
    }
    \ElseIf{$G$ matches $(U\!\otimes\! \id) (\id\!\otimes\! U') (A\!\otimes \!\id)$}{
        $(\id\otimes A) (V\otimes \id) (\id\otimes V') \gets G$; $\quad\,\,\,$\texttt{(GYB)}
    }
    \Else{ Terminate loop.}
    
    Append $A$ to RSF or combine with other gate.
}
\end{algorithm}

A proof that the output circuit of this procedure is again in RSF is contained in Appendix~B of Ref.~\cite{MoLa25}. An illustration of this algorithm with a maximal RSF is provided in Fig.~\ref{fig: inner product}a.

\subsection{Symmetric Euler Decomposition}\label{subsec:sed}

The absorption algorithm is used to transform an arbitrary known circuit that generates an FGS $\ket\psi$ into RSF. In contrast to that, we present now an algorithm that produces the RSF circuit given the covariance matrix $\Gamma$ of $\ket{\psi}$ as an input~\footnote{Since the covariance matrix does not contain phase information, the circuit prepares $\ket{\psi}$ up to a global phase. The phase can be fixed by, e.g., knowing the inner product with a reference state~\cite{DiKo24,ReOs24}.}. Recall that $\Gamma$ is both orthogonal and antisymmetric in this case. Similar to an Euler decomposition, the main idea is to compute a sequence of Givens rotations that eliminate specific entries when applied to the CM by setting them to zero. One key difference is, however, that these rotations have to act from both sides of the CM in order to correspond to matchgates. Another distinction is that only every second column of $\Gamma$ needs to be eliminated, whereas the relevant entries in the remaining columns become zero automatically, as we will show later.

In the following, we present an algorithm for computing a symmetric Euler decomposition. Before doing so, we define the subroutine \texttt{elim\_second} which takes as an input a vector $v = \begin{pmatrix} v_1 & v_2 \end{pmatrix}^\transpose$ and outputs a $2 \times 2 $ rotation matrix $S$. This matrix is constructed to give $Sv = \begin{pmatrix}\vert v \vert & 0\end{pmatrix}^\transpose$, where $|v|^2 = v_1^2 + v_2^2$. The idea is now to repeatedly use this subroutine and apply its output to the CM in order to eliminate each second column (see Fig.~\ref{fig: sed}a for an illustration). We present now the algorithm based on this idea. As we will see, the MGC corresponding to the output of this algorithm is not yet in RSF. We explain how to obtain the RSF circuit below.

\begin{algorithm}[H]
\caption{Symmetric Euler decomposition}
\label{alg:symmetric_euler_decomposition}
\KwIn{
    CM $\Gamma$ of a pure FGS.
}
\KwOut{
    A sequence of Givens rotations $R_i$ and the CM $\Gamma'$ of a computational basis state, s.t. $\prod_i R_i \Gamma (\prod_i R_i)^\transpose = \Gamma'$.
}

$\Gamma^{(0,0)} \gets \Gamma$

\For{$q = 1,\ldots, n - 1$}{ \label{algline:sed:outer_loop}

 $r_q \gets \max\{r \vert \Gamma^{(q-1, 2q-2)}_{r, 2q-1} \neq 0\}$ \label{algline:sed:choice_rq}

 $\Gamma^{(q,r_q)} \gets \Gamma^{(q-1, 2q-2)}$

 \If{$r_q \geq 2q+1$}{

     \For{$k = r_q, r_q -1 , \ldots, 2q+1$}{
     
        $\tilde R^{(q,k)} \gets$ \texttt{elim\_second($\Gamma^{(q,k)}_{\{k-1,k\},2q-1}$)} \label{algline:sed:def_rotation}
        
        $R^{(q,k)} \gets \id_{k-2} \oplus \tilde R^{(q,k)} \oplus \id_{2n-k}$ \label{algline:sed:def_full_rotation}
    
        $\Gamma^{(q,k-1)} \gets R^{(q,k)} \Gamma^{(q,k)} R^{(q,k)\transpose}$ \label{algline:sed:apply_rotation}
     }

 }

}

\Return{} inverse of sequence $(R^{(q,k)})$, $\Gamma^{(n-1,2n-2)}$ 

\end{algorithm}

Let us now explain why it suffices to eliminate only every second column. To this end, we show that the intermediate CMs $\Gamma^{(q,2q)}$ are of the form $\Gamma^{(q,2q)} = (\bigoplus_{i=1}^{q-1} \covblock) \oplus\tilde \Gamma^{(q)}$. Due to the choice of $r_q$ in line~\ref{algline:sed:choice_rq} and the action of the rotations $R^{(q,k)}$ in lines~\ref{algline:sed:def_rotation}, \ref{algline:sed:def_full_rotation} and \ref{algline:sed:apply_rotation}, for any $q$ and $k$, all entries with a row index greater than $k$ in the $2q-1$-th column of $\Gamma^{(q,k)}$ are zero. If the above claim is true for $q-1$, the CM for $k=2q$ reads
\[
    \Gamma^{(q,2q)} = (\bigoplus_{i=1}^{q-1} J_2) \oplus \begin{pmatrix}
        0 & -x & 0 \\
        x & 0 & -v^\transpose \\
        0 & v & \tilde \Gamma^{(q)}
    \end{pmatrix}
\]
with a suitable number $x$ and vector $v$. Since all intermediate CMs correspond to some pure FGS, we get $|x|=1$, $v=0$ and that $\tilde \Gamma^{(q)}$ is the CM of a pure FGS from orthogonality of $\Gamma^{(q,2q)}$. From how the rotations $R^{(q,k)}$ are constructed, we get $x=1$, leaving the corresponding qubit in state $\ket{0}$. In the last step, i.e., for $q=n-1$, the $2\times 2$ CM $\tilde \Gamma^{(n-1)}$ necessarily corresponds to that of a computational basis state $\ket p$. Altogether, this shows the MGC corresponding to the output of Algorithm~\ref{alg:symmetric_euler_decomposition} generates the input state $\ket{\psi}$ when acting on $\ket{0^{n-1}p}$.

Finally, we explain how the output of Algorithm.~\ref{alg:symmetric_euler_decomposition} can be transformed into an RSF circuit. For each qubit~$q$, the algorithm outputs a (possibly empty) sequence of $r_q - 2q$ Givens rotations, corresponding to $l_q = \lceil r_q/2 - q\rceil$ matchgates $G^{(q,j)} = \exp(\ii \alpha_{q,j} Z_{q+j})\exp(\ii \beta_{q,j} X_{q+j-1}X_{q+j})$ arranged in a diagonal. Each of those matchgates have only one nonzero nonlocal parameter and one diagonal of such gates disentangles one qubit from the rest. One can, however, combine two such diagonals starting on qubits $(q, q+1)$ and $(q+1,q+2)$ into a single matchgate diagonal of length $\max\{1+l_q, l_{q+1}\}$. The gates in this diagonal are defined by combining the $Z$-rotation of $G^{(q, j-1)}$ and the $X\!\otimes\!X$-rotation of $G^{(q,j)}$ with the gate $G^{(q+1,j)}$, and by appending the remaining gates that have not been combined this way. This procedure is illustrated in Fig.~\ref{fig: sed}b. The circuit obtained by combining diagonals this way is in RSF.

In Appendix~\ref{app:sed_two_columns_version}, we present another symmetric Euler decomposition algorithm that does not require this last step for producing an RSF circuit. Instead, in each step matchgates are chosen whose action eliminates two columns of the CM at once. Finally, a more careful selection of which columns of the CM are eliminated gives yet another algorithm, where the nonlocal parameters of each matchgate are optimally utilized in non-generic cases. Using this third algorithm, one can prove that state preparation of FGS with matchgates gives asymptotically optimal gate count compared to arbitrary nearest neighbor gate sets (see Appendix~\ref{app:enhaced_decomposition}).

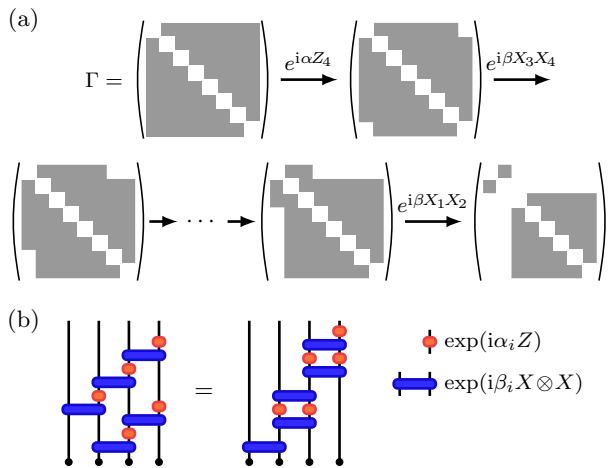
\begin{figure}[]
    \centering
    \begin{tikzpicture}[scale=0.8]
    \node at (0,5) {(a)};
    \begin{scope}[xshift=1.5cm, yshift=4.6cm, scale=1.18]
        \helpdrawSEDfigureAlgAction
    \end{scope}
    \node at (0,0) {(b)};
    \begin{scope}[xshift=0.75cm, yshift=-2.35cm]
        \helpdrawSEDfigureMerging
    \end{scope}
    \end{tikzpicture}
    \caption{(a) Action of the symmetric Euler decomposition algorithm for $q=1$ on the covariance matrix $\Gamma$ of a pure FGS on four qubits. Blank fields indicate zeros whereas grey fields mark nonzero entries. The applied Givens rotations correspond to the matchgates shown above the arrows.
    (b) The procedure to combine two diagonals of gates outputted by the symmetric Euler decomposition algorithm into a single diagonal of matchgates.}
    \label{fig: sed}
\end{figure}

\section{Applications of algebraic matchgates identities} \label{sec:application_lr_yb}

The formalism of representing FGS via circuits, and especially circuits in RSF, has several applications. In this section, we use the techniques and algorithms presented so far to prove that generic MGC in RSF have the optimal gate count for generating states. Secondly, using the GYB and LR relations, we present an algorithm that allows one to compute the phase sensitive inner product between two FGS. Finally, we present a generalization of the RSF and the inner product algorithm for states created by MGs and a finite number of non-gaussian operations.

\subsection{Optimality proof for MGCs in RSF} \label{subsec:rsf_optimal}

Both the absorption algorithm and the symmetric Euler decomposition (Algorithms~\ref{alg:absorption} and~\ref{alg:symmetric_euler_decomposition} respectively) independently show that any pure FGS can be represented by an MGC in RSF acting on a computational basis state. Moreover, if an FGS can be represented by a non-maximal RSF circuit, it can be represented by other RSF circuits as well, which can be seen, e.g., by adding identity gates. This raises the question of what is the minimal RSF circuit to represent a given state. Below, we answer this question, and in fact show that for any RSF layout, generic FGS generated by an MGC $U$ in the given layout cannot be represented with by any other MGC that has fewer gates than $U$.

Before stating the result, we explain what we mean by `generic FGS' in this context. Consider a state generated by an MGC in RSF. In order to claim that this circuit has minimal gate count, we clearly need to exclude some edge cases. First, we need all gates in the first layer to generate entangled states~\footnote{Due to parity preservation, if a matchgate does not generate entanglement when acting on a computational basis state, it simply applies a global phase and possibly flips the bits.}. Second, it must not be possible to combine gates in different diagonals (as can for instance be done if all gates in the circuit are of the form $\exp(\ii \alpha X\otimes X)$). Although this is a property of the full circuit rather than the individual gates, we show that it is sufficient to exclude gates that have only a single nonzero nonlocal parameter, e.g., gates of the form $U_1\otimes U_2 \exp(\ii \alpha X\otimes X) U_3 \otimes U_4$. Finally, we allow circuits to act on arbitrary computational basis states~\footnote{One can generate any computational basis state from any other one by acting with $\mathcal{O}(n)$ matchgates of the form $G(X,X) = X\otimes X$ arranged in two layers (possibly using an auxiliary qubit for preserving parity). Such circuits, however, do not satisfy the conditions of the theorem.}. Our statement is summarized in the following theorem.

\begin{theorem}
    \label{thm:RSF_optimal}
    Suppose the FGS $\ket{\psi} = U \ket{\cbs}$ is generated by an MGC $U$ in RSF, 
    \[
    U = D_1 \ldots D_m, \quad D_i = U_{il_i} \ldots U_{i 1},
    \]
    with parameters $((k_i, l_i))_{i=1}^m$, and a total of $N = \sum_{i} l_i$ gates. Suppose furthermore that each gate $U_{i1}$ in the first layer generate an entangled state when acting  on  $\ket{\cbs_{k_i}, \cbs_{k_i+1}}$, and the non-local parameters of the remaining gates are not integer multiples of $\pi/2$. Then, any other n.n. MGC which generates $\ket{\psi}$ by acting on any other computational basis state contains at least $N$ gates.
\end{theorem}

Our proof of Theorem~\ref{thm:RSF_optimal} is based on contradiction, and is contained in full detail in Appendix~\ref{app:proof_rsf_circuits_optimal}. Here, we briefly summarize the main idea. Suppose that $\ket{\psi}$ is generated by an MGC $V$ which has the fewest number of gates possible, i.e., is optimal. W.l.o.g., $V$ is in RSF with parameters $((h_1, r_1), \ldots, (h_p, r_p))$. If this were not the case, we can use the absorption algorithm to bring $V$ into RSF. Crucially, the number gates cannot decrease in this process, as otherwise we would contradict the assumption that $V$ has optimal gate count. To complete the proof, we will show that the RSF parameters of $U$ and $V$ have to coincide. At first, we only compare the locations $k_1$ and $h_1$, and lengths $l_1$ and $r_l$ of the first diagonals of $U$ and $V$. To this end, we make use of the restrictive nature of RSF layouts. In particular, the only gates affecting the first two columns of the CM are contained in the first diagonal. If they do not coincide, one can either show that $U$ creates some entanglement that $V$ cannot produce, or one can find another circuit $V'$ with even fewer gates, thereby contradiction the assumption. Thus $h_1 = k_1$ and $r_1 = l_1$. Next, we compare the action of each gate in the respective first diagonals. In order for the first columns of the CM to match, the gates in the first diagonal of $V$ have to coincide with those in $U$ apart from some local gates. The latter can be absorbed into the remaining circuit. When now undoing the first diagonals, one is presented with a similar instance of the problem with one fewer diagonal. Repeating this argument $m$ times then shows that $U$ and $V$ have the same gate count, and accordingly, that no other circuit can have less gates than $U$.

Note that Theorem~\ref{thm:RSF_optimal} is not an `if and only if' statement. In particular, there exist circuits that do not satisfy the second condition of Theorem~\ref{thm:RSF_optimal}, but still have optimal gate count for generating a state. Consider, e.g., the state $\ket{\psi}$ generated by a diagonal of $n-1$ gates of the form $\exp{\ii\alpha X\otimes X}$ acting on $\ket{0^n}$. While this gates do not satisfy the second condition, using fewer than $n-1$ gates necessarily implies that the resulting state is a product in some bipartition, and hence cannot be $\ket\psi$. For generic FGS, after the combining the diagonals outputted by Algorithm~\ref{alg:symmetric_euler_decomposition}, the resulting RSF circuit satisfy the conditions of Theorem~\ref{thm:RSF_optimal}. In the case of non-generic FGS, the enhanced version of the algorithm presented in Appendix~\ref{app:enhaced_decomposition} is specifically designed to utilize all the nonlocal degrees of freedom in each matchgate. As we show also in Appendix~\ref{app:enhaced_decomposition}, a construction based on this algorithm shows that matchgates circuits can have asymptotically minimal gate count w.r.t. to all other circuits composed of arbitrary nearest neighbor gates.

\subsection{Inner product algorithm} \label{subsec:inner_product_alg}

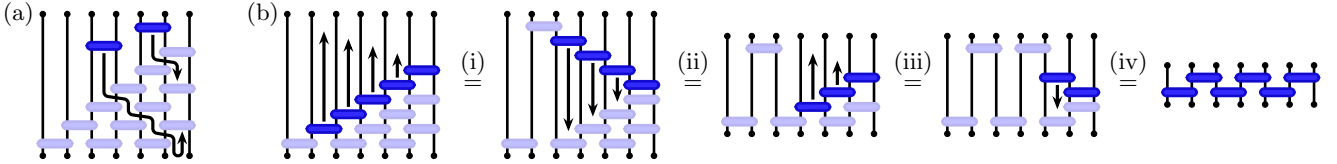
\begin{figure*}
    \begin{tikzpicture}[scale=0.65]
        \helpdrawInnerProduct

    \end{tikzpicture}
    
    \caption{(a) Absorption algorithm: The GYB and LR relations can be used to absorb gates into a matchgate circuit~\cite{MoLa25}. In this example, the circuit is and remains in the maximal RSF. (b) Inner product algorithm: We start with a matchgate circuit in RSF. In step (i), the routine \texttt{invert\_diagonal}, which sequentially applies the LR move from left to right (see main text), is applied to the highlighted gates, effectively flipping the corresponding structure. For step (ii), each but the last gate is absorbed into the resulting RSF circuit on two fewer qubits using the routine \texttt{to\_rsf}. Steps (iii) and (iv) are repetitions of the first two steps for smaller circuits. The inner product finally is evaluated using tensor network techniques. }
    \label{fig: inner product}
\end{figure*}

In the context of fermionic linear optics and matchgate circuits, efficient classical algorithms are known for computing the absolute value $\vert\braket{\psi}{\varphi}\vert$ of the inner product between two FGSs $\ket{\psi}$ and $\ket{\varphi}$~\cite{Loe55,TeDi02,BrGo17}. This can be expressed in a closed formula using the CMs of the two FGS, with a computational cost of $\mathcal{O}(n^3)$.  Extensions to phase-sensitive algorithms have also been developed~\cite{BrGo17,DiKo24,ReOs24,GuHa24}, which additionally require knowledge of the overlaps of the states with a reference FGS, since the CM alone does not encode the global phase~\footnote{Note that a representation of (unnormalized) fermionic Gaussian states via a Gaussian operator $\exp(\sum_{kl} c_k^\dagger h_{kl} c_l^\dagger)$, where $c_k^\dagger = (\majo{2k-1}-i\majo{2k})/2$ are the creation operators, acting on the vacuum state $\ket{0\ldots 0}$ allows to compute amplitudes in the computational basis~\cite{BeSo17} or other product bases~\cite{BaRe24,RaMi25}. One needs to take care, however, to preserve the global phase when transforming any other representation of FGSs to the above one.}.

Here, we present a general procedure to compute the overlap of two arbitrary FGSs without relying on their CMs. Instead, we assume to have knowledge of the matchgate circuits generating the two states. Exploiting this, we describe an efficient procedure using the GYB and LR properties that reduces the computation of the inner product to the contraction of a circuit of depth two, which can be efficiently computed.

Consider two FGSs and their generating MGC, $\ket{\psi} = U \ket{\cbs_1}$ and $\ket{\varphi}=V\ket{\cbs_2}$. The overlap is given by
\begin{equation}
    \braket{\varphi}{\psi} = \bra{\cbs_2}V^\dagger U\ket{\cbs_1}.
\end{equation}
The goal of the algorithm is to simplify the circuit $U^\dagger V$ by using a sequence of GYB and LR moves. Note that in this setting, in contrast with the absorption algorithm, the LR move can be applied both with the original circuit and with its adjoint, since the circuit acts on a computational basis state on both sides. In the following we first define several subroutines and then provide a description of the algorithm.

The first subroutine, \texttt{invert\_diagonal}, takes as an input a diagonal of gates $D=U_l \ldots U_1$ and the computational basis state $\bra{\cbs}$, and output an `inverted' diagonal $D' = V_1 \ldots V_l$, where the first gate $V_l$ acts on the rightmost qubits and $V_1$ acts on the leftmost ones, s.t. $\bra\cbs D = \bra \cbs D'$ (see Fig.~\ref{fig: inner product}b step~(i)). It is obtained by repeatedly using the LR move. To be concrete, the gates are computed sequentially via $U'_l = U_l$,  $\bra\cbs U'_{k} U_{k-1} =\bra\cbs U'_{k-1} V_{k}$ (LR move) for $k = l, \ldots,2$, and $V_1 = U'_1$. Another subroutine we define here is called \texttt{to\_rsf}. It takes as an input a MGC $U$ and a computational basis state $\ket\cbs$ and produces as an output a MGC $U_\text{RSF}$ in right standard form s.t. $U\ket\cbs = U_\text{RSF}\ket\cbs$. This is done by one application of the absorption algorithm (Alg.~\ref{alg:absorption}) for each gate in $U$. W.l.o.g., we assume the output of \texttt{to\_rsf} to be in the maximal RSF. Furthermore, we define three simple routines for decomposing RSF circuits and diagonals. The first one, \texttt{split\_rsf}, takes as an input an MGC $U = D_1 \ldots D_m$ in the maximal RSF, and outputs the first diagonal $D_1$ and $V = D_2\ldots D_m$. The second subroutine, \texttt{split\_first}, takes as an input a diagonal $D$, and outputs its first gate $A$ and the remaining gates $\tilde D$ s.t. $D = \tilde D A$. The last one, \texttt{split\_last}, similarly outputs the last gate $B$, and the remaining gates $\tilde F$, of an inverted diagonal $F$, i.e., $F = B \tilde F$. The inner product algorithm now repeatedly uses these subroutines.

\begin{algorithm}[H]
\caption{Inner product algorithm}
\label{alg:inner_product}
\KwIn{
    Computational basis states $\ket{\cbs_1}$ and $\ket{\cbs_2}$ and MGCs $U$ and $V$.
}
\KwOut{
    The matrix element $\bra{\cbs_2}V^\dagger U\ket{\cbs_1}$.
}

$U_1 \gets \mathtt{to\_rsf}(V^\dagger U)$

\For{$k =1\ldots \ceil{n/2}$}{
    $D_k, V_k \gets \mathtt{split\_rsf}(U_k)$

    $A_k, \tilde D_k \gets \mathtt{split\_first}(D_k)$
    
    $F_k \gets \mathtt{invert\_diagonal}(D_k, \ket{\cbs_2})$

    $B_k, \tilde F_k \gets \mathtt{split\_last}(F_k)$
    
    $U_{k+1} \gets \mathtt{to\_rsf}(\tilde F_k V_k, \ket{\cbs_1})$

}

$v \gets \bra{\cbs_2} \prod_i B_i \prod_i A_i \ket{\cbs_1}$ \label{algline:ipa:contraction}

\Return{$v$}  

\end{algorithm}
Fig.~\ref{fig: inner product} shows a diagrammatic representation of the application of this algorithm. The expression in line~\ref{algline:ipa:contraction} coincides with the desired matrix element $\bra{\cbs_2}V^\dagger U\ket{\cbs_1}$. Note that all the gates $A_k$ act on different qubits (the same is true for the $B_k$ gates). The computation of $v$ can therefore be viewed as an MPS contraction~\cite{FaNa92, OsSt95, RoOs97, HaPo18} between two MPSs of bond dimension at most~2. Alternatively, the following procedure can be used: We consider each $A_i$ and $B_i$ as a rank-4 tensor, with two input and two output indices. We start from the left boundary by contracting the leftmost open leg of the output state with the first $A$ gate and the corresponding input state for the first two qubits, obtaining a rank-1 tensor. This tensor is then contracted with the following $B$ gate and the corresponding output leg, producing again a rank-1 tensor. The procedure is repeated recursively, alternating $A$ and $B$ contractions, until reaching the right boundary, where the final contraction yields the scalar value of the overlap. Each intermediate tensor has at most one or two qubit indices, so the computational cost of each contraction is constant, resulting in an overall runtime that scales linearly with $n$. The value obtained is therefore exact up to numerical precision and requires only constant memory.

The complexity of the inner product algorithm is $\mathcal{O}(kn + n^3)$, where $k$ is the number of gates in the initial circuit $V^\dagger U$. The first contribution arises from transforming the initial circuit to RSF, and the second comes from the repeated applications of the absorption algorithm in order to keep the standard form in step 3.

The method presented in this section can also be applied to compute expectation values of Pauli operators. For any $P_1,P_2 \in\{\id, X,Y,Z\}$, the tensor product $P_1\otimes P_2$ is a matchgate (or matchgate times $X$, see Sec.~\ref{sec:preliminaries}). Therefore, an arbitrary Pauli operator $P$ can be decomposed as a single layer of matchgates. One can then use the algorithm to express the expectation value $\bra{\cbs} U^\dagger P U \ket{\cbs}$ as a simple depth-2 circuit, computable via a tensor contraction. Furthermore, using the techniques developed in Ref.~\cite{BrGo22} gives an algorithm to perform weak simulation, i.e., to sample from the probability distribution $p(x) = \vert \bra{x} U \ket{0^n} \vert^2$. Finally, note that the inner product algorithm, as well as other circuit rewriting methods, are not restricted to matchgates, but can be applied to all gate groups satisfying the GYB and LR properties (see Appendix~\ref{app:other_GYB_LR}).

\subsection{Generalization to $t$-doped matchgate circuits} \label{subsec:generalization_doped}

In this section we use the algebraic matchgate identities in order to define a standard form for circuits containing matchgates and a few parity-preserving resourceful gates, such as SWAP or controlled phase gates~\cite{JoMi08, BrGa11}. We will show that any $t$-doped MGC can be compressed into a generalized triangle layout, similar to that in Refs.~\cite{CaKo22,KoCa22}, where the circuit is composed by at most $n(n-1)/2+t(2n-3)$ MGs and $t$ resourceful gates as depicted in Fig.~\ref{fig:tdoped}a. This generic form can be derived from the GYB relation only. Then, using also the LR move, we present a generalized RSF for a $t$-doped FGS. In both cases, we obtain a circuit depth of $\mathcal{O}(n+t)$. Finally, we generalize the inner product algorithm from the previous section to also compute the inner product of non-Gaussian states, specifically those generated by matchgates and a total of $t$ resourceful gates. This approach reduces the computation to a tensor contraction that evaluates to the inner product. The cost of performing such a contraction grows exponentially with the number of resourceful gates. While alternative approaches based on decomposing the resourceful gates into linear combinations of matchgates exhibit a more favorable exponential scaling~\cite{DiKo24,ReOs24}, our algorithm provides a complementary framework in which approximate contractions using standard MPS techniques may offer practical advantages. 

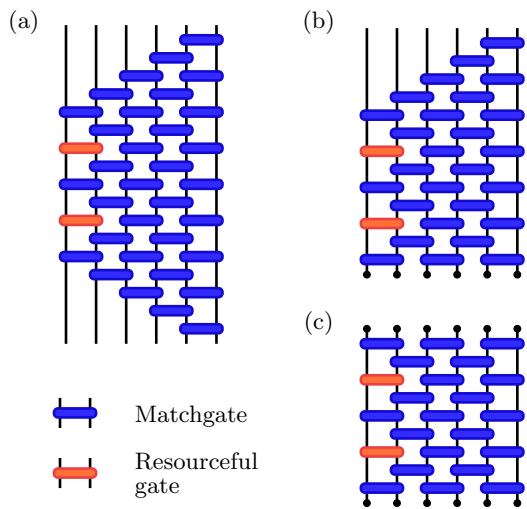
\begin{figure}
    \begin{tikzpicture}[scale=0.8]
    \helpdrawtdoped
    \end{tikzpicture}
    \caption{
    Standard forms for $t$-doped (a) matchgate circuits, (b) states, and (c) inner products. The standard form for the circuits are obtained by applying the GYB relation to general $t$-doped matchgate circuits (i.e., matchgate circuits interleaved with $t$ resourceful non-matchgates). For states, the depth is further reduced by the application of the LR move in a similar way as for the absoption algorithm. Finally, the inner product circuit is found by also using the inverted LR move with the other product state.
    }
    \label{fig:tdoped}
\end{figure}

We start considering $t$-doped matchgate circuits~\cite{MeHe25}, i.e., circuits composed by matchgates and $t$ parity preserving resourceful gates, such as SWAP or controlled phase gates \footnote{More generally, we can consider unitaries generated by an exponential of $\kappa = 4$ Majorana operators. Further generalization of our constructions to arbitary even $\kappa = \mathcal{O}(1)$ is likewise possible.}. A generic $t$-doped matchgate circuit can be written as \[G_t W_t G_{t-1}\dots W_{1} G_0,\] where $G_i$ are generic MGCs and $W_i$ are the resourceful gates. It is known that via a sequence of fermionic swap gates $G(Z,X)$, such parity preserving gates can be moved to act on any other position in the circuit~\cite{JoMi08,MeHe25}. W.l.o.g., we can assume the $W$ gates to act exclusively on the first two qubits of the chain. 

The most general MGC can be written as a series of diagonals~\cite{CaKo22,KoCa22,OsDa22}: $G_i=D^{(i)}_{1}\dots D^{(i)}_{n-1}$, where $D^{(i)}_k = \prod_{j=k}^{n-1} U_{j,j+1}$. Consider the case with a single $W$ gate, described in general by the circuit $G_1W_1G_0$. Note that the gates of $G_1$ comprising the diagonals $D^{(1)}_{3},\dots,D^{(1)}_{n-1}$ do not act on qubits 1 and 2, which are the ones affected by the $W_1$ gate. Thus, they can be commuted and acted directly in $G_0$. But $G_0$ is already the most general MGC, so all gates can be absorbed. We can thus write $D^{(1)}_{1}D^{(1)}_{2}W_1G'_0$ as the most general circuit with a single resourceful gate. The same reasoning can be applied recursively: Each $W_i$ gate is conjugated by two diagonals of matchgates, $D^{(i)}_{1}$ and $D^{(i)}_{2}$. A graphical representation of this structure is shown in Fig.~\ref{fig:tdoped}a. This construction implies that at most $n(n-1)/2 + t(2n-3)$ matchgates (plus the $t$ resourceful gates) are sufficient to construct any $t$-doped matchgate circuit.

A $t$-doped FGS can be similarly written in circuit form. In particular, it can be created by applying a generic $t$-doped MGC onto a computational basis state. However, now the LR identity can be applied to the MGs acting on the product state. This immediately implies, by virtue of the absorption algorithm in Sec.~\ref{subsec:absorption}, that the first generic MGC acting on the state can be brought to RSF. Therefore, a $t$-doped FGS can be represented as a circuit containing $\floor{n^2/4} + t(2n-3)$ matchgates. An example of such construction is shown in Fig.~\ref{fig:tdoped}b.

Following similar arguments, one can compute the overlap between $t_1$-doped FGS $\ket{\psi}=U\ket{0^n}$ and a $t_2$-doped FGS $\ket\phi = V\ket{0^n}$, with $t_1+t_2=t$. The idea is to rewrite the state $U^\dagger V\ket{0^n}$ into the previously described general form for a $t$-doped FGS. Then, a straightforward generalization of the inner product algorithm leads to a simplified brickwall circuit with depth $4t+1$. An example of such construction is shown in Fig.~\ref{fig:tdoped}c. This overlap can be computed as a tensor contraction, with a contraction complexity scaling as $2^{\min(n,\, 4t+1)}$, corresponding to the maximum number of open qubit indices. Observe that this is the complexity for the exact contraction. However, TEBD-like algorithms~\cite{Vi04} may further improve the performance by truncating small singular values.

\section{Shallow matchgate circuits} \label{sec:shallow}

So far, we have described algorithms that are applicable to arbitrary pure FGS. These algorithms serve as a simpler complement to existing methods to represent generic FGS with circuits of depth~$\mathcal{O}(n)$~\cite{JiSu18,KiMc18,DaDe19,OsDa22}. Moreover, we have demonstrated that the circuits we consider attain optimal gate count in generic and specialized settings. However, besides being bounded by $n$, we do not get any guarantees on the depth of the circuits produced.

In this section, we restrict the class of FGS to those which can be represented with circuits of low depth. More specifically, we consider MGCs of depth $\depth$, where $\depth$ can be an arbitrary function of $n$. First, we give a characterization of those states in terms of their CM. Second, we describe another algorithm that produces a generating MGC when given the CM. It is tailored to perform well on states with a low~$\depth$, and is also applicable in an approximate setting, as we show numerically in section~\ref{sec:numerical_examples}. To derive this algorithm, we prove a tripartite version of the result of Botero and Reznik~\cite{BoRe04}, which is applicable in the cases of interest.

\subsection{Characterization of bounded depth circuits}

Here, we show that pure FGS that have short-range correlations can be described exactly with a short-depth MGC. A matrix $M$ is said to be \emph{$\bdns$-banded}, if $M_{kl} = 0$ for $|k-l| > \bdns$, i.e., all non-zero entries are among the first $\bdns$ minor diagonals. Our result is summarized in the following theorem.

\begin{theorem} \label{thm:banded_cm_shallow_mgc}
Suppose $\ket{\psi}$ is a pure FGS with a $\bdns$-banded CM. Then, $\ket{\psi}$ can be generated by a matchgate circuit of depth $\mathcal{O}(\bdns)$. Conversely, the CM of any FGS generated by a circuit of depth $\depth$ is clearly at most $\mathcal{O}(\depth)$-banded.
\end{theorem}

Proving the forward direction of Theorem~\ref{thm:banded_cm_shallow_mgc} is straightforward if one follows the steps taken by the symmetric Euler decomposition (Algorithm~\ref{alg:symmetric_euler_decomposition}). Since the CM is $\bdns$-banded, a diagonal of length no more that $\bdns/2$ is sufficient to disentangle the first qubit. A careful examination then shows that the action of this diagonal leaves the CM $\bdns$-banded, thus allowing the argument to be applied recursively to the remaining qubits. The converse direction of Theorem~\ref{thm:banded_cm_shallow_mgc} follows from an argument based on the properties of lightcones. All details are presented in Appendix~\ref{app:proof_shallow_banded}. We remark that the latter argument also shows that MGCs attain optimal depth to generate FGS. That is, no other family of gates could generate FGS with asymptotically lower circuit depth.

Let us emphasize again that we do not require any particular dependence between the circuit depth $\depth$, and the system size $n$. Thus, our result holds in the generic case of a linear circuit depth~\cite{JiSu18,KiMc18,DaDe19,OsDa22}, as well as commonly encountered cases of a constant $d$, or $d= \mathcal{O}(\log n)$~\cite{MaSt24}. However, as we explain here, the generic bound on the circuit depths can be undercut in some more exotic cases as well. For instance, consider a CM corresponding to a state on a two-dimensional lattice of fermionic modes of size $L^2$, such that correlations vanish after a distance $k$. Suppose that the first row of modes is labeled by $1,\ldots,L$, the second one by $L+1,\ldots,2L$, and so on. When reinterpreted as a one-dimensional system with qubits labeled $1,\ldots, L^2$, the CM is $\mathcal{O}(kL)$-banded, and we get the existence of a circuit of depth $\mathcal{O}(kL)$ as opposed to a circuit of depth $\mathcal{O}(L^2)$ when using the standard approach.

\subsection{Entanglement cutting algorithm for disentangling states with banded CM}
\label{subsec:cutting_alg}

\begin{figure}[]
    \centering
    \begin{tikzpicture}[scale=0.8]
    \node at (0,1.7) {(a)};
    \begin{scope}[xshift=0.5cm, yshift=0.cm, scale=0.7]
        \helpdrawCuttingFigureTheorem
    \end{scope}
    \node at (0,-0.3) {(b)};
    \begin{scope}[xshift=0cm, yshift=-1.75cm, scale=0.7]
        \helpdrawCuttingFigureAlg
    \end{scope}
    \end{tikzpicture}
    \caption{(a) Illustration of Observation~\ref{observation:threepartite_bore}. In case party~$\partyB$ comprises sufficiently many qubits, any FGS with $\bdns$-banded CM is equivalent to some entangled pairs shared between parties $\partyA$ and $\partyB$, and $\partyB$ and $\partyC$ up to matchgate circuits acting on the parties.
    (b) Illustration of the algorithm to determine the circuit $U_\partyB$. The matchgate circuit $U_\text{diag}$ is chosen to diagonalize the reduced CM on party $\partyB$. The second circuit, $U_\text{perm}$, permutes the entangled qubits such each qubit is adjacent to the party to which it is entangled}
    \label{fig:ent_cutting}
\end{figure}
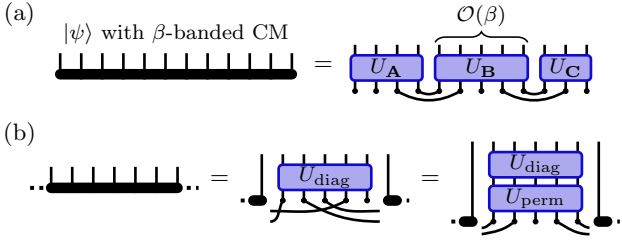

In the symmetric Euler decomposition (Algorithm~\ref{alg:symmetric_euler_decomposition}) and its variants, matrix entries become zero, both due to the explicit construction of the Givens rotations and due to implicit effects. With finite numerical precision, however, such zeros are not always detected reliably. In particular, checks within the algorithms that test whether a matrix element is zero may instead register a small non-zero value even when the entry should theoretically vanish. This typically results in the placement of more gates than necessary~\footnote{This effect can be observed, for instance, in states generated by applying a few layers of random matchgates to a product state.}, thereby preventing the algorithm from correctly identifying circuits of minimal depth. Furthermore, these algorithms do not perform well on CMs that are only approximately banded (see Section~\ref{sec:approximately_banded}).

To deal with this issue, we introduce here another algorithm that performs better for states with banded CMs. The idea is to find a small matchgate circuit acting on a few consecutive qubits for a given qubit index $k$, which splits the state into a product of qubits $1,\ldots k,$ and qubits $k+1,\ldots, n$.

If the CM of the state is $\bdns$-banded, a circuit of depth $\depth=\lceil\frac{\bdns+1}{2}\rceil$ exists to generate the state when acting on some product state. Thus, taking the inverse of the gates around qubit $k$ proves the existence of such a circuit. The problem that remains is determining this local circuit if one does not have knowledge of the depth-$\depth$ circuit that generates $\ket \psi$. From the following observation and its proof, we can derive an algorithm to solve this problem.

\begin{observation} \label{observation:threepartite_bore}
    Consider a pure FGS $\ket{\psi}$ on $n$ qubits with a $\bdns$-banded CM shared between three parties $\partyA, \partyB,\partyC$, each holding a contiguous block of qubits.
    Then, if the size of the party in the center, $\partyB$, satisfies $\vert\partyB\vert \geq \bdns + 2$, there exist MGCs $U_\partyA$, $U_\partyB$, and $U_\partyC$ acting on $\partyA, \partyB, \partyC$, such that 
    \[
        U_\partyA\otimes U_\partyB \otimes U_\partyC \ket{\psi} = \ket{\cbs_\partyA}\ket{\psi_{\partyA\partyB}}\ket{\cbs_\partyB}\ket{\psi_{\partyB\partyC}}\ket{\cbs_\partyC},
    \] where $\ket{\cbs_i}$ are computational basis states on a subset of qubits of the respective systems, and the state of the remaining qubits is given by $\ket{\psi_{ij}}$, which are products of entangled pairs shared between the respective systems.
\end{observation}

We illustrate Observation~\ref{observation:threepartite_bore} in Fig.~\ref{fig:ent_cutting}a.

\begin{proof}
    Due to Theorem~\ref{thm:banded_cm_shallow_mgc}, a MGC $U$ of depth $\depth =\lceil\frac{\bdns+1}{2}\rceil$ with $\ket{\psi} = U \ket{\cbs}$ exists. Now decompose $U = VW$ into two MGCs $V$ and $W$ such that $W$ contains all the gates in the backward lightcones \footnote{The backward lightcone for a splitting $1,\ldots,k\vert k+1,\ldots,n$ can be defined by two conditions: (i) The last gate acting on qubits $k$ and $k+1$ is in the lightcone, and (ii) whenever for two two-qubit gates $G$, $G'$, one of the structures $G G'$, $(G\otimes \id_2)(\id_2 \otimes G')$, or $(\id_2\otimes G)(G' \otimes \id_2)$ is present in the circuit, if $G$ is in the lightcone, then $G'$ is also in the lightcone.} at the boundary between $\partyA$ and $\partyB$, and $\partyB$ and $\partyC$, and $V$ contains the remaining gates outside the lightcones. Tracing the backwards lightcones in the bipartitions $\partyA\vert\partyB\partyC$ and $\partyA\partyB\vert\partyC$ reveales that the gate in $W$ act on two non-overlapping regions.
    For the same reason, we can decompose $V =V_\partyA\otimes V_\partyB\otimes V_\partyC$ into circuits acting only on the three respective parties. When acting only with $W$, one gets
\[
W\ket{\cbs} = \ket{\cbs_\partyA} \ket{\phi_{\partyA\partyB}}\ket{\cbs_\partyB}\ket{\phi_{\partyB\partyC}}\ket{\cbs_\partyC} = V_\partyA^\dagger \otimes V_\partyB^\dagger \otimes V_\partyC^\dagger \ket{\psi}
\]
where $\ket{\phi_{\partyA\partyB}}$ and $\ket{\phi_{\partyB\partyC}}$ are some FGS shared between the respective systems. Using the result of Botero and Reznik~\cite{BoRe04} for both these states,  one can find MGCs $V'_\partyA$ and $V'_\partyB$ ($V''_\partyB$ and $V''_\partyC$), s.t. $V'_\partyA \otimes V'_\partyB \ket{\phi_{\partyA\partyB}} = \ket{\psi_{\partyA\partyB}}$ ($V''_\partyB \otimes V''_\partyC \ket{\phi_{\partyB\partyC}} = \ket{\psi_{\partyB\partyC}}$), with $\ket{\psi_{\partyA\partyB}}$ ($\ket{\psi_{\partyB\partyC}}$) being products of entangled pairs. The circuits proving the observation are thus given by $U_\partyA = (\id \otimes V'_\partyA) V_\partyA^\dagger$, $U_\partyB = (V'_\partyB \otimes\id \otimes V''_\partyB) V_\partyB^\dagger$, and $U_\partyC = ( V''_\partyC\otimes \id) V_\partyC^\dagger$.
\end{proof}

Let us now describe how to determine the circuit $U_\partyB$ as described in this observation. Clearly, the circuits $U_\partyA,U_\partyB$ and $U_\partyC$ diagonalize the respective CMs. For simplicity, assume first that the Williamson eigenvalues of the reduced CM on $\partyB$ are non-degenerate (we deal with the degenerate case in Appendix~\ref{app:cutting_alg_degenerate_exact}). Let now $\tilde U_\partyB$ be any MGC whose action transforms $\Gamma_\partyB$ to its Williamson normal form. Since the Williamson normal form is unique up to permutations and the symmetry of $\covblock$~\cite{HoJo13}, $\tilde U_\partyB$  must coincide with $U_\partyB$ up to a circuit $S_\partyB$ composed solely of fermionic swap gates and local phase gates. Thus, the only remaining task is to find $S_\partyB$. After applying $\tilde U_\partyB$, the CM is given by $\tilde \Gamma$, and the reduced CM on $\partyB$ by $\bigoplus_i \lambda_i \covblock$. Each Williamson eigenvalue $\lambda_i$ with $|\lambda_i|<1$ corresponds to an entangled pair of the form $\cos\alpha_i\, \ket{00} - \sin\alpha_i\, \ket{11}$, which is shared either with $\partyA$ or with $\partyC$. Let us now focus on one eigenvalue, say $\lambda_i$. To determine to which party the corresponding qubit $i$ in $\partyB$ is entangled to, one inspects the off-diagonal block $\tilde \Gamma_{\partyA\partyB}$ of $\tilde \Gamma$ with row indices corresponding to $\partyB$ and column indices corresponding to $\partyA$, as this block contains all correlations between $\partyA$ and $\partyB$ (one could equally well use the block corresponding to $\partyC$). If the rows of $\tilde \Gamma_{\partyA\partyB}$ with indices $2i-1$ and $2i$ are zero, the entangled pair is shared between $\partyB$ and $\partyC$; otherwise, it is shared between $\partyB$ and  $\partyA$. Qubit $i$ in $\partyB$ then needs to be moved adjacent to the party to which it is entangled. The same is repeated for all other entangled qubits in $\partyB$. An illustration of the whole procedure is given in Fig.~\ref{fig:ent_cutting}b. 

We briefly comment on the case of degenerate eigenvalues, where the above procedure breaks down. One way to deal with this is to lift the degeneracy with high probability by first applying random MGCs on the boundaries between the systems. However, degeneracies can also be directly addressed with some additional steps, as we explain in Appendix~\ref{app:cutting_alg_degenerate_exact}.

Let us finish this section by condensing the results into an algorithm that disentangles the state $\ket{\psi}$, so that its inverse can be used to generate $\ket{\psi}$. Firstly, we determine $\bdns$ such that the CM is $\bdns$-banded. Let now $s = \bdns+2$. Next, we group the $n$ qubits into $N = \floor{n/s}$ consecutive blocks $\partyB_i$ of consecutive qubits with sizes $s_i = s$ (or $s_i  = s+1$, with the aim of distributing the remaining $n - Ns$ qubits). For each block $\partyB_i$, we compute a matchgate unitary $U_i$ that diagonalizes the reduced CM and rearranges the qubits within $\partyB_i$ so that the leftmost qubits are precisely those entangled with $\partyB_1\ldots\partyB_{i-1}$, and the rightmost qubits are those entangled with $\partyB_{i+1}\ldots\partyB_N$. In the second step, we determine and apply the local MGCs $V_j$ that diagonalize the CM restricted to the entangled qubits lying across consecutive blocks. After this procedure, the global state is fully disentangled. The depth required to implement any of the circuit $U_i$ is given by $s_i\leq s+1$ (see, e.g., Refs.~\cite{OsDa22,CaKo22} or Appendix~\ref{appendix:euler_decomp_orth}). Note that at most $s/2 + 1$ entangled pairs can be shared between consecutive blocks. Hence, a worst case estimate for the circuit depth of any of the $V_j$ is given by $s+2$. Hence, an upper bound on the total circuit depth obtained this way is given by $2s+3$.

\subsection{Approximate disentangling algorithm} \label{sec:approximately_banded}

Many physically relevant FGS, such as some ground states of local Hamiltonians, have a CM that is approximately banded. In this context, approximately banded means that the largest matrix elements are concentrated near the diagonal, while the entries $\Gamma_{ij}$ sufficiently far away from the diagonal decay, e.g., exponentially with the distance $|i-j|$. For such states, a modified version of the cutting algorithm continues to perform well, producing circuits that approximately generate the desired state, as we demonstrate numerically in the next section. In the following, we describe the required modifications to the algorithm.

An important conceptual difference is that the parameter governing the circuit depth, $\bdns$, is not determined anymore by the CM, but freely selectable. One can hence try different values of $s$ (as before, the qubits are partitioned into blocks with sizes $s$ and $s+1$), and then choose the $s$ that gives the best tradeoff between circuit depth and achieved fidelity. As before, for a given $s$ one identifies small MGCs acting on each block that disentangle, as much as possible, the qubits preceding the block from those following it. Applying those approximately factorizes the state over several regions. In a second step, circuits are determined that disentangle the states in those regions. Contrary to the algorithm in the exact case, for a given $s$ we do not get a guarantee that two consecutive blocks share at most $s/2+1$ entangled pairs~\footnote{Here, there is generally no well defined notion of entangled pairs anymore.}. One of the circuits in the second step may hence act on two entire consecutive blocks. The total circuit depth can hence be upper bounded by $3s+5$.

Note that the algorithm to find the circuit acting on a block, say $\partyB$, that disentangles the state between neighboring regions $\partyA$ and $\partyC$ as described in the previous section would not work in case the CM is not exactly banded. Another algorithm is thus needed that approximately performs the same task and remains robust when the CM deviates from an exactly banded form. There are, of course, many distinct approaches one could take to obtain an approximate cutting algorithm. As an example, one may choose to find via numerical optimization a circuit that minimizes the entanglement entropy in the corresponding bipartition. Our approach instead is to devise an algorithm that determines the exact circuit in the case the CM is exactly banded. This can be achieved by slightly modifying the cutting algorithm described in Sec.~\ref{subsec:cutting_alg}. Once the CM of $\partyB$ is diagonalized, one can encounter two different situations. In the case where all Williamson eigenvalues are non-degenerate, a straightforward modification gives a good result. The steps are as follows. First, apply the MGC $\tilde U_\partyB$ that brings the reduced CM of $\partyB$ to Williamson normal form. Then, for each eigenvalue $\lambda_i$ with $|\lambda_i - 1| > \varepsilon_\lambda$ (where $\varepsilon_\lambda$ is a freely selectable cutoff parameter) investigate the blocks $\tilde \Gamma_{\partyA\partyB}$ and $\tilde \Gamma_{\partyB\partyC}$ as described in the previous section. For both blocks, compute the norm of the block obtained when selecting only the rows with indices $2i-1,2i$. Swap the corresponding qubit next to the subsystem with greater corresponding norm. The case of degenerate eigenvalues can for instance be addressed by first applying a random MGC $U'$ of depth $d'$, and then applying this algorithm for the resulting state (note that $U'$ then needs to be concatenated to the output circuit). We describe in Appendix~\ref{app:cutting_alg_degenerate_approx} an alternative procedure that does not require such an additional circuit.

\subsection{Numerical examples} \label{sec:numerical_examples}

Now, we look at some example FGSs where the CM is approximately banded. We apply our algorithm and show how the fidelity improves with the allowed depth.

We consider a Majorana chain model with Hamiltonian
\begin{equation}
    H_1 = -\sum_{i=1}^{n-1}\majo{2i}\majo{2i+1}-g\sum_{i=1}^{n} \majo{2i-1}\majo{2i},
\end{equation}
which corresponds to the Ising model on qubits and has a phase transition between an ordered and a disordered phase, with a critical point at $g_c=1$~\cite{Pf70}. In what follows, we focus on the regime $g > 1$, where the ground state is non-degenerate and its CM is approximately banded, making it a suitable test case for our method~\footnote{For $g<1$, the CM formalism yields as a ground state an FGS with long-range correlations, for which our algorithm is not well-suited.}. 

\begin{figure}[t!]
    \centering
    \includegraphics[width=1\columnwidth]{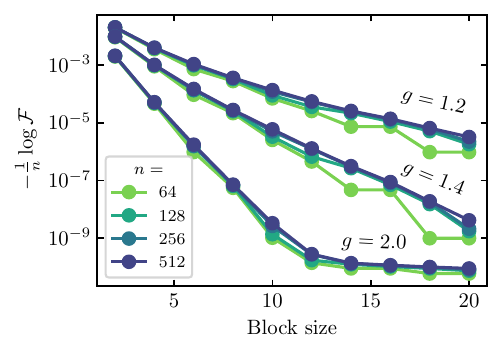}
    \caption{Normalized logarithm of the fidelity of the reconstructed state for the Ising model for different system sizes $n$ and magnetic fields $g$ in the disordered phase ($g>1$, with the critical point located at $g_c=1$), as a function of the block size selected in the entanglement cutting algorithm (proportional to the depth of the MGC required to create the state).}
    \label{fig: approximate algorithm}
\end{figure}

In our simulations, using standard techniques~\cite{SuTa20}, we compute the CM of the ground state $\ket{\psi_0}$ of the Ising model for different parameters of the field in the disordered phase $g>1$. Then, we use the approximate version of the entanglement cutting algorithm with various choices of the block size $s$ in order to find a circuit of depth at most $3s +5$ that creates an state approximating the ground state, $\ket{\psi^{\text{app}}_0}$. We then compute the fidelity between the original state and the approximate one, defined as
\begin{equation}
    \mathcal{F} = |\braket{\psi_0}{\psi_0^{\text{app}}}|^2.
\end{equation}

As discussed in more detail in Sec.~\ref{sec:approximately_banded}, there are several distinct approaches to implement an approximate cutting algorithm. Here, we give the numerical results for the algorithm yielding the best fidelities with the exact ground state, which is described in detail in Appendix~\ref{app:cutting_alg_degenerate_approx}. In this implementation, we need to set a cutoff $\varepsilon_\lambda$ of the Williamson eigenvalues such that any eigenvalue with $|1-\lambda|<\varepsilon_\lambda$ is considered to be numerically 1. In our simulations we use the value $\varepsilon_\lambda=10^{-8}$. Furthermore, in order to handle the (approximate) degeneracies of the Williamson eigenvalues, we choose a parameter $\varepsilon_{\text{deg}}$ such that if two Williamson eigenvalues satisfy $|\lambda_1-\lambda_2|<\varepsilon_\text{deg}$ then they are considered to be degenerate. We set this value to $\varepsilon_{\text{deg}}=10^{-2}$. In Appendix~\ref{app:cutting_alg_degenerate_approx} we describe also a method that does not require such a parameter. However, for identical block size choices, in this example the first method typically yields lower fidelities.

In Fig.~\ref{fig: approximate algorithm} we show the results obtained for the normalized logarithm of the fidelity $\mathcal{F}$ as a function of the block size (note that $-\log\mathcal{F} \approx \mathcal{I}$ for small values of the infidelity $\mathcal{I} = 1- \mathcal{F}$). We observe that, as expected, increasing the block size (and therefore the depth) yields a better fidelity. Notice that for $g=2$ and a large enough block size, the fidelity reaches a plateau, meaning that increasing the block size does not improve the results anymore. This effect corresponds to reaching the numerical precision of the computer. We also observe that when approaching the critical point at $g=1$, where there are long range correlations, the block size required to get a good approximation increases. The characterization of finite block-size scaling at the critical point, analogous to the finite-depth scaling considered in Ref.~\cite{JoSm22}, is left for future work.

Finally, let us note that in Fig.~\ref{fig: approximate algorithm} we find approximately a linear decrease (in a log scale) of the log-fidelity versus the block size. This behavior implies that, in order to achieve a fixed target fidelity $\mathcal F$, the required block size must scale proportionally to $\log n$, with a proportionality constant determined by the inverse of the slope.

The Ising model in the disordered phase has exponentially decaying correlations. However, this is not a necessary condition for our algorithm to be useful. In Appendix~\ref{app:further_numerical_examples}, we present some further numerical examples. First, we study a long-range Kitaev model~\cite{VoLe14,VoLe16}, whose ground state exhibits algebraically decaying correlations. Secondly, we also consider a model defined on a two-dimensional fermionic lattice, and prepare its ground state on a one-dimensional chain. We numerically find that also in these examples the necessary circuit depths for achieving a given fidelity scale sublinearly with the system size.

\section{Conclusion}

In this paper we have studied the representation of fermionic Gaussian states via the matchgate circuits generating them. We introduced two algorithmic approaches for constructing matchgate circuits that prepare a target FGS from its covariance matrix. For the first type, the symmetric Euler decompositions, when viewed as a disentangling algorithm, we showed that in each step, the maximal possible amount of entanglement is removed. Hence, those circuits have optimal gate count scaling compared to any other (even non-matchgate) circuit. The entanglement cutting algorithm, on the other hand, is better suited for states with short-range correlation, and can be shown to work due to the special entanglement structure of these FGS (Observation~\ref{observation:threepartite_bore}). Moreover, it still can be applied in an approximate setting. We showcased the performance of the algorithm by finding circuits that approximate well the ground states of some physical models.

We have also considered the problem of the circuit complexity of FGSs, i.e., the minimal number of matchgates required to prepare a given state. Under mild assumptions on the gates, we proved that matchgate circuits in RSF achieve optimal gate count.

Finally, we used the algebraic matchgate identities in an algorithm for computing the inner product, including phase information, between two FGSs represented by matchgate circuits. Our method reduces the computation to a tensor contraction involving a depth-2 quantum circuit, providing an efficient alternative to covariance-matrix–based techniques. We also defined a standard form for $t$-doped matchgate circuits and the corresponding doped FGSs. Doing so, we showed such unitaries can implemented with a circuit of depth $\mathcal{O}(n+t)$. Computing an inner product with such a circuit can be reduced to contracting a brickwall tensor network of depth $4t+1$.

Our study of the circuit representation of FGSs opens several promising research directions. For instance, one could ask which of our algorithms and results on gate count and depth optimality and algorithms can be generalized to free fermion systems defined on, e.g., a two-dimensional lattice. Given that classical simulability arises from the GYB and LR relations, it would be interesting to study further gate groups that satisfy these relations. A physically relevant example are bosonic Gaussian systems where, as we studied in Appendix~\ref{app:bosonic_gaussian_systems}, the GYB relation is not always exactly satisfied, but may hold approximately. Understanding the implications of such approximate structures for circuit complexity and classical simulation remains an open challenge.

Finally, it would be interesting to relate our work to other quantum resources considered in the literature, such as nonstabilizerness and non-Gaussianity. On the one hand, a recent paper has studied the decomposition of matchgates into braiding gates doped with $T$ gates~\cite{CaBr26}. This opens questions about the representability and approximability of FGSs with braiding circuits with finite number of $T$ gates. On the other hand, the introduction of a generalized RSF formalism for $t$-doped systems opens new approaches to the study of fermionic non-Gaussian resources. This includes the development of quantitative measures of non-Gaussianity~\cite{GoMa05, GoMa06, LuTi24, LyBu24, CoSm25, CuSt25, LeMe25, SiSt26}, as well as potential applications to state preparation protocols and simulation complexity. Exploring these connections may provide a new perspective on circuit representations, resource quantification, and computational power.

\begin{acknowledgments}
We are grateful to Sheng-Hsuan Lin, Fabian Pichler, Poetri Tarabunga, Esther Cruz-Rico, Daniel Molpeceres, Shachar Fraenkel, Beatriz Dias, Georgios Stylaris, Caterina Zerba, Bernhard Jobst, Luca Tagliacozzo, Daniel Malz, Alexander Hahn, Lorenzo Leone, and Germán Sierra for useful comments and insightful discussions. M.L. and B.K. acknowledge funding from the BMW endowment fund and the Horizon Europe programmes HORIZON-CL4-2022-QUANTUM-02-SGA via the project 101113690 (PASQuanS2.1) and HORIZON-CL4-2021-DIGITAL-EMERGING-02-10 under grant agreement No. 101080085 (QCFD). R.M-Y. and F.P. acknowledges support from the Munich Center for Quantum Science and Technology (MCQST), supported by the DFG under Germany’s Excellence Strategy Grant No. EXC–2111–390814868, and FOR 5522 (project-id 499180199). This research is part of the Munich Quantum Valley, which is supported by the Bavarian state government with funds from the Hightech Agenda Bayern Plus. A.G-S. acknowledges support from UK Research and Innovation (UKRI) under the UK government’s Horizon Europe funding guarantee [grant number EP/Y036069/1], and the Leverhulme Trust through Research Project Grant [RPG-2024-112].

\textbf{Data and materials availability:} Raw data and simulation codes are available on Zenodo upon reasonable request~\cite{zenodo}.
\end{acknowledgments}

\appendix

\newcommand\multiplicity[2]{\ensuremath{\operatorname{multiplicity}(#1,#2)}}

\section{Generalized Euler decomposition of orthogonal matrices}\label{appendix:euler_decomp_orth}

In this appendix, we review some well-known Euler decompositions of orthogonal matrices. First, we have to set a definition for denoting non-commuting products.
\begin{definition}[Non-commuting product notation] With $a\leq b$ 
and a sequence of operators $A_a,\ldots,A_b$, we make the conventions 
\[\prod_{i=a}^b A_i = A_b A_{b-1}\ldots A_{a+1} A_a,\]
and
\[\prod_{i=b}^a A_i = A_a A_{a+1}\ldots A_{b-1} A_b.\]
Furthermore, $\prod_{i=a}^a A_i = A_a$.
\end{definition}
The rationale of using this convention is that if a circuit is denoted as $\prod_{i=1}^K U_i$, then the gate $U_1$ with the lowest index acts first and so on.

The Euler decomposition is a central tool in the study of (free) fermionic systems and matchgate circuits~\cite{JoMi08,JiSu18,KiMc18,DaDe19,OsDa22}. Due to its importance and conceptual similarity to our algorithms, we include a particular construction here.

\begin{lemma}[Generalized triangle Euler decomposition (see e.g. Refs.~\cite{JiSu18,KiMc18,DaDe19,OsDa22,RaRu69,HoRa72})]
Given a special orthogonal matrix $R$ of size $N\times N$, there exist angles $\{\alpha_{k,j}\}$ s.t.
\begin{equation}
	R = \prod_{k=1}^{N-1} \prod_{j=1}^{N - k} E_{j, j+1}(\alpha_{k,j}). \label{eq:lemma_generalized_euler_1}
\end{equation}
Likewise, there exist angles $\{\beta_{k,j}\}$ s.t.
\begin{equation}
	R = \prod_{k=1}^{N-1} \prod_{j=1}^{N-k} E_{N-j, N-j+1}(\beta_{k,j}). \label{eq:lemma_generalized_euler_2}
\end{equation}
In case $R$ is orthogonal, but not special orthogonal, one of the Givens rotation must be replaced by a reflection.
\end{lemma}

\begin{proof}
Given a vector $(a,b)^\transpose$, when choosing an angle $\gamma$ with
\[
	\sin\gamma = -\frac{b}{\sqrt{a^2+b^2}}, \quad 
	\cos\gamma = \frac{a}{\sqrt{a^2+b^2}},
\]
one gets
\[
	E(\gamma) \begin{pmatrix} a \\ b\end{pmatrix} =
	\begin{pmatrix} \cos\gamma & -\sin\gamma \\ \sin\gamma & \cos\gamma \end{pmatrix} \begin{pmatrix} a \\ b\end{pmatrix} =
	\begin{pmatrix} \sqrt{a^2+b^2} \\ 0 \end{pmatrix}.
\]
Hence, one can easily define a sequence that rotates the first column of $R$ into the first standard basis vector (i.e., by rotating first the last two entries, then the second-to-last ones and so on). The resulting matrix is still orthogonal and hence must be of the form $\id_1 \oplus R^{(N-1)}$ 
with an $N-1\times N-1$ orthogonal matrix $R^{(N-1)}$. This argument can be repeated another $N-2$ times until one obtains a $2\times 2$ orthogonal matrix $R^{(2)}$ with $\det R = \det R^{(2)}$. The transpose of the sequence obtained this way gives Eq.~\eqref{eq:lemma_generalized_euler_1}. Equation~\eqref{eq:lemma_generalized_euler_2} can be obtained similarly, by first rotating the last column to the last standard basis vector, then the second-to-last column to the second-to-last basis vector and so on.
\end{proof}

From the Euler decomposition on three qubits, one immediately gets a proof of the GYB identity for matchgates up to phases (see also~\cite{MoLa25}). On $n$ qubits, one finds that an arbitrary matchgate unitaries can be decomposed into an MGC in a triangle layout of depth $2n -3$ (e.g.~\cite{OsDa22}). The same can also be shown starting from an arbitrary MGC and repeatedly using the GYB relation~\cite{CaKo22,KoCa22}. To obtain a slightly better circuit depth, one can instead decompose the matchgate unitary in a brickwall layout. Apart from direct computation, the following lemma can be shown by starting from a triangle decomposition and then repeatedly using the GYB relation to obtain a different layout~\cite{CaKo22}.

\begin{lemma}[Brickwall decomposition of MGC~\cite{OsDa22,CaKo22}]
An arbitrary matchgate unitary acting on $n$ qubits can be realized by an MGC in a brickwall layout with a circuit depth of~$n$.
\end{lemma}

\section{Other systems obeying GYB and LR relations}
\label{app:other_GYB_LR}
Our main focus are matchgate circuits, or equivalently, free fermionic circuits in one dimension. Crucial tools to derive our results are 
the generalized Yang-Baxter and the Left-Right relations (Eqs.~\eqref{eq:yb} and~\eqref{eq:lr}). One might now wonder whether there are other gate sets or physical systems in which such relations hold. In this appendix, we give a positive answer to this question. First, in Appendix~\ref{app:trivial_additional_gyb_lr}, we list two other gate sets on qubits that can be shown to satisfy the GYB and LR relations. Second, in Appendix~\ref{app:bosonic_gaussian_systems}, we turn our focus towards bosonic Gaussian unitaries and states. We show that the GYB relation holds for any bosonic Gaussian unitary unless a single particular equation involving the unitaries is satisfied. On other hand, if that is the case, we give an example of three bosonic Gaussian unitaries arranged in one layout of Eq.~\eqref{eq:yb} that cannot be decomposed into the other layout. We then also investigate the LR relation and comment on the potential applicability of the symmetric Euler decomposition algorithm for bosonic Gaussian states.

\subsection{Trivial examples} \label{app:trivial_additional_gyb_lr}

We present two trivial examples of gate sets and states for which a GYB and an LR relation hold. While this enables simulation of these circuits with the inner product algorithm, it is well known that these examples can be simulated with other means. However, their existence shows that there are non-matchgate (by a local basis change) solutions to the GYB and LR equations.

\newcommand{\swapgate}{\ensuremath{\operatorname{S}}}

For the first example, we consider the group of gates generated by the nearest-neighbor swap gate and arbitrary single-qubit gates. The swap itself satisfies the Yang-Baxter equation. 
Moreover, single qubit gates can be rearranged with swap gates by exchanging their qubit labels when moved past a swap. 

As a result, this group satisfies the GYB relation. When acting on product states, an LR identity follows similarly.

For the second example, consider gates that are only generated by a Hamiltonian with terms of the form $Z_j$ and $Z_jZ_{j+1}$. Circuits composed out of such gates are a subset of the so-called IQP (instantaneous quantum polytime) circuits~\cite{ShBr09}. However, as they remain composed of two-local gates acting on nearest neighbors in a single spacial dimension, the known complexity theoretic results (see Ref.~\cite{BrJo10}) do not apply. Clearly, any two of such gates commute irrespective of which qubits they act on. They therefore trivially satisfy a GYB relation and the LR property.

\subsection{Bosonic Gaussian systems} \label{app:bosonic_gaussian_systems}

In the previous section, we gave two examples of gate sets on qubits that satisfy the GYB and LR relations. Here, we investigate bosonic systems in which the local Hilbert spaces are infinite-dimensional. Below, we briefly review the basics of bosonic Gaussian unitaries and states (for a more complete introduction, see, e.g., Ref.~\cite{Se17qcv}). We then investigate decompositions for the GYB relation in Observation~\ref{obs:symplectic_yb}, and for the LR relation in Observation~\ref{obs:bosonic_lr}. The techniques for finding such decompositions rely on Lemma~\ref{lemma:symplectic_eliminate_second} and are very similar to the fermionic case.

Consider a set of $n$ bosonic modes with creation (annihilation) operators $a_i^\dagger$ ($a_i$), satisfying the canonical commutation relations $[a_j, a_j^\dagger] = \delta_{i,j}\id$ and $[a_i, a_j] = [a_i^\dagger, a_j^\dagger] = 0$. Introduce the quadrature operators 
\[
R_{2j-1} = \frac{a_j + a_j^\dagger}{\sqrt2}, \quad\text{and}\quad R_{2j} = -\ii\frac{a_j - a_j^\dagger}{\sqrt2},
\]
which take the role of the Majorana operators in the bosonic case. Bosonic Gaussian unitaries $U$ are generated by Hamiltonians that are at most quadratic in the components of $R$, and bosonic Gaussian states are thermal states of such Hamiltonians (together with limit points). Any bosonic Gaussian unitary $U$ satisfies 
\[U^\dagger R_k U = \sum_{l=1}^{2n} S_{kl} R_l + d_k \id,\]
where the matrix $S$ is an element of the \emph{symplectic group}, meaning that $S J_{2n} S^\transpose = J_{2n}$ with $J_{2n} = \bigoplus_{k=1}^n \covblock$, and $d$ is a real vector. In fact, $U$ is completely specified by $S$ and $d$. Likewise, any bosonic Gaussian state $\rho$ is completely specified by its first moments
\[m_i = \Tr(R_i \rho),\]
and its covariance matrix
\[\Gamma_{ij} = \Tr(\{R_i - m_i, R_j - m_j\} \rho), \]
which is a symmetric matrix. Here, $\{A,B\} = AB+BA$ denotes the anticommutator. The CM $\Gamma$ of any bosonic Gaussian state $\rho$ can be decomposed into
\[ \Gamma = S \, \left(\bigoplus_{i=1}^n \nu_i \id_2\right) \, S^\transpose, \]
where $S$ is symplectic and $\nu_i \geq 1$. In case $\rho$ is a pure bosonic Gaussian state, $\nu_i = 1$ for all $i$, which implies that $\Gamma$ is symplectic. Conversely, it is known that any symmetric, symplectic and positive semidefinite matrix $\Gamma$ is the CM of a pure bosonic Gaussian state.

As in the fermionic case, to investigate the GYB and LR relations for bosonic Gaussian states and unitaries, we turn our attention towards the corresponding covariance and symplectic matrices. For simplicity, we consider only displacement-free unitaries, i.e., those with $d_i=0$, and states with vanishing first moments.

We start considering a bosonic Gaussian unitary acting on three modes and investigate a decomposition into three two-mode bosonic unitaries in either layout appearing in the GYB equation. When cast in terms of the corresponding symplectic matrix, we get the following result.

\begin{observation}\label{obs:symplectic_yb}
    Consider a symplectic matrix
    \[
    S = \begin{pmatrix}
        A & D & G \\
        B & E & H \\
        C & F & I
    \end{pmatrix},
\]
where all block are $2\times 2$ matrices. Let furthermore $A$, $B$, and $C$ be invertible. In case $\det B \neq -\det C$, there exist three $4\times 4$ symplectic matrices $Q_1, Q_2$ and $Q_3$ with 
\begin{equation} \label{eq:symplectic_yb_to_show}
    (\id_2 \oplus  Q_3) ( Q_2 \oplus \id_2)(\id_2 \oplus Q_1) S = \id_6. 
\end{equation}
On the other hand, in case $\det B = -\det C$, the decomposition in Eq.~\eqref{eq:symplectic_yb_to_show} cannot exist.
\end{observation}

The requirement for the individual blocks to be invertible is for the sake of simplicity. On the level of Gaussian unitaries, this proves a decomposition of an arbitrary unitary acting on three modes
into the right hand side of the GYB Eq.~\eqref{eq:yb} up to a phase. The other decomposition can be computed in a similar way under similar constraints.

To prove Observation~\ref{obs:symplectic_yb}, the following two observations are helpful: First, in case a matrix $M$ is symplectic, then $M^\transpose$ is also symplectic. Second, for any $2\times 2$ matrix $A$, $A \covblock A^\transpose = \det A \, \covblock$. We also need the following lemma.

\begin{lemma}\label{lemma:symplectic_eliminate_second}
    Consider two invertible $2\times 2$ matrices $B$ and $C$. In case $\det B \neq -\det C$, there exists a $4\times 4$ symplectic matrix $Q$ and another invertible $2\times 2$ matrix $B'$ s.t.
    \begin{equation}\label{eq:symplectic_eliminate_second}
        Q \begin{pmatrix} B \\C \end{pmatrix} = \begin{pmatrix} B' \\0 \end{pmatrix}.
    \end{equation}
    If, on the other hand, $\det B = -\det C$, then no symplectic $Q$ satisfying Eq.~\eqref{eq:symplectic_eliminate_second} exists.
\end{lemma}

\begin{proof}
Suppose $\det B \neq -\det C$. Let $M = BC^{-1}$ and $s = 1$ if $\det M < -1$, or $s=0$ if $\det M > -1$. A straightforward check shows that the matrix 
\begin{equation}
    Q =
    \frac{1}{\sqrt{\vert 1 + \det M\vert}} (\id_2 \otimes Z)^s \begin{pmatrix}
        -\covblock M^\transpose \covblock & \id_2 \\ \id_2 & - M
    \end{pmatrix}
    \label{eq:solution_eliminate_second_bosons}
\end{equation}
is symplectic and satisfies Eq.~\eqref{eq:symplectic_eliminate_second}. Let now $\det B = -\det C$. We assume that the desired symplectic matrix
\[ Q = \begin{pmatrix} E & F \\ G & H \end{pmatrix}\]
with $2\times 2$ blocks $E,F,G$ and $H$ exists, and derive a contradiction. From Eq.~\eqref{eq:symplectic_eliminate_second}, we get $GB + HC = 0$, or equivalently, $H = - GM$ with $M = B C^{-1}$ as above. Since $Q$ is symplectic, $G \covblock G^\transpose + H \covblock H^\transpose = \covblock$. With the observation that $G \covblock G^\transpose = \det G \,\covblock$ and likewise for $H$, this is equivalent to $\det G + \det H = 1$. The contradiction is obtained with $H = -GM$ and $\det M = -1$.
\end{proof}

We can now state the proof of Observation~\ref{obs:symplectic_yb}.

\begin{proof}[Proof of Observation~\ref{obs:symplectic_yb}]
From Lemma~\ref{lemma:symplectic_eliminate_second}, we get the existence of $Q_1$ with 
\[
    (\id_2 \oplus Q_1) S = \begin{pmatrix}
        A & D & G \\
        B' & E' & H' \\
        0 & F' & I'
    \end{pmatrix}
\]
for some $2\times2$ matrices $B', E', F', H'$ and $I'$. Since the matrix obtained this way is symplectic, we have $\det A + \det B' = 1$, and hence $\det A \neq -\det B'$. One can hence construct another $4\times 4$ symplectic matrix $Q_2$ with
\[
S_2 = (Q_2 \oplus \id_2) (\id_2 \oplus Q_1) S = \begin{pmatrix}
    \id & D' & G' \\
    0 & E'' & H'' \\
    0 & F' & I'
\end{pmatrix}.
\]
Note that the top left bock is $\id$, which can be achieved by a suitable $Q_2$. We now show that $D' = G' = 0$, as is the case with orthogonal matrices of this form. From examining the equations obtained by $S_2 J_{6} S_2^\transpose = J_6$,  we find that the block 
$\begin{pmatrix} E'' & H'' \\ F' & I'\end{pmatrix}$ is itself symplectic and hence has an inverse. We furthermore get the equation
\[
    \begin{pmatrix}
        D' & G'
    \end{pmatrix} J_4 \begin{pmatrix} E'' & H'' \\ F' & I'\end{pmatrix} = 0,
\]
and can infer that $D'=G'=0$. Thus, we can introduce $Q_3$ s.t. $S_2 = \id_2 \oplus Q_3^{-1}$. This proves the decomposition in Eq.~\eqref{eq:symplectic_yb_to_show}.

We now show that in case $\det B = -\det C$, there cannot exist such a decomposition. Note that for any decomposition $(\id_2 \oplus  Q_3) ( Q_2 \oplus \id_2)(\id_2 \oplus Q_1) S = \id_6$, the bottom left block of $(\id_2 \oplus  Q_1)S$ must vanish. Indeed, suppose this were not the case. An application of $ Q_2 \oplus \id_2$ leaves the block untouched, and a following application of $\id_2 \oplus Q_3$ would preserve the rank of the bottom left $4\times 2$ block. Thus, $(\id_2 \oplus  Q_3) (\tilde Q_2 \oplus \id_2)(\id_2 \oplus \tilde Q_1) S \neq \id_6$. The non-existence of the decomposition hence follows from the second statement in Lemma~\ref{lemma:symplectic_eliminate_second}.
\end{proof}

Note that the solution we provide in Eq.~\eqref{eq:solution_eliminate_second_bosons} diverges when $\det M = \det B/\det C \to -1$. Physically, this means that one unitary in the GYB decomposition would require infinite squeezing in case $\det B = -\det C$ (any bosonic Gaussian unitary can be decomposed into two passive and a squeezing part). This leaves the question of whether there exists a symplectic $S$ as in Observation~\ref{obs:symplectic_yb} with $\det B = -\det C$. An explicit example of such a matrix is given by
\[
    S = \begin{pmatrix}
        \id & \frac{\beta}{\sqrt{\beta^2 + 1}} Z & \frac{\beta}{\sqrt{\beta^2 + 1}} \id \\
        \beta Z & \sqrt{\beta^2 + 1} \id & 0 \\
        \beta \id & \frac{\beta^2}{\sqrt{\beta^2 + 1}} Z & \frac{-1}{\sqrt{\beta^2 + 1}} \id
    \end{pmatrix},
\]
where $\beta$ is a non-zero real number. Therefore, this matrix cannot be decomposed in the layout of Observation~\ref{obs:symplectic_yb}. However, with the symplectic matrices
\[
    Q_1 = \begin{pmatrix} 0 & \id \\ \id & 0\end{pmatrix}, \quad 
    Q_2 = \frac{1}{\sqrt{1+\beta^2}} \begin{pmatrix} \id & \beta\id \\ \beta\id & -\id \end{pmatrix},
\]
and
\[
    Q_3 = \begin{pmatrix} \beta Z & \sqrt{1+\beta^2} \id \\ \sqrt{1+\beta^2}\id & \beta Z \end{pmatrix},
\]
we get a decomposition
\[
    (Q_3 \oplus \id_2) (\id_2 \oplus Q_2) (Q_1 \oplus \id_2) S =\id_6
\]
into the other layout for the GYB relation. 

In the next observation, we comment on the LR relation for bosonic Gaussian states.

\begin{observation}\label{obs:bosonic_lr}
Let
\[ \Gamma = \begin{pmatrix}
       A & B^\transpose & C^\transpose \\
        B & D &E^\transpose\\
        C & E & F
    \end{pmatrix}
\]
composed of $2\times$ 2 matrices be the CM of a pure bosonic Gaussian state on three modes, and let $B$ and $C$ be invertible. If $\det B \neq -\det C$, there exist two symplectic $4\times4$ matrices $Q_1$ and $Q_2$ s.t.
\[
    (Q_2\oplus\id_2)(\id_2 \oplus Q_1)\Gamma (\id_2 \oplus Q_1^\transpose)(Q_2^\transpose\oplus\id_2) = \id_6.
\]
\end{observation}
Note that $\id_6$ is the CM of the vacuum state. The requirement of invertibility is again for the sake of simplicity. Up to a phase, from the observation we get a decomposition into one of the two gate layouts necessary for the LR relation. The other decomposition can be obtained analogously and again under similar constraints.

\begin{proof}
Following the arguments from above, there exists a $4\times 4$ symplectic matrix $Q_1$ s.t. the lower left and upper right blocks of $\Gamma_1 = (\id_2\oplus Q_1) S (\id_2\oplus Q_1^\transpose)$ vanish. Denote now
\begin{equation} \label{eq:bgs_cm_close_to_factorization}
    \Gamma_1 = \begin{pmatrix}
        A & B'^\transpose & 0 \\
        B' & D' & E'^\transpose\\
        0 & E' & F'
    \end{pmatrix}.
\end{equation}
We will show that $E' = 0$. From $\Gamma_1J_6\Gamma_1^\transpose = J_6$, we get $E' \covblock B' = 0$. Since, $\begin{pmatrix} B'^\transpose & 0\end{pmatrix} =\begin{pmatrix} B^\transpose & C^\transpose\end{pmatrix} Q_1^\transpose$, we have that $B'$ is invertible, and hence $E' = 0$. Furthermore, $F'$ and $\begin{pmatrix}A & B'^\transpose \\ B' & D'\end{pmatrix}$ are symplectic, and from the positivity of $\Gamma_1$, we get that they correspond to the CMs of pure bosonic Gaussian states on a single and two modes respectively. Those can be each decomposed into Gaussian unitaries on a single mode (which can be appended to $Q_1$), and one on two modes respectively, acting on the vacuum state with a CM of $\id$. 
\end{proof}

Finally, note that the arguments surrounding Eq.~\eqref{eq:bgs_cm_close_to_factorization} can be repeated in case $F'$ corresponds to $n-2$ modes instead of a single one: We get that the corresponding state factorizes between the first two and remaining $n-2$ modes in case $B'$ is invertible. Such a CM arises when applying the symmetric Euler decomposition algorithm (specifically, the two column version we discuss in Appendix~\ref{app:sed_two_columns_version}) to the CM of an arbitrary pure bosonic Gaussian state. In principle, this method can therefore be used to construct bosonic RSF circuits. In order for the algorithm to function as intended, it is however necessary that in each step, the requirements for the first statement in Lemma~\ref{lemma:symplectic_eliminate_second} are fulfilled.

To summarize, we have investigated the GYB and LR relations for other gate sets. Next to two trivial examples on qubits, we have shown that bosonic Gaussian systems potentially satisfy such equations. In full generality, however, this cannot be true as there are explicit counterexamples. One may be able to circumvent the related restriction by, e.g., solving instead an approximate GYB relation. For applications of this in algorithms such as the inner product algorithm, one then needs to bound the error repeatedly imposed by such a procedure. A way of possibly doing this is by bounding the errors imposed on the states rather than on the unitaries (see, e.g., Ref.~\cite{VaGa24}).

\onecolumngrid
\section{Alternative and enhanced Symmetric Euler Decomposition Algorithms} \label{app:alternative_enhanced_sed}

In the main text, we presented Algorithm~\ref{alg:symmetric_euler_decomposition}, which takes as an input the CM of a pure FGS, and outputs sequences of Givens rotations, from which the CM can be constructed. In order to obtain a matchgate circuit in RSF, one has to combine these rotations in a specific way. Here, we present two alternative algorithms which directly output circuits in RSF. The first one is a simple modification of Algorithm~\ref{alg:symmetric_euler_decomposition}. It simultaneously eliminates two columns in the CM at once instead of a single one. The idea here is to find $4\times 4$ rotation matrices that eliminate the last two components of two given vectors of size $4$. With this procedure, two qubits can be disentangled with a single diagonal. In the generic case, the circuits produced by this algorithm will satisfy the conditions of Theorem~\ref{thm:RSF_optimal} and hence have minimal matchgate count for preparing the state. There are, however, some special cases in which some gates outputted by the algorithm have only a single nonzero nonlocal parameter, i.e., the circuits are not necessarily optimal. One particular case here is when certain blocks in the first two columns of the CM do not have two linearly independent vectors (see also Appendix~\ref{app:proof_rsf_circuits_optimal}) and hence some degrees of freedom in the matchgate are not used. To deal with this issue, we introduce here another algorithm that progressively selects larger blocks in the CM until the block contains two linearly independent vectors. Using this third algorithm, we furthermore can prove that MGCs have asymptotically minimal gate count for generating FGSs.

We use, here and in the following, the notation $[a:b] = \{ a, a+1,\ldots, b-1,b\}$. As in the main text, given a matrix $M$, we denote by $M_{ \{a_1,\ldots,a_M\}, \{b_1,\ldots,b_N\}}$ the matrix obtained when selecting rows and columns with indices $\{a_1,\ldots,a_M\}$, and $\{b_1,\ldots,b_N\}$. By $0_n$, we denote the zero column vector of length $n$. Given two vectors $a = (a_1,\ldots, a_m)^\transpose$ and $b = (b_1,\ldots, b_n)^\transpose$, we denote their concatenation by $(\concat{a,b}) = (a_1, \ldots a_m, b_1, \ldots, b_n)^\transpose$. Finally, with $\multiplicity{\lambda}{M}$, we denote the multiplicity of the eigenvalue $\lambda$ of the matrix $M$.

\subsection{Two column version} \label{app:sed_two_columns_version}

We start by presenting a simple modification to Algorithm~\ref{alg:symmetric_euler_decomposition}. Before doing so, we define two subroutines. The first one, called \texttt{elim\_last\_two}, takes as an input two real vectors $x= \begin{pmatrix}x_1&x_2&x_3&x_4\end{pmatrix}^\transpose$ and y= $\begin{pmatrix}y_1&y_2&y_3&y_4\end{pmatrix}^\transpose$, and outputs a rotation matrix that eliminates the last two components of $x$ and $y$. This can be implemented, for instance, by first determining $R_1$ s.t. $R_1 x = \begin{pmatrix}\vert x\vert &0&0&0\end{pmatrix}^\transpose$, and then determining a rotation $R_2 = \id_1 \oplus R_2'$ which eliminates the last two components of $R_1 y$. Note that if $x$ and $y$ are linearly dependent, then $R_2$ can be any $3\times 3$ rotation matrix. The second subroutine is called \texttt{to\_williamson}, takes as an input a CM and outputs a rotation $R$ that transforms the CM to its Williamson normal form. For CMs of pure states, this rotation can be computed using Algorithm~\ref{alg:symmetric_euler_decomposition} (the general cases can be addressed with a real Schur decomposition~\cite{SuTa20}). We can now present the modification of Algorithm~\ref{alg:symmetric_euler_decomposition} in which two columns are eliminated simultaneously.

\begin{algorithm}[H]
\caption{Symmetric Euler decomposition, two column version}
\label{alg:twocol_symmetric_euler_decomposition}
\KwIn{
    A matrix real orthogonal antisymmetric matrix $\Gamma$.
}
\KwOut{
    A sequence of rotations, corresponding to matchgates, and a CM of a computational basis state obtained from applying the rotations to $\Gamma$.
}

$q \gets 1$, $\Gamma^{(1)} \gets \Gamma$

\While{$q \leq n$}{

  \eIf{$|\Gamma^{(q)}_{2q-1,2q}|=1$ \label{algline:tsed:check_nonentangled} }{
  
    $\Gamma^{(q+1)} \gets \Gamma^{(q)}$ \tcp{Qubit $q$ not entangled, can be skipped.}
    
    $q\gets q+1$
    
    }{

    $k \gets $ smallest row index s.t. $\Gamma^{(q)}_{[2k+1:2n],[2q-1, 2q]} = 0$, or $n$ if no block vanishes \label{algline:tsed:assign_l}
    
    $\Gamma^{(q,k)} \gets \Gamma^{(q)}$
    
    \While{$k \geq q + 2$}{

        $\tilde R^{(q,k)} \gets$ \texttt{elim\_last\_two($\Gamma^{(q,k)}_{[2k-3:2k],[2q-1:2q]}$)} \label{algline:tsed:rotate_zero}
        
        $R^{(q,k)} \gets \id_{2k-4} \oplus \tilde R_i \oplus \id_{2n-2k}$
    
        $\Gamma^{(q,k-1)} \gets R^{(q,k)} \Gamma^{(q,k)} R^{(q,k)\transpose}$ \label{algline:tsed:add_gate1}

        $k \gets k-1$,
         
        }

    $\tilde R^{(q,k)} \gets$ \texttt{to\_williamson($\Gamma^{(q,k)}_{[2q-1: 2q+2],[2q-1:2q+2]}$)} \label{algline:tsed:to_williamson}

    $R^{(q,k)} \gets \id_{2q - 2} \oplus \tilde R_i \oplus \id_{2n-2q - 2}$ \label{algline:tsed:add_gate2}
    
    $\Gamma^{(q+2)} \gets R^{(q,k)} \Gamma^{(q,k)} R^{(q,k)\transpose}$\label{algline:tsed:last_update_gamma}
    
    $q \gets q + 2$ \label{algline:tsed:increase_k} \tcp{Two qubits have been disentangled with a single diagonal.}
  }
}

\Return{} sequence of rotations $(R^{(i, j)})$, $\Gamma^{(q)}$

\end{algorithm}

In the following, we briefly justify why this algorithm works. First, note that the output circuit is indeed in RSF, since the gates produced in lines~\ref{algline:tsed:add_gate1} and~\ref{algline:tsed:add_gate2} sequential values of $k$ are in a diagonal pattern. Moreover, due to line~\ref{algline:tsed:increase_k}, all of these diagonals act on distinct qubits. Next, we verify that the output CM of the algorithm corresponds to a computational basis state. To this end, we consider the CM $\Gamma^{(q,k)}$ in line~\ref{algline:tsed:last_update_gamma} before the update. We show this inductively. Suppose that the first $q-1$ corresponding qubits are already in a computational basis state. We can then write
\[
    \Gamma^{(q,k)} = \bigoplus_{i=1}^{q-1} \lambda_i \covblock \, \oplus \, \begin{pmatrix}
        \Gamma_{\partyA\partyA} & -\Gamma_{\partyB\partyA}^\transpose & 0 \\
        \Gamma_{\partyB\partyA} & \Gamma_{\partyB\partyB} & -\Gamma_{\partyB\partyC}^\transpose \\
        0 & \Gamma_{\partyB\partyC} & \Gamma_{\partyC\partyC} \\
    \end{pmatrix},
\]
with $\vert\lambda_i\vert = 1$, and the blocks $\Gamma_{\partyA\partyA}$, $\Gamma_{\partyB\partyB}$ corresponding each to a single qubit and $\Gamma_{\partyC\partyC}$ to the remaining $n-q-1$ qubits. We show that $\Gamma_{\partyB\partyC} = 0$, i.e., that the state corresponding to $\Gamma$ factorizes over the first $q+1$ and the remaining qubits. Since $\Gamma_{\partyA\partyA}$ is a $2\times 2$ antisymmetric matrix, and due to the check in step~\ref{algline:tsed:check_nonentangled}, we have $\Gamma_{\partyA\partyA}\Gamma_{\partyA\partyA}^\transpose = \alpha^2\id$ for some $0 \leq \alpha < 1$. From $\Gamma^{(q,k)}\Gamma^{(q,k)\transpose} = \ID$, we get $\Gamma_{\partyA\partyA}\Gamma_{\partyA\partyA}^\transpose + \Gamma_{\partyB\partyA}^\transpose\Gamma_{\partyB\partyA} = \ID$ and $\Gamma_{\partyB\partyA}^\transpose\Gamma_{\partyB\partyC}^\transpose = 0$. Thus, we have that $\Gamma_{\partyB\partyA}^\transpose\Gamma_{\partyB\partyA} = (1-\alpha^2)\ID$. That is, $\Gamma_{\partyB\partyA}$ has full rank, and $\Gamma_{\partyB\partyC} = 0$. The update in line~\ref{algline:tsed:last_update_gamma} disentangles the state on qubits $q$ and $q+1$. Since no gate have acted on the first $q-1$ qubits, the first $q+1$ qubits are now in a computational basis state. It hence follows that the output CM of Algorithm~\ref{alg:twocol_symmetric_euler_decomposition} indeed corresponds to a computational basis state.

\subsection{Enhanced Symmetric Euler decomposition}
\label{app:enhaced_decomposition}

So far, we have presented two simple algorithms for obtaining RSF matchgate circuits to generate any pure FGS when given the CM as an input. Algorithm~\ref{alg:symmetric_euler_decomposition} eliminates single columns in the CM using Givens rotations, whereas  Algorithm~\ref{alg:twocol_symmetric_euler_decomposition} eliminates two columns at once by finding instead a sequence of $4 \times 4$ rotations. In case the $4\times 2$ block in step~\ref{algline:tsed:rotate_zero} of Algorithm~\ref{alg:twocol_symmetric_euler_decomposition} contains two linearly dependent columns, the nonlocal parameters of the matchgate corresponding to $\tilde R^{(q,k)}$ are not uniquely determined. Here, we introduce another algorithm that utilizes the nonlocal degrees of freedom of each matchgate more effectively. The core idea of eliminating entries in the CM remains the same, whereas the main modification is concerned with selecting blocks which always contain two linearly independent vectors. Doing this has a significant consequence: As we show below, each gate placed by this algorithm maximally reduces the Schmidt rank. The number of gates is therefore upper bounded by \[K = \sum_{i=1}^{n-1} \lsr_i( \ket{\psi}),\]
where $\lsr_i( \ket{\psi})$ denotes the logarithm of the Schmidt rank in the bipartition $1,\ldots,i\vert i+1,\ldots,n$. Hence, the MGC produced by this algorithm has asymptotically optimal gate count among all composed of nearest neighbor gates, as the best such circuit needs at least $K/2$ gates. In the following, we present the algorithm and later prove the upper bound on the number of gates.

To improve readability, we separate the enhanced symmetric Euler decomposition (\texttt{esed}) algorithm into two subroutines. The first, \texttt{esed\_inner}, takes as an input the CM of an FGS $\ket{\psi}$ that does not factorize in any bipartition, and outputs a diagonal of matchgates that disentangles at least the first two qubits. The second, \texttt{esed\_outer}, identifies regions of consecutive qubits in which the state does not factorize, and uses \texttt{esed\_inner} to disentangle them. After each application of \texttt{esed\_inner}, \texttt{esed\_outer} is used again with the resulting CM of the qubits that were not disentangled. As we see below, this recursion will always terminate since the entanglement in each block is always reduced by \texttt{esed\_inner}. We again make use of the subroutines \texttt{elim\_last\_two} and \texttt{to\_williamson} introduced in the previous subsection. Below, we state now the first subroutine.

\begin{algorithm}[H]
\caption{\texttt{esed\_outer}: Enhanced Symmetric Euler decomposition main routine}
\label{alg:outer_enhanced_symmetric_euler_decomposition}
\KwIn{
    A CM $\Gamma$ of a pure FGS on $n$ qubits.
}
\KwOut{
    A sequence of rotations, corresponding to matchgates, that maps $\Gamma$ to the CM of a computational basis state.
}

\uIf{$n=1$}{
\Return nothing
}

\uElseIf{$n=2$}{
\Return \texttt{to\_williamson($\Gamma$)}
}

\Else{
$k \gets$ the smallest $k'>0$ s.t. $\Gamma_{[1:2k'], [1:2k']}$ is the CM of a pure FGS. 

sequence $(\tilde R_i) \gets$ \texttt{esed\_inner($\Gamma_{[1:2k], [1:2k]}$)}.

sequence $(R_i) \gets (\tilde R_i \oplus \id_{2(n-k)})$

$R \gets \prod_i R_i$, $\tilde \Gamma \gets R \Gamma R^\transpose$.

sequence $(\tilde S_i) \gets$ \texttt{esed\_outer($\tilde \Gamma_{[5 : 2k],[5 : 2k]}$)} \label{algline:esed_outer_recursion1}

sequence $(S_i) \gets (\id_4 \oplus \tilde S_i \oplus \id_{2(n-k)})$

\If{$k<n$}{

sequence $(\tilde T_i) \gets $ \texttt{esed\_outer($\Gamma_{[2k+1:2n], [2k+1:2n]}$)} \label{algline:esed_outer_recursion2}

sequence $(T_i) \gets (\id_{2k} \oplus \tilde T_i)$

}

\Return the concatenation of the rotation sequences $(R_i)$, $(S_i),$ and $(T_i)$.
}

\end{algorithm}

Next, we state the algorithm \texttt{esed\_inner} for disentangling two qubits in a way that most efficiently uses the non-local degrees of freedom in each gate.

\begin{algorithm}[H]
\caption{\texttt{esed\_inner}: Enhanced Symmetric Euler decomposition subroutine}
\label{alg:inner_enhanced_symmetric_euler_decomposition}
\KwIn{
    A CM $\Gamma$ of a pure FGS on $n$ qubits, $\Gamma \neq \Gamma_1 \oplus \Gamma_2$ for any CMs $\Gamma_1, \Gamma_2$.
}
\KwOut{
    A sequence of $n-1$ rotations $R_i$ s.t. $\prod_i R_i \Gamma (\prod_i R_i)^\transpose = \lambda_1 \covblock \oplus \lambda_2 \covblock \oplus \Gamma'$, where $|\lambda_1 | = |\lambda_2| = 1$, and $\Gamma'$ is the CM of another pure FGS on $n-2$ qubits.
}

\If{$n = 1$}{\Return{} nothing}

\If{$n > 2$}{
$\Gamma^{(0)} \gets \Gamma$ 

\For{$i \gets 1$ \KwTo $n-2$}{

    $k_i \gets$ smallest possible $k'$ s.t. $\Gamma^{(i-1)}_{[2n - 2i - 1: 2n], [1:k']}$ has two linearly independent columns \label{algline:iesed:def_ki}

    $\tilde \alpha^{(i)}, \tilde\beta^{(i)} \gets$ two linearly independent columns of $\Gamma^{(i-1)}_{[2n - 2i - 1:2n], [1:k_i]}$ \label{algline:iesed:selection_large_vectors}

    $\alpha^{(i)} \gets \tilde \alpha^{(i)}_{[1:4]}$, $\beta^{(i)} \gets \tilde \beta^{(i)}_{[1:4]}$
    \label{algline:iesed:vector_truncation}

    $\tilde R_i \gets$ \texttt{elim\_last\_two($\begin{pmatrix} \alpha^{(i)} & \beta^{(i)} \end{pmatrix}$)}
    \label{algline:iesed:ised:def_small_Ri}

    $R_i \gets \id_{2n - 2i -2} \oplus \tilde R_i \oplus \id_{2i - 2}$ \label{algline:iesed:def_Ri}

    $\Gamma^{(i)} \gets R_i \Gamma^{(i-1)} R_i^\transpose$ \label{algline:iesed:update_gamma}
}}

$\tilde R_{n-1} \gets$ \texttt{to\_williamson($\Gamma^{(n-2)}_{[1:4],[1:4]}$)} \label{algline:iesed_last_R_computation}

$R_{n-1} \gets \tilde R_{n-1} \oplus \id_{2n-4}$

\Return the sequence of rotations $(R_1, \ldots, R_{n-1})$

\end{algorithm}

Note that, although the algorithm \texttt{esed\_outer} (Algorithm~\ref{alg:outer_enhanced_symmetric_euler_decomposition}) is recursive, it still terminates after a finite number of rounds. This is because the recursions in lines~\ref{algline:esed_outer_recursion1} and~\ref{algline:esed_outer_recursion2} of Algorithm~\ref{alg:outer_enhanced_symmetric_euler_decomposition} are applied to CMs of states on a strictly smaller number of qubits. For the rest of this appendix, we will prove that statements on the Schmidt rank reduction are true. This also shows that the state after this procedure is indeed a product state. We start by stating some supplementary results. The following observation is a direct consequence of the result in Ref.~\cite{BoRe04}. 

\begin{observation}\label{observation:lsr_pure_fgs}
    Consider a pure FGS $\ket{\psi}$ shared between parties $\partyA$ and $\partyB$, comprising qubits $1,\ldots,k$ and $k+1,\ldots,n$, with covariance matrix
    \[ \Gamma = \begin{bmatrix}\Gamma_{\partyA\partyA} & -\Gamma_{\partyB\partyA}^\transpose \\ \Gamma_{\partyB\partyA} & \Gamma_{\partyB\partyB}\end{bmatrix}. \]
    For the logarithm of the Schmidt rank $\lsr(\ket\psi, \partyA,\partyB)$ in the bipartition $\partyA|\partyB$, it holds
    \[
    \lsr(\ket{\psi}, \partyA, \partyB) = \frac{1}{2} \operatorname{rank} (\Gamma_{\partyB\partyA}) = |\partyA| - \multiplicity{\ii}{\Gamma_{\partyA\partyA}} = |\partyB| - \multiplicity{\ii}{\Gamma_{\partyB\partyB}}.
    \]
\end{observation}
This implies that for pure states on two qubits, the rank of the offdiagonal $2\times 2$ block of the corresponding CM is hence either $0$ or $2$ for a product or entangled state respectively. We will frequently use this fact later.

\begin{proof}
    As stated in Ref.~\cite{BoRe04}, there exist orthogonal matrices $R_\partyA$ and $R_\partyB$, s.t. 
    \[
        (R_\partyA \oplus R_\partyB) \, \Gamma (R_\partyA \oplus R_\partyB)^\transpose = 
        \Big( \bigoplus_{i=1}^{n_\partyA} a_i \covblock \Big)
        \oplus 
            \begin{pmatrix}
                \lambda_1 \covblock & 0  & \cdots & 0 & 0 & \cdots &0 & \mu_1 X \\
                0 & \lambda_2 \covblock & \cdots & 0 & 0 & \cdots & \mu_2 X & 0 \\
                \vdots & \vdots  &  & \vdots & \vdots &  & \vdots & \vdots \\
                0 & 0  & \cdots &  \lambda_{n_{\partyA\partyB}} \covblock & \mu_{n_{\partyA\partyB}} X& \cdots & 0 & 0 \\
                0 & 0  & \cdots &  -\mu_{n_{\partyA\partyB}} X & \lambda_{n_{\partyA\partyB}} \covblock & \cdots & 0 & 0 \\
                \vdots & \vdots  &  & \vdots & \vdots &  & \vdots & \vdots \\
                0 & -\mu_2 X & \cdots & 0 & 0 & \cdots & \lambda_2 \covblock & 0 \\
                -\mu_1 X & 0  & \cdots & 0 & 0 & \cdots &0 & \lambda_1 \covblock
            \end{pmatrix}
        \oplus \Big( \bigoplus_{i=1}^{n_\partyB} b_i \covblock \Big),
    \]
    with $| a_i| = |b_i| = 1$, $|\lambda_i| < 1$, and $\lambda_i^2 + \mu_i^2 = 1$. That is, there are $n_{\partyA\partyB}$ entangled pairs shared between $\partyA$ and $\partyB$, and the number of qubits held by $\partyA$ is given by $|\partyA| = n_\partyA + n_{\partyA\partyB}$, as well as $|\partyB| = n_\partyB + n_{\partyA\partyB}$. The rank of the transformed offdiagonal block is clearly $\rank(R_\partyB \Gamma_{\partyA\partyB} R_\partyA^\transpose )=2 n_{\partyA\partyB}$, and therefore $\rank(\Gamma_{\partyA\partyB}) = 2n_{\partyA\partyB}$. Likewise the multiplicities $n_\partyA = \vert A \vert - n_{\partyA\partyB}$ and $n_\partyB = \vert B \vert - n_{\partyA\partyB}$ of the eigenvalue $\ii$ of the diagonal blocks do not change with the orthogonal transformation. The statements hence follow from $\lsr(\ket{\psi}, \partyA, \partyB) = n_{\partyA\partyB}$.
\end{proof}

The next lemma deals with some properties of the eigenvectors of CMs. It captures the physical intuition that an isolated mode in a subsystem (which is characterized by the existence of a MGC $U$ s.t. $U\rho U^\dagger = \ketbra{0} \otimes \rho'$) is also an isolated mode in the larger system (see also, e.g., Ref.~\cite{FiWh15}).

\begin{lemma}\label{lemma:unity_eigenvectors_of_antisym_submatrix}
Consider a not necessarily pure FGS shared among parties $\partyA$ and $\partyB$ with covariance matrix 
\[ \Gamma = \begin{bmatrix}\Gamma_{\partyA\partyA} & -\Gamma_{\partyB\partyA}^\transpose \\ \Gamma_{\partyB\partyA} & \Gamma_{\partyB\partyB} \end{bmatrix}.\]
In case $v$ is an eigenvector of $\Gamma_{\partyA\partyA}$ to eigenvalue $\ii$, then $(\concat{v, 0_{2 \vert\partyB\vert}})$ is an eigenvector of $\Gamma$ to eigenvalue $\ii$. 
\end{lemma}
Due to this lemma, we clearly have that  $\multiplicity{\ii}{ \Gamma_{\partyA\partyA}} \leq \multiplicity{\ii}{\Gamma}$. An analogous statement holds for eigenvectors to eigenvalue $\ii$ of $\Gamma_{\partyB\partyB}$.

\begin{proof}
Since $\Gamma_{\partyA\partyA}$ is itself the CM of the reduced state on party $\partyA$ (i.e., an antisymmetric matrix with $\Gamma_{\partyA\partyA} \Gamma_{\partyA\partyA}^\transpose \leq \ID$), it can be brought to the Williamson normal form
\[R \Gamma_{\partyA\partyA} R^\transpose = \bigoplus_{i=1}^{\vert\partyA\vert} \lambda_i \covblock = (\ID_k \otimes \covblock) \oplus D \] with an orthogonal matrix $R$, $\lambda_i \geq \lambda_j$ for $i<j$ and $\lambda_i = 1$ for $i \leq k$, where $k = \operatorname{multiplicity}(\ii, \Gamma_{\partyA\partyA})$. The matrix $D$ is of size $2(\vert\partyA\vert - k) \times 2(\vert\partyA\vert - k)$. A basis of the eigenspace to eigenvalue $\ii$ of $\Gamma_{\partyA\partyA}$ can be chosen as
\[ v_j =  R^\transpose \, (\concat{0_{2(j - 1)}, \begin{pmatrix}1 & \ii\end{pmatrix}^\transpose, 0_{2(k - j)}, 0_{2(\vert\partyA\vert - k)}} )\] for $j=1,\ldots,k$. To prove the claim, we need to show that $(\concat{v_j,0_{2|\partyB|}})$ is an eigenvector of $\Gamma$, and hence that $\Gamma_{\partyB\partyA} v_j = 0$ for each $j$. Consider now the matrix $\Gamma_{\partyB\partyA}R^\transpose$ and split up the first $2k$ and the remaining $2(\vert\partyA\vert -k)$ columns into two matrices $E$ and $F$, i.e.
\[ \Gamma_{\partyB\partyA} R^\transpose = \begin{pmatrix} E & F \end{pmatrix}.\]
With this definition, one gets
\[
\Gamma_{\partyB\partyA} v_j 
    = \Gamma_{\partyB\partyA} R^\transpose R v_j = \begin{pmatrix} E & F \end{pmatrix} 
\begin{pmatrix} \concat{0_{2(j - 1)}, (1 \,\,\, \ii)^\transpose , 0_{2(k - j)}} \\ 0_{2(\vert\partyA\vert - k)} \end{pmatrix}
= E \, (\concat{0_{2(j - 1)}, (1 \,\,\, \ii)^\transpose , 0_{2(k - j)}}). 
\]
To complete the proof, we will show that $E=0$: Define a new CM $\tilde \Gamma$, obtained by applying the rotation $R\oplus\id$ to $\Gamma$. That is,
\[ \tilde \Gamma = (R\oplus\id) \Gamma  (R\oplus\id)^\transpose =
	 \begin{pmatrix}
        \ID_k\otimes\covblock & 0 & -E^\transpose \\ 
        0 & D & -F\transpose \\ 
        E & F & \Gamma_{\partyB\partyB}
    \end{pmatrix}, \]
and the top left $2k\times 2k$ block of $\tilde \Gamma \tilde\Gamma^\transpose$ is given by $E^\transpose E + \ID_{2k}$. Since $\tilde \Gamma$ must satisfy $\tilde \Gamma \tilde \Gamma^\transpose \leq \ID$, it follows that $E=0$.
\end{proof}

Let us now verify that Algorithms~\ref{alg:outer_enhanced_symmetric_euler_decomposition} and~\ref{alg:inner_enhanced_symmetric_euler_decomposition} function as intended. We will start with the latter, whose goal it is  to find a diagonal of gates that disentangle two qubits from a given FGS. Note that this goal can in principle can be achieved with a single application of the main loop in Algorithm~\ref{alg:twocol_symmetric_euler_decomposition}. The difference here is, as stated before, that each gate is chosen to maximally reduce the Schmidt rank. Let us elaborate on this point. As can be easily seen, all states satisfy the property
\begin{equation}
    \vert \lsr_{i-1}(\ket{\psi})  - \lsr_{i}(\ket{\psi}) \vert \leq 1, \label{eq:appendix_esed_entanglement_monogamy}
\end{equation}
that is, the Schmidt rank in a bipartition can differ by at most a factor of two when considering another bipartition that differs by a single qubit. For convenience, we set $\lsr_0(\ket{\psi}) = \lsr_n(\ket{\psi}) = 0$ such that Eq.~\eqref{eq:appendix_esed_entanglement_monogamy} also holds on the boundary of the chain. We show for each gate $U$ that, when it is applied to qubits $i$ and $i+1$, the Schmidt rank $\lsr_i(U\ket{\psi})$ attains the minimal value w.r.t. the Schmidt ranks in the bipartitions $1,\ldots, i-1|i,\ldots,n$ and $1,\ldots, i+1|i+2,\ldots,n$. This statement is contained in the following proposition.

\begin{proposition} \label{proposition:iesed_entanglement_reduction}
    Let $\ket{\psi}$ be a FGS on $n$ qubits, which cannot be written as a product state in any bipartition over consecutive qubits. Denote by $(U_i)_{i=1}^{n-1}$ the diagonal of $n-1$ matchgates corresponding to the rotations $(R_i)$ obtained when applying Algorithm~\ref{alg:inner_enhanced_symmetric_euler_decomposition} on $\ket{\psi}$, where $U_i$ acts on qubits $n-i$ and $n-i+1$. The intermediate states $\ket{\psi_i} = \prod_{j\leq i} U_j \ket{\psi}$ satisfy
    \begin{equation} \label{eq:prop_iesed_entanglement_reduction}
    \lsr_{n-i}(\ket{\psi_i}) = \max\{\lsr_{n-i-1}({\ket{\psi_i}), \lsr_{n-i+1}}(\ket{\psi_i})\} - 1.
    \end{equation} For the final state $\ket{\psi_{n-1}}$, we have
    \begin{equation}
        \label{eq:prop_iesed_sum_schmidt_rank_reduction}
        \sum_{i=1}^{n-1} \lsr_i(\ket{\psi_{n-1}}) = \sum_{i=1}^{n-1} \lsr_i(\ket{\psi}) - (n-1),
    \end{equation}
    and
    \begin{equation}\label{eq:prop_iesed_final_state}
        \ket{\psi_{n-1}} = \ket{\cbs_1} \ket{\cbs_2}\ket{\psi'} 
    \end{equation}
    where $\cbs_1,\cbs_2\in\{0,1\}$, and $\ket{\psi'}$ is a pure FGS on $n-2$ qubits.
\end{proposition}

\begin{proof}
Since the input state $\ket{\psi}$ is promised to not factorize over any bipartition of consecutive qubits, we have
\begin{equation}
\lsr_i(\ket{\psi}) \geq 1 \text{ for } 1\leq i \leq n-1, \label{eq:proof_iesed_initial_state_entangled}
\end{equation}
as well as
\begin{equation}
\lsr_1(\ket{\psi}) = \lsr_{n-1}(\ket{\psi}) = 1.\label{eq:proof_iesed_initial_state_boundaries}
\end{equation}
Each gate $U_i$ acts only on qubits $n-i$ and $n-i+1$, and hence the CMs $\Gamma^{(i)}$ and $\Gamma^{(i+1)}$ of the intermediate states $\ket{\psi_i}$ and $\ket{\psi_{i+1}}$ differ in at most in the four rows and columns with indices $2(n-i)-1,\ldots, 2(n-i)+2$. Moreover, the Schmidt ranks of $\ket{\psi_i}$ and $\ket{\psi_{i+1}}$ coincide in all bipartitions $1,\ldots k \vert k+1, \ldots n$, except for $k = n-i$.

From now on, assume that $n\geq 3$ since the cases $n=1$ and $n=2$ follow trivially. Our aim is to show the equation
\begin{equation} \label{eq:proof_iesed_schmidt_rank_reduction_indices}
    \lsr_{n-i}(\ket{\psi_i}) = \lsr_{n-i-1}(\ket{\psi_i}) - 1
\end{equation}
for the intermediate state $\ket{\psi_i}$ for any index $i$ with $1\leq i\leq n-2$ (we comment below on the case $i=n-1$). As we show in the following, all statements in Proposition~\ref{proposition:iesed_entanglement_reduction} can be derived from Eq.~\eqref{eq:proof_iesed_schmidt_rank_reduction_indices}. First, Eq.~\eqref{eq:prop_iesed_entanglement_reduction} follows since we have that $\lsr_{n-i+1}(\ket{\psi_i}) \leq \lsr_{n-i}(\ket{\psi_i})+ 1 = \lsr_{n-i-1}(\ket{\psi_i})$ due to Eqs.~\eqref{eq:appendix_esed_entanglement_monogamy} and~\eqref{eq:proof_iesed_schmidt_rank_reduction_indices}. Proving Eq.~\eqref{eq:prop_iesed_sum_schmidt_rank_reduction}, can be done as follows: Before applying the gate $U_i$, the Schmidt rank difference $\lsr_{n-i}(\ket{\psi_{i-1}}) - \lsr_{n-i-1}(\ket{\psi_{i-1}})$ can take one of the three values $-1,0,$ and $+1$. In the first case, after applying $U_i$ on qubits $n-i$ and $n-i+1$ the Schmidt rank $\lsr_{n-i}(\ket{\psi_{i}})$ in the corresponding bipartition coincides with the Schmidt rank $\lsr_{n-i}(\ket{\psi_{i-1}})$ before applying the gate. In the second case, it is reduced by $1$ compared to $\lsr_{n-i}(\ket{\psi_{i-1}})$ and in the third case it is reduced by $2$. The first and third cases, however, have to occur equally often as otherwise, Eq.~\eqref{eq:proof_iesed_initial_state_boundaries} is violated. For each $i$ with $1\leq i \leq n-2$, one gate is applied. That is, we have
\[
      \sum_{i=1}^{n-1} \lsr_i(\ket{\psi_{n-2}}) = \sum_{i=1}^{n-1} \lsr_i(\ket{\psi}) - (n-2).
\]
To see how Eq.~\eqref{eq:prop_iesed_entanglement_reduction} follows for the last state $\ket{\psi_{n-1}}$, note that in step $i=n-2$ the Schmidt rank between the first two and the remaining qubits is zero. That is, $\ket{\psi_{n-2}}$ is a product state in this bipartition, and the first two qubits can be transformed to a computational basis state with a single gate (see line~\ref{algline:iesed_last_R_computation}). This also shows Eq.~\eqref{eq:prop_iesed_final_state}.

We now show Eq.~\eqref{eq:proof_iesed_schmidt_rank_reduction_indices}. The idea is to construct some linearly independent eigenvectors to eigenvalue $\ii$ of the reduced CM corresponding to qubits $1,\ldots,n-i$. To be more precise, we construct $m+2$ such eigenvectors, where $m$ is the multiplicity of the eigenvalue $\ii$ of the reduced CM corresponding to qubits $1,\ldots,n-i-1$. Equation~\eqref{eq:proof_iesed_schmidt_rank_reduction_indices} then follows from Observation~\ref{observation:lsr_pure_fgs} and Eq.~\eqref{eq:appendix_esed_entanglement_monogamy}. We present the details below.

Pick a particular $i$ with $1\leq i \leq n-2$. As in Algorithm~\ref{alg:inner_enhanced_symmetric_euler_decomposition} (lines~\ref{algline:iesed:def_ki}, \ref{algline:iesed:selection_large_vectors} and~\ref{algline:iesed:vector_truncation}), denote by $k_i$ the smallest integer such that the block $\Gamma^{(i-1)}_{[2n-2i-1:2n],[1:k_i]}$ has rank~$2$,
and by $\tilde \alpha^{(i)}$ and $\tilde \beta^{(i)}$ two linearly independent vectors in first four rows $\Gamma^{(i-1)}_{[2n-2i-1:2n-2i+2],[1:k_i]}$. Note that such $k_i$ with $2 \leq k_i \leq 2n$ always exists. The rotation matrix 
\[R_i = \id_{2n-2i-2} \oplus \tilde R_i \oplus \id_{2i-2}\]
corresponding to $U_i$ obtained in lines~\ref{algline:iesed:ised:def_small_Ri} and~\ref{algline:iesed:def_Ri} satisfies 
$ \tilde R_i \tilde \alpha^{(i)} = (\concat{\gamma, 0_2})$ and $ \tilde R_i \tilde \beta^{(i)} = (\concat{\delta, 0_2})$ for some linearly independent vectors $\gamma$ and $\delta$.

Our first aim is to show the bound $k_i \leq 2n-2i-2$. This can be done by contradiction. Suppose that $k_i > 2n-2i-2$, implying that $\operatorname{rank}(\Gamma^{(i-1)}_{[2n - 2i - 1:2n],[1:2n-2i-2]}) < 2$. Due to Observation~\ref{observation:lsr_pure_fgs}, this rank cannot be an odd number, and therefore must be zero. As can be seen in lines~\ref{algline:iesed:def_Ri} and~\ref{algline:iesed:update_gamma} for preceding values of $i$, the block $\Gamma^{(i-1)}_{[2n - 2i - 1:2n],[1:2n-2i-2]}$ is related to the same block in the CM $\Gamma$ of the initial state by a sequence of rotations (those correspond to the matchgates acting on qubits $n, n-1,\ldots, n-i$), and thus one would obtain 
\[
 \lsr_{n-i-1}(\ket{\psi}) = \frac12 \operatorname{rank}(\Gamma_{[2n - 2i - 1:2n],[1:2n-2i-2]}) = 0,
\]
which contradicts Eq.~\eqref{eq:proof_iesed_initial_state_entangled}. With this upper bound, we get that the block with indices $2n-2i-1, \ldots,2n$ and $1,\ldots,k_i$ of the CM $\Gamma^{(i)} = R_i \Gamma^{(i-1)}R_i^\transpose$ is given by 
\[
    \Gamma^{(i)}_{[2n-2i-1:2n],[1:k_i]} = (\tilde R_i \oplus \id_{2i-2}) \Gamma^{(i-1)}_{[2n-2i-1:2n],[1:k_i]}.
\]
Furthermore, we clearly have $k_i \leq k_{i-1}$ (for $i>1$), since for all but the first two rows of $\Gamma^{(i-1)}_{[2n-2i-1: 2n],[1:k_{i-1}]}$, we have 
\[\Gamma^{(i-1)}_{[2n-2i+1: 2n],[1:k_{i-1}]} = (\tilde R_{i-1} \oplus \id_{2i-4}) \Gamma^{(i-2)}_{[2n-2i+1: 2n],[1:k_{i-1}]},\]
and then appending the two rows $\Gamma^{(i-1)}_{\{2n-2i-1,2n-2i\}, [1:k_{i-1}]}$ can only increase the rank. Due to the action of all previous rounds including the current round $i$, we have that
\[
\Gamma^{(i)}_{[2n-2i+1:2n],[1:k_i]} = 0, \quad\text{ and }\quad \Gamma^{(i)}_{\{2n-2i-1, 2n-2i\},[1:k_i]} = \begin{pmatrix} \gamma_1 & \delta_1 \\ \gamma_2 & \delta_2 \end{pmatrix} M 
\]
for some $2\times k_i$ matrix $M$.

Introduce now three parties $\partyA$, $\partyB$ and $\partyC$, comprising qubits $[1:n-i-1]$, $\{n-i\}$ and $[n-i+1,n]$ respectively. Equation~\eqref{eq:proof_iesed_schmidt_rank_reduction_indices} then can be expressed as 
\begin{equation} \label{eq:proof_iesed_schmidt_rank_reduction}
    \lsr(\ket{\psi_i}, \partyA\partyB, \partyC) = \lsr(\ket{\psi_i}, \partyA, \partyB\partyC) - 1.
\end{equation}
Apart from a permutation on the rows and columns corresponding to party~$\partyA$ \footnote{The aim of this permutation is to have the two linearly independent vectors in the first two columns. Any permutation acting only on party $\partyA$ cannot change the entanglement w.r.t. to the bipartitions of interest.}, one can write
\[
    \Gamma^{(i)} = \begin{pmatrix}
        0 & -\eta & -x^\transpose & -\gamma^\transpose & 0^\transpose_{2i} \\
        \eta &  0 & -y^\transpose & -\delta^\transpose & 0^\transpose_{2i} \\
        x & y & \tilde \Gamma_\partyA & -S^\transpose & -T^\transpose  \\
        \gamma & \delta & S &  \Gamma_\partyB &  -U^\transpose \\
        0_{2i} & 0_{2i}  & T & U &  \Gamma_\partyC
    \end{pmatrix},
\]
with a number $\eta$, suitable vectors $x$ and $y$, matrices $S$, $T$ and $U$. The reduced CMs corresponding to the parties are 
\[
    \Gamma_\partyA = \begin{pmatrix}
        0 & -\eta & -x^\transpose \\
        \eta &  0 & -y^\transpose \\
        x & y & \tilde \Gamma_\partyA   \\
    \end{pmatrix},
\]
$\Gamma_\partyB$ and $\Gamma_\partyC$. Note that the case $i=n-2$ can only be consolidated with this equation when considering $x$, $y$, $\tilde \Gamma_\partyA$, $S$, and $T$ to be of size zero. In the following equations, this case can be treated by disregarding occurences of these symbols. Denote now by $\Gamma_{\partyA\partyB}$ the reduced CM of the joint system $\partyA\partyB$. Due to Observation~\ref{observation:lsr_pure_fgs}, we have
\[
    \lsr(\ket{\psi_i}, \partyA\partyB, \partyC) = n - i - \multiplicity{\ii}{\Gamma_{\partyA\partyB}}
\] and, when introducing $m = \multiplicity{\ii}{\Gamma_{\partyA}}$,
\[
    \lsr(\ket{\psi_i}, \partyA, \partyB\partyC) = n - i - 1 - m.
\]
Choose now a basis $(v_j)_{j=1}^m$ of the eigenspace to eigenvalue $\ii$ of $\Gamma_\partyA$. Due to Lemma~\ref{lemma:unity_eigenvectors_of_antisym_submatrix}, the vectors $(\concat{v_j,0_2})_{j=1}^m$ are eigenvectors to eigenvalue $\ii$ of $\Gamma_{\partyA\partyB}$. Next, we construct two additional such eigenvectors. We have $\Gamma^{(i)\transpose}\Gamma^{(i)}= \id_{2n}$ since $\Gamma^{(i)}$ is the CM of a pure FGS. From this, we get the equations
\begin{align*}
    \eta^2 + x^\transpose x + \gamma^\transpose \gamma = \eta^2 + y^\transpose y +\delta^\transpose \delta &= 1, \\
    y^\transpose x + \delta^\transpose \gamma &= 0, \\
    -\eta y - \tilde\Gamma_\partyA x + S^\transpose \gamma = \eta x - \tilde\Gamma_\partyA y + S^\transpose \delta &= 0_{2(n-i-2)}, \\
    \eta \delta + S x + \Gamma_\partyB \gamma  = \eta \gamma - S y - \Gamma_\partyB \delta  &= 0_2.
\end{align*}
Thus, we get
\[
    \Gamma_{\partyA\partyB} \begin{pmatrix} 0 \\ \eta \\ x \\ \gamma \end{pmatrix} = \begin{pmatrix} -1 \\ 0 \\ 0_{2(n-i-2)} \\ 0_2 \end{pmatrix}
    \quad\text{and}\quad
    \Gamma_{\partyA\partyB} \begin{pmatrix} \eta \\ 0 \\ -y \\ -\delta \end{pmatrix} = \begin{pmatrix} 0 \\ 1 \\ 0_{2(n-i-2)} \\ 0_2 \end{pmatrix}
\]
Now define
\[
    w_1 = \begin{pmatrix} \ii\\ 0\\ 0_{2(n-i-2)} \\0_2 \end{pmatrix} + \begin{pmatrix} 0 \\ \eta \\ x \\ \gamma \end{pmatrix}
    \quad\text{and}\quad
    w_2 = \begin{pmatrix} 0 \\ -\ii \\ 0_{2(n-i-2)} \\ 0_2 \end{pmatrix} + \begin{pmatrix} \eta \\ 0 \\ -y \\ -\delta \end{pmatrix},
\]
which are clearly eigenvectors to eigenvalue $\ii$ of $\Gamma_{\partyA\partyB}$. Since $\gamma$ and $\delta$ are linearly independent, the same is true for $w_1$ and $w_2$, as well as either of them and any of the eigenvector $(\concat{v_j,0_2})$. We conclude that $\multiplicity{\ii}{\Gamma_{\partyA\partyB}} \geq m +2$, and
\[
    \lsr(\ket{\psi_i}, \partyA\partyB, \partyC) \leq \lsr(\ket{\psi_i}, \partyA, \partyB\partyC) - 1. 
\]
Together with Eq.~\eqref{eq:appendix_esed_entanglement_monogamy}, Eq.~\eqref{eq:proof_iesed_schmidt_rank_reduction} follows. 
\end{proof}

To close this appendix, we now state the result about the scaling of the number matchgates with the quantity $K= \sum_{i=1}^{n-1} \lsr_i( \ket{\psi})$. As we explain in the main text, at least $K/2$ gates of an arbitrary nearest neighbor gate set are necessary to generate the state. Thus, with matchgates, one obtains optimal scaling with this number.

\begin{corollary}
    When applying Algorithm~\ref{alg:outer_enhanced_symmetric_euler_decomposition} to a pure FGS $\ket{\psi}$ one obtains at most $K$ gates $(U_i)_{i=1}^K$. It holds that $\prod_i U_i \ket{\psi} = \ket{\cbs}$ for some computational basis state $\ket{\cbs}$. 
\end{corollary}

\begin{proof}
Each application of Algorithm~\ref{alg:outer_enhanced_symmetric_euler_decomposition} first checks whether the state $\ket\psi$ is a product $\ket\psi=\ket{\psi_1}\ket{\psi_2}$ on two regions with $n_1$ and $n_2$ qubits (possibly with $n_2 = 0$), where $\ket{\psi_1}$ does not factorize further. Algorithm~\ref{alg:outer_enhanced_symmetric_euler_decomposition} is then applied again to $\ket{\psi_2}$. We hence focus on $\ket{\psi_1}$. In the nontrivial case $n_1 \geq 2$, the sum over logarithms of Schmidt ranks of this state is at least $n_1 - 1$. As we have shown in Proposition~\ref{proposition:iesed_entanglement_reduction}, Algorithm~\ref{alg:inner_enhanced_symmetric_euler_decomposition} produces $n_1-1$ gates which reduce this sum by $n_1-1$ and disentangle the first two qubits from the remaining ones. The latter ones are then again used as the input to Algorithm~\ref{alg:outer_enhanced_symmetric_euler_decomposition}, and the argument can be repeated as many times until the initial state is fully disentangled. We remark that under some conditions, a gate outputted by an application of \texttt{elim\_last\_two} is a product of local gates. In this case, the gate can be left out, and hence $K$ is an upper bound on the number of gates.
\end{proof}

\section{RSF circuits are optimal for representing arbitrary pure FGS} \label{app:proof_rsf_circuits_optimal}

In this appendix, we give a proof of Theorem~\ref{thm:RSF_optimal}. To recall, the theorem states that under some mild conditions on the gates comprising a MGC in RSF, any state that is generated with that circuit cannot be generated with any other MGC that has fewer gates. As explained in the main text, the proof relies on comparing two circuits in RSF~\cite{MoLa25}. Doing so is sufficient, since with the absorption algorithm, any other MGC can be transformed into RSF. One then assumes that one RSF circuit has minimal gate count, and the gates in the other satisfy the conditions of Theorem~\ref{thm:RSF_optimal}. By then comparing diagonals of both RSF circuits, one must conclude that pairs of diagonals in both circuits coincide up to local gates. If this were not the case, a contradiction is obtained.

To compare two diagonals, we need two key ingredients. For the first one, we derive some simple conditions under which the nonlocal part of a matchgate is fully specified (Lemma~\ref{lemma:characterization_very_entangling} and Observation~\ref{lemma:orthogonal_matrix_determined}). The conditions are based on the CMs of the states before and after applying a gate. We hence also need to study the CMs of states generated by a particular RSF circuit.
For our purpose, it is sufficient to investigate the first two columns of a CM and their relation to the first diagonal of the RSF circuit (Lemma~\ref{lemma:RSF_state_covariance_matrix}). In Lemma~\ref{lemma:diagonals_coincide}, we combine these techniques and show that the first diagonals of the circuit of interest  and the optimal one have to coincide up to local gates. The proof of Theorem~\ref{thm:RSF_optimal} then follows as a corollary.

Throughout this appendix, subscripts on a gate denote the qubits on the gate acts on. A circuit $U$ in RSF with $m$ diagonals, parameters $((k_1,l_1),\ldots,(k_m,l_m))$ and gates $U^{(ij)}$ can hence be written as 
\[
	U = \prod_{i=m}^1 D_i,\quad\quad D_i = \prod_{j=1}^{l_i} U^{(ij)}_{k_i + j -1, k_i+j},
\]
where $D_i$ denotes the $i$-th diagonal. In the following, we denote by $\hat e_k$ the $k$-th standard basis vector. Given a CM $\Gamma$, we use the short notation $\Gamma_{(k,l)} = \Gamma_{[2k-1,2k],[2l-1,2l]}$ to denote the block corresponding to qubits $l$ and $k$ (i.e., the block with row and column indices $\{2k-1,2k\},\{2l-1,2l\}$ respectively). Node that $\Gamma_{(k,l)} = 0$ if and only if the reduced state on qubits $k$ and $l$ is a tensor product.

Before presenting the lemmas and their proofs, we recall the precise conditions imposed on the gates of the circuit considered in Theorem~\ref{thm:RSF_optimal}. For the first condition, it is required that each gate in the first layer produces an entangled state when applied to the input state. The second condition demands that the nonlocal parameters as defined in Eq.~\eqref{eq:parametrization_of_mg_local_nonlocal} of each further gate are not integer multiples of $\pi/2$. The following lemma restates this condition in terms of the orthogonal matrix associated with the gate. Specifically, we show that the offdiagonal blocks of that matrix have maximal rank if and only if the condition holds.

\begin{lemma} \label{lemma:characterization_very_entangling}
    Consider a matchgate
    \begin{equation} \label{eq:decomposition_mg_local_nonlocal_appendix}
        U = \exp(\ii \phi_1 Z)\!\otimes\!\exp(\ii \phi_2 Z) \, \exp(\ii\alpha X\!\otimes\!X + \ii \beta Y\!\otimes\! Y) \, \exp(\ii \phi_3 Z)\!\otimes\!\exp(\ii \phi_4 Z),
    \end{equation}
    and the corresponding orthogonal matrix $R$. Then, $\vert \det R_{[1:2],[3:4]} \vert =\vert \det R_{[3:4],[1:2]} \vert = \vert \sin(2\alpha)\sin(2\beta)\vert$. Hence, the blocks $R_{[1:2],[3:4]}$ and $R_{[3:4],[1:2]}$ have full rank if and only if both $\alpha$ and $\beta$ are not integer multiples of $\pi/2$.
\end{lemma}
\begin{proof}
    The orthogonal matrix $R'$ corresponding to a gate $\exp(\ii\alpha X\!\otimes\!X + \ii \beta Y\!\otimes\! Y)$ is given by
    \[
        R' = \begin{pmatrix}
            \cos(2\beta) & 0 & 0 & -\sin(2\beta) \\
            0 & \cos(2\alpha) & \sin(2\alpha) & 0 \\
            0 & -\sin(2\alpha) & \cos(2\alpha) & 0 \\
            \sin(2\beta) & 0 & 0 & \cos(2\beta)
        \end{pmatrix},
    \] with $\vert \det R'_{[1:2],[3:4]} \vert = \vert \det R'_{[3:4],[1:2]} \vert = \vert \sin(2\alpha)\sin(2\beta)\vert$. The rotation corresponding to any phase gate $\exp(\ii \phi_j Z_k)$ on qubit $k$ only acts nontrivially on the $2k-1$ and $2k$-th basis vector, and therefore cannot change the determinant of the block $R'_{[1:2],[3:4]}$. The converse direction is shown by constructing the matchgate corresponding to a given $R$ and decomposing it into the form given in Eq.~\eqref{eq:decomposition_mg_local_nonlocal_appendix}. With this, one gets the decomposition $R=(R_1 \oplus R_2) \, R' \, (R_3 \oplus R_4)$, with $R'$ as above, and $R_i$ corresponding to the phase gates. This completes the proof.
\end{proof}

Next, we study to which extent a matchgate is already specified given some knowledge of how it acts on a state, or to be more precise, of how the corresponding rotation matrix $R$ acts on the CM. A simple, but central observation is the following: Suppose $Rx$ and $Ry$ are the image of $x$ and $y$ under a $4\times 4$ rotation $R$. Then, any other rotation that maps $x\mapsto Rx$ and $y\mapsto Ry$ is of the form $R'R$, where $R'$ is a rotation in the orthogonal complement of the subspace spanned by $Rx$ and $Ry$. Here, we are interested in a weaker version where, rather than knowing $Rx$ and $Ry$, we only know that the images are contained in the subspace spanned by $\hat e_1$ and $\hat e_2$. Any other rotation that does the same is of the form $(S\oplus T) R$, where $S$ and $T$ are $2\times 2$ rotation (or reflection) matrices. Physically, this means that the nonlocal part of the corresponding matchgate is fully specified, and that the open degrees of freedom correspond to local phase and Pauli-$X$ gates. We will later reference this fact often and hence put it in following observation.

\begin{observation} \label{lemma:orthogonal_matrix_determined}
A $4\times 4$ orthogonal matrix $R$ is a direct sum $R=S\oplus T$ of $2\times 2$ matrices $S$ and $T$ iff the image of two linearly independent vectors in the subspace spanned by $\hat e_1$ and $\hat e_2$ remains in that subspace.
\end{observation}

Our next aim is now to relate the properties of MGC in RSF to their action on the CM. In particular, we want to emphasize the relation between the first diagonal of the circuit, and the first two columns of the CM. Consider, as an example, the state
\begin{equation*}
    \ket\psi = \vcenter{\hbox{
    \begin{tikzpicture}[scale = 0.8]
        \begin{scope}[rotate=90]
            \helpdrawCircuitLines{0}{7}{0}{1.75}
            \circuitInitXAt{0.25}
            \helpdrawRSFdiagonal{0}{5}
            \helpdrawMGColorToFaint
            \helpdrawRSFdiagonal{3}{4}
            \helpdrawRSFdiagonal{5}{2}
            \helpdrawMGColorReset
        \end{scope}
    \end{tikzpicture}}}
\end{equation*}
generated by an RSF circuit with three diagonals acting on $\ket{0^{n}}$. The state $\ket{\psi'}$ obtained when acting only with the two non-highlighted diagonals is clearly a product in some suitable bipartitions. In particular, since the first qubit is not correlated with any other qubit, the first two columns of $\Gamma'$ are zero (except for the first two rows thereof). When now applying the gates in the highlighted diagonal of length $l=5$, qubit $1$ can at most be correlated with qubit $l+1$. In other words, the blocks in the CM corresponding to qubits $1$ and $k$ with $k>l+1$ are zero. We can infer more structure if the gates in the circuit satisfy the conditions of Theorem~\ref{thm:RSF_optimal}. This is summarized in the following lemma.

\begin{lemma} \label{lemma:RSF_state_covariance_matrix}
Consider an MGC $U$ in RSF with parameters $((1,l), (k_2, l_2), \ldots, (k_m,l_m))$, and a computational basis state $\ket{\cbs}$. For the CM $\Gamma$ of $U\ket{\cbs}$, all $2\times 2$ blocks $\Gamma_{(k,1)}$ describing correlations between qubits $1$ and $k$, with $k > l+1$ are zero. 
If furthermore the first diagonal of $U$ satisfies the requirements of Theorem~\ref{thm:RSF_optimal}, the block $\Gamma_{(l+1,1)}$ has maximal rank. 
\end{lemma}

\begin{proof}
As the first diagonal is the only one that correlates the first qubit with any other, the first statement is clear: If the diagonal has a length of $l$, any qubit with index $k>l+1$ cannot share correlations with the first qubit (see the example above). For the second part, we need to trace the action of each gate in the first diagonal on the CM. We do this inductively over the length of the first diagonal. First we consider the case $l=1$. Here, the state $\ket\psi$ is a product $\ket\psi = \ket{\psi_\partyA}\ket{\psi_{\partyB}}$ between the first two and remaining $n-2$ qubits, and its CM $\Gamma$ is a direct sum $\Gamma = \Gamma_\partyA\oplus\Gamma_\partyB$. Due to the conditions on the gates, $\ket{\psi_\partyA}$ is entangled, and the offdiagonal $2\times2$ block of $\Gamma_\partyA$ has maximal rank \cite{BoRe04} (see also Observation~\ref{observation:lsr_pure_fgs}). Now we consider the case of a diagonal with length $l$. To this end, suppose the statement is true for all RSF circuits in which the first diagonal has length $l-1$. Showing the statement also for length $l$ is straightforward to see as follows: Decompose the circuit $U$ into $U = G U'$, where $G$ is the last gate in the first diagonal of $U$ (in the example given above, the last gate in the highlighted diagonal), and $U'$ consists of all remaining gates. Note that $U'$ is likewise in RSF, and its first diagonal has length $l-1$. Denote by $R$ the orthogonal matrix corresponding to $G$, with $\Gamma'$ the CM of the state before applying $G$, and with $\Gamma = R \Gamma'R^\transpose$ the CM of $GU'\ket{\cbs}$. We know that the block $\Gamma'_{(l,1)}$ has maximal rank, and that $\Gamma'_{(l+1,1)} = 0$. Therefore, $\Gamma_{(l+1,1)} = R' \Gamma'_{(l,1)}$, where $R'$ is the corresponding offdiagonal block of $R$. Since $G$ has nontrivial nonlocal parameters, $\Gamma_{(l+1,1)}$ also has maximal rank.
\end{proof}

We use these results now to show the subsequent lemma, which will be the main tool to prove Theorem~\ref{thm:RSF_optimal}. As explained above, the idea will be to compare the first diagonals of two circuits in RSF that generate the same state $\ket\psi$. One of the circuits has optimal gate count to generate $\ket\psi$, i.e., no other circuit acting on any possibly different computational basis state generates $\ket\psi$ with fewer gates (recall that we allow arbitrary computational basis states as initial states). The other circuit satisfies the conditions of Theorem~\ref{thm:RSF_optimal}. From the latter, we can infer some properties of the first two columns of the CM corresponding to the generated state. As the optimal circuit has produce these same properties, the first diagonals have to coincide.

\renewcommand{\cbs}{c}

\begin{lemma}\label{lemma:diagonals_coincide}
    Consider an MGC $U$ in RSF $((k_1,l_1),\ldots,(k_m,l_m))$, s.t. its first diagonal satisfies the conditions of Theorem~\ref{thm:RSF_optimal}, and an FGS $\ket\psi = U \ket{b}$, where $\ket{b} = \ket{b_1\ldots b_n}$ is a computational basis state. Consider furthermore another MGC $V$ in RSF $((h_1,r_1),\ldots,(h_p,r_p))$ s.t. $\ket\psi = V \ket{\cbs}$ for possibly another computational basis state $\ket\cbs = \ket{\cbs_1\ldots\cbs_n}$ and suppose that $V$ has optimal gate count for generating $\ket\psi$. Then, $h_1 = k_1$ and $r_1 = l_1$. Furthermore, if $l_1>1$, the first diagonals apart from the first gate of $U$ and $V$ coincide up to local gates $T^{(j)}$, i.e., 
    \begin{equation}\label{eq:diagonal_coincide}
	\prod_{j=2}^{l_1} V_{k_1 +j-1, k_1+j}^{(1j)} = \prod_{j=2}^{l_1} U_{k_1 +j-1, k_1+j}^{(1j)} \prod_{j=1}^{l_1}T_{k_1+j}^{(j)}.
\end{equation}
    Each gate $T^{(j)}$ is of the form $\exp(\ii\phi_j Z) P$ with $P\in\{\id_2, X\}$.
\end{lemma}

\begin{proof} We proceed in three steps. First, we show $h_1 = k_1$, i.e., that the locations of the first diagonals coincide. Second, we show also that the lengths $l_1$ and $r_1$ are equal. In the third step, we then prove Eq.~\eqref{eq:diagonal_coincide}. \par\medskip

\textbf{Step 1.} We show $k_1 = h_1$ by deriving contradictions from the alternatives $h_1 < k_1$ and $h_1 > k_1$. Suppose that $h_1 < k_1$. W.l.o.g., we set $h_1= 1$ as this corresponds just to ignoring the first qubits on which neither $U$ nor $V$ act. Since then $U$ does not act on the first qubit, the state $\ket\psi = U\ket{b} = \ket{b_1}\ket\phi$ is a product. Only a single gate in $V$ acts on the first qubit, specifically the first gate in the first diagonal $V^{(11)}$. Therefore, $V^{(11)}\ket{\cbs_1\cbs_2} = \alpha \ket{b_1}\ket{\cbs'}$, where $\ket{\cbs'}$ is a computational basis state due to parity preservation and $\alpha$ is a phase. By replacing the first two bits of $\ket{\cbs}$ with $\ket{b_1}\ket{\cbs'}$ and acting with the circuit obtained by omitting $V^{(11)}$ from $V$, one generates $\ket{\psi}$ with one gate less. This is a contradiction to $V$ having the lowest possible matchgate count, and we have $h_1\geq k_1$.

	Suppose now that $h_1 > k_1$. From Lemma~\ref{lemma:RSF_state_covariance_matrix}, we know that the block in the CM corresponding to qubits $k_1$ and $k_1+l_1$ is non-vanishing, i.e., that these qubits share correlations. On the other hand, since $V$ never acts on qubit $k_1$, it cannot produce these correlations. Hence $V\ket{\cbs} \neq \ket\psi$ and we conclude that $h_1 = k_1$. W.l.o.g., set from now on $k_1=h_1 = 1$.

\par\medskip
\textbf{Step 2.} Here, we show $r_1=l_1$ by contradiction. Assume first $r_1 < l_1$. Due to Lemma~\ref{lemma:RSF_state_covariance_matrix}, we know that the block of the CM corresponding to qubits $1$ and $l_1+1$ is non-vanishing, i.e., that the reduced state of $\ket\psi$ on qubits $1$ and $l_1 + 1$ is not a tensor product. On the other hand, the state $V\ket{\cbs}$ cannot have correlations between these qubits since $r_1 < l_1$ (i.e., the reduced state of $V\ket{\cbs}$ on qubits $1$ and $l_1+1$ can only be a tensor product). We deduce that $V\ket{\cbs}$ cannot be $\ket\psi$ in this case, and hence we need $r_1 \geq l_1$.

Assume now $r_1 > l_1$. The idea is to show that $V^{(1\,l_1+1)}$, the $(l_1+1)$-th gate in the first diagonal of $V$, must be a product of two local gates. To this end, we will use of Observation~\ref{lemma:orthogonal_matrix_determined}. Decompose the first diagonal of $V$ into $V_3V_2V_1$ (we illustrate this below) s.t. \[V_2 = V^{(1\,l_1+1)}_{l_1+1,l_1+2}\] is the gate of interest, and 
\[V_1 = \prod_{j=1}^{l_1} V_{j,j+1}^{(1j)}\quad\text{and}\quad V_3= \begin{cases}\prod_{j=l_1+2}^{r_1} V_{j,j+1}^{(1j)}, &\text{ if } r_1 \geq l_1+2, \\
	\ID, &\text{ if } r_1 = l_1+1
	\end{cases} \]
contain the preceding and possible subsequent gates. Denote by $V_\text{L}$ the circuit obtained from $V$ when omitting the first diagonal. We can hence write
\[
    \begin{tikzpicture}[scale = 0.8]
        \node at (-1, 0.825) {$\ket \psi = $};
        \node at (4.25, 0.825) {$ = $};
        \node at (8.7, 0.825) {,};
        \node[anchor = south] at (1.75,1.75) {$U\ket{b}$};
        \node[anchor = south] at (6.75,1.75) {$V_3V_2V_1V_\text{L}\ket{\cbs}$};
        \begin{scope}[xshift=0cm, rotate=90]
            \helpdrawCircuitLines{0}{7}{0}{1.75}
            \circuitInitXAt{0.25}
            \helpdrawMGColorToFaint
            \helpdrawRSFdiagonal{2}{5}
            \helpdrawRSFdiagonal{5}{2}
            \helpdrawMGColorReset
            \helpdrawRSFdiagonal{0}{3}
        \end{scope}
        \begin{scope}[xshift=5cm, rotate=90]
            \helpdrawCircuitLines{0}{7}{0}{1.75}
            \circuitInitXAt{0.25}
            \helpdrawMGColorToFaint
            \helpdrawRSFdiagonal{3}{2}
            \helpdrawRSFdiagonal{5}{2}
            \helpdrawMGColorReset
            \helpdrawRSFdiagonal{0}{3}
            \circuitAdvanceXBy{\flatgateadvance}
            \circuitAdvanceXBy{\flatgateadvance}
            \circuitAdvanceXBy{\flatgateadvance}
            \circuitMultiGate[\flatgatesize][othergatecolorbg][othergatecolorfg][]{}{-3}{-4}{}
            \circuitAdvanceXBy{\flatgateadvance}
            \circuitMultiGate[\flatgatesize][white][black][]{}{-4}{-5}{}
        \end{scope}
    \end{tikzpicture}
\]
where the light blue, dark blue, red and framed gate(s) on the right hand side correspond to $V_\text{L}, V_1,V_2$ and $V_3$ respectively. Since the first diagonal of $U$ (left side of the illustration) satisfies the conditions of Theorem~\ref{thm:RSF_optimal}, we know due to Lemma~\ref{lemma:RSF_state_covariance_matrix} that the first two columns $C$ of the CM of $\ket\psi$ have the form
\[
C^\transpose = \begin{pmatrix}
	\alpha_1 & \ldots & \alpha_{2l_1} & x_1 & x_2 & 0 & \ldots& 0\\
	\beta_1 & \ldots & \beta_{2l_1} & y_1 & y_2 & 0 & \ldots &0 
	\end{pmatrix},
\]
where $\alpha_i, \beta_i$ are some numbers, and $x=(x_1,x_2)$ and $y = (y_1,y_2)$ are linearly independent. In particular, all blocks corresponding to qubits $1$ and $k\geq l_1+2$ vanish. These are exactly the qubits on which $V_3$ acts nontrivially. Hence, while $V_3^\dagger\ket\psi$ is not necessarily equal to $\ket\psi$, the first two columns of the CM of $V_3^\dagger\ket\psi$ coincide with those of $\ket\psi$. That is, we have $C = R^\transpose_3 C$ with $R_3$ being the orthogonal matrix corresponding to $V_3$. On the other hand, consider the state $V_1V_\text{L}\ket{\cbs}$. The circuit $V_1V_\text{L}$ is clearly also in RSF and its first diagonal has length $l_1$. Again, from Lemma~\ref{lemma:RSF_state_covariance_matrix} it follows that the first two columns $D$ of the CM of $V_1V_\text{L}\ket{\cbs}$ are 
\[
	D^\transpose = \begin{pmatrix}
		\gamma_1 & \ldots & \gamma_{2l_1} & \gamma_{2l_1+1} & \gamma_{2l_1+2} & 0 & \ldots & 0\\
		\delta_1 & \ldots & \delta_{2l_1} & \delta_{2l_1+1} & \delta_{2l_1+2} & 0 & \ldots & 0
	\end{pmatrix},
\]
with some numbers $\gamma_i,\delta_i$, and vanishing entries for rows corresponding to qubits $k\geq l_1+2$. Denote now by $R_2 = \id_{2l_1}\oplus R'_2 \oplus \id_{2n-2l_1-4}$ the orthogonal matrix corresponding to $V_2$ (recall that $V_2$ acts nontrivially only on qubits $l_1+1$ and $l_1+2$). From $V_3^\dagger \ket\psi = V_2V_1V_\text{L}\ket{\cbs}$, we have that $C = R_2 D$ and hence
\[ R_2'^\transpose	\begin{pmatrix}
		x_1 & y_1 \\
		x_2 & y_2 \\
		0 & 0\\0&0
	\end{pmatrix}
	=
	\begin{pmatrix}
		\gamma_{2l_1+1} & \delta_{2l_1+1} \\
		\gamma_{2l_1+2} & \delta_{2l_1+2} \\
		0 & 0\\0&0
	\end{pmatrix}.
\]
It follows from Observation~\ref{lemma:orthogonal_matrix_determined} that $R'_2$ is a direct sum of two $2\times 2$ orthogonal matrices, and that therefore $V_2$ is a tensor product of local (combination of phase and Pauli-$X$ gates). As we explain next, these gates can be absorbed into the remaining circuit $V_1V_\text{L}$, and the input state $\ket{\cbs}$. For phase gates, this is clear since multiplication with any matchgate yields another matchgate, while applying a phase gate to a computational basis state contributes only a global phase. Similarly, Pauli-$X$ gates can also be absorbed: For any matchgate $G$, the gates $(X \otimes \id_2) G (X \otimes \id_2)$ and $(\id_2 \otimes X) G (\id_2 \otimes X)$ are again matchgates. Consequently, by adding $\id_2=X^2$ to the circuit we can update the matchgates in the circuit until a Pauli-$X$ gate acts directly on the input state, where it simply flip the corresponding bit. Note that this procedure does not change the number of gates or the layout of the circuit. With this, we could generate $\ket\psi$ with a different circuit, acting possibly on a different computational basis state, with one fewer gate. This leads to a contradiction, and therefore $(k_1, l_1)=(h_1,r_1)$.

\par\medskip
\textbf{Step 3.} From the previous steps, we know that the locations and lengths of the first diagonals of $U$ and $V$ coincide. We now want to prove Eq.~\eqref{eq:diagonal_coincide}, i.e., that the diagonals coincide up to local gates. Our proof will be inductive. Let us first illustrate the main idea graphically for the first induction step. To this end, we write 
\[
    \begin{tikzpicture}[scale = 0.8]
        \node at (-1, 0.825) {$\ket \psi = $};
        \node at (4.25, 0.825) {$ = $};
        \node at (8.7, 0.825) {,};
        \node[anchor = south] at (1.75,1.75) {$U_\text{F} U_\text{R} \ket{b}$};
        \node[anchor = south] at (6.75,1.75) {$V_\text{F} V_\text{R}\ket{\cbs}$};
        \begin{scope}[xshift=0cm, rotate=90]
            \helpdrawCircuitLines{0}{7}{0}{1.75}
            \circuitInitXAt{0.25}
            \helpdrawMGColorToFaint
            \helpdrawRSFdiagonal{2}{5}
            \helpdrawRSFdiagonal{5}{2}
            \helpdrawMGColorReset
            \helpdrawRSFdiagonal{0}{3}
            \circuitAdvanceXBy{\flatgateadvance}
            \circuitAdvanceXBy{\flatgateadvance}
            \circuitAdvanceXBy{\flatgateadvance}
            \circuitMultiGate[\flatgatesize][white][black][]{}{-3}{-4}{}
        \end{scope}
        \begin{scope}[xshift=5cm, rotate=90]
            \helpdrawCircuitLines{0}{7}{0}{1.75}
            \circuitInitXAt{0.25}
            \helpdrawMGColorToFaint
            \helpdrawRSFdiagonal{3}{2}
            \helpdrawRSFdiagonal{5}{2}
            \helpdrawMGColorReset
            \helpdrawRSFdiagonal{0}{3}
            \circuitAdvanceXBy{\flatgateadvance}
            \circuitAdvanceXBy{\flatgateadvance}
            \circuitAdvanceXBy{\flatgateadvance}
            \circuitMultiGate[\flatgatesize][othergatecolorbg][othergatecolorfg][]{}{-3}{-4}{}
        \end{scope}
    \end{tikzpicture}
\]
where we decompose $U = U_\text{F} U_\text{R}$ and $V = V_\text{F} V_\text{R}$ into the last gates of the first diagonal (the framed and red gate), and all remaining gates, respectively.
Apply now the inverse of $U_\text{F}$ to $\ket\psi$. One obtains
\[
    \begin{tikzpicture}[scale = 0.8]
        \node at (-1, 0.825) {$U^\dagger_\text{F}\ket \psi = $};
        \node at (4.25, 0.825) {$ = $};
        \node at (9.25, 0.825) {$ = $};
        \node at (13.7, 0.825) {,};
        \node[anchor = south] at (1.75,1.75) {$U_\text{R} \ket{b}$};
        \node[anchor = south] at (6.75,1.75) {$U^\dagger_\text{F} V_\text{F} V_\text{R} \ket{\cbs}$};
        \node[anchor = south] at (11.75,1.75) {$(S \otimes T) V_\text{R} \ket{\cbs}$};
        \begin{scope}[xshift=0cm, rotate=90]
            \helpdrawCircuitLines{0}{7}{0}{1.75}
            \circuitInitXAt{0.25}
            \helpdrawMGColorToFaint
            \helpdrawRSFdiagonal{2}{5}
            \helpdrawRSFdiagonal{5}{2}
            \helpdrawMGColorReset
            \helpdrawRSFdiagonal{0}{3}
        \end{scope}
        \begin{scope}[xshift=5cm, rotate=90]
            \helpdrawCircuitLines{0}{7}{0}{1.75}
            \circuitInitXAt{0.25}
            \helpdrawMGColorToFaint
            \helpdrawRSFdiagonal{3}{2}
            \helpdrawRSFdiagonal{5}{2}
            \helpdrawMGColorReset
            \helpdrawRSFdiagonal{0}{4}
            \circuitAdvanceXBy{\flatgateadvance}
            \circuitAdvanceXBy{\flatgateadvance}
            \circuitAdvanceXBy{\flatgateadvance}
            \circuitMultiGate[\flatgatesize][othergatecolorbg][othergatecolorfg][]{}{-3}{-4}{}
            \circuitAdvanceXBy{\flatgateadvance}
            \circuitMultiGate[\flatgatesize][white][black][]{}{-3}{-4}{}
        \end{scope}
        \begin{scope}[xshift=10cm, rotate=90]
            \def\circuitGateHeight{\circuitlinespacing*0.5}
            \helpdrawCircuitLines{0}{7}{0}{1.75}
            \circuitInitXAt{0.25}
            \helpdrawMGColorToFaint
            \helpdrawRSFdiagonal{3}{2}
            \helpdrawRSFdiagonal{5}{2}
            \helpdrawMGColorReset
            \helpdrawRSFdiagonal{0}{3}
            \circuitAdvanceXBy{\flatgateadvance}
            \circuitAdvanceXBy{\flatgateadvance}
            \circuitAdvanceXBy{\flatgateadvance}
        \circuitSingleGate[\flatgatesize][othergatecolorbg][othergatecolorfg]{}{-3}{}
        \circuitSingleGate[\flatgatesize][othergatecolorbg][othergatecolorfg]{}{-4}{}
        \end{scope}
    \end{tikzpicture}
\]
with two local gates $S$ and $T$, where the last equality follows from a similar argument as above. In particular, since qubits $1$ and $l_1+1$ of $U_\text{F}^\dagger\ket{\psi}$ do not share correlations but the block in its CM corresponding to qubits $1$ and $l_1$ has maximal rank, the gate $U^\dagger_\text{F} V_\text{F}$ cannot be entangling. The gate $T$ corresponds to the last local gate in Eq.~\eqref{eq:diagonal_coincide}. One now absorbs $S$ in the next gate in the first diagonal of $V_\text{R}$, and repeats the same procedure. In the following, we explain this in detail.

Let us define circuits
\[U_\text{L} = \prod_{i=m}^2\prod_{j=1}^{l_i} U^{(ij)}_{k_i + j -1, k_i+j} \quad\text{and}\quad V_\text{L} = \prod_{i=p}^2\prod_{j=1}^{r_i} V_{h_i+j-1,h_i+j}^{(1j)} \]
obtained from $U$ and $V$ by omitting the first diagonals. We furthermore define for $s=2,\ldots, l_1+1$ the states
\[
	\ket{\psi^s} = \prod_{j=1}^{s-1} U_{j,j+1}^{(1j)} U_\text{L} \ket{b} \quad\text{ and }\quad
	\ket{\phi^s} = \prod_{j=1}^{s-1} V_{j,j+1}^{(1j)} V_\text{L} \ket{\cbs}
\]
obtained by applying $U_\text{L}$ ($V_\text{L}$), and the first $s$ gates of the respective first diagonals. We will demonstrate the existence of single qubit (combination of phase and Pauli-$X$) gates $S^{(s)}$ and $T^{(s)}$ s.t.
\begin{equation} \label{eq:diagonal_coincide_proof_recursion_step_2}
\ket{\psi^s} = S^{(s)}_{s} H^{(s)} \ket{\phi^{s}}
\quad\text{where}\quad
	H^{(s)} = \begin{cases}
		\prod_{t=s+1}^{l_1+1} T_{t}^{(t-1)}, &\text{if } s\leq l_1,\\
		\ID, &\text{if } s = l_1+1, 
	\end{cases}
\end{equation}
for $s=l_1+1,\ldots,2$. As a byproduct, we will obtain 
\begin{equation} \label{eq:diagonal_coincide_proof_recursion_step}
	(\ID_2 \otimes S^{(s+1)}) V^{(1s)} = U^{(1s)} (S^{(s)} \otimes T^{(s)}),
\end{equation}
i.e., some equations from which Eq.~\eqref{eq:diagonal_coincide} follows immediately. To start the induction, set $S^{(l_1+1)} = \id$. With this, Eq.~\eqref{eq:diagonal_coincide_proof_recursion_step_2} holds trivially for $s=l_1+1$. Assume now that Eq.~\eqref{eq:diagonal_coincide_proof_recursion_step_2} holds for a given $s+1$. From the definitions of $\ket{\psi^s}$ and $\ket{\phi^s}$, one can infer the recurrences
\[
\ket{\psi^{s}} = U_{s,s+1}^{(1s) \dagger} \ket{\psi^{s+1}} \quad\text{and}\quad
\ket{\phi^{s}} = V_{s,s+1}^{(1s) \dagger} \ket{\phi^{s+1}},
\]
and with Eq.~\eqref{eq:diagonal_coincide_proof_recursion_step_2} for $s+1$, one gets 
\[
	\ket{\phi^{s}}
    = \big(S^{(s+1)}_{s+1} V_{s,s+1}^{(1s)} \big)^\dagger H^{(s+1)\dagger} \ket{\psi^{s+1}} 
    = G^{(s)}_{s,s+1}H^{(s+1)\dagger}\ket{\psi^{s}}, 
\]
where we have set $G^{(s)} = [(\ID_2 \otimes S^{(s+1)}) V^{(1s)}]^\dagger U^{(1s)}$ (note also that $H^{(s+1)}$ and $U_{s,s+1}^{(1s)}$ act on disjoint qubits). We aim to show now that $G^{(s)}$ is a tensor product (see the illustration above). Since $\ket{\psi^s}$ and $\ket{\phi^s}$ are generated by MGCs in RSF, we get from Lemma~\ref{lemma:RSF_state_covariance_matrix} again that the first two columns $C$ and $D$ of their CMs are given by
\[
	C_s^\transpose = \begin{pmatrix}
		\alpha_1 & \ldots & \alpha_{2s-2} & x_1 & x_2 & 0 & \ldots & 0\\
		\beta_1 & \ldots & \beta_{2s-2} & y_1 & y_2 & 0 & \ldots & 0
	\end{pmatrix},
    \quad\text{and}\quad
    D_{s}^\transpose = \begin{pmatrix}
		\gamma_1 & \ldots & \gamma_{2s-2} & \gamma_{2s-1} & \gamma_{2s} & 0 & \ldots & 0\\
		\delta_1 & \ldots & \delta_{2s-2} & \delta_{2s-1} & \delta_{2s} & 0 & \ldots & 0
	\end{pmatrix},
\]
with some suitable numbers, respectively. The vanishing blocks correspond to the correlations of qubit $1$ with qubits $s+1,\ldots,n$. Note that $H^{(s+1)}$ acts nontrivially only on those qubits. Therefore, we get $D_s = R_s C_s$, with $R_s$ being the orthogonal matrix corresponding to $G^{(s)}_{s,s+1}$. Furthermore, $(x_1,y_1)$ and $(y_1,y_2)$ are linearly independent (Lemma~\ref{lemma:RSF_state_covariance_matrix}). It follows from Observation~\ref{lemma:orthogonal_matrix_determined} that $G^{(s)} = (S^{(s)} \otimes T^{(s)})^\dagger$ for some $S^{(s)}$ and $ T^{(s)}$. This proves the induction step and shows Eq.~\eqref{eq:diagonal_coincide_proof_recursion_step} for the given $s$.
\end{proof}

We now have established all the tools to prove Theorem~\ref{thm:RSF_optimal}.

\begin{proof}[Proof of Theorem~\ref{thm:RSF_optimal}]
Let $\ket{\psi} = U\ket{b}$, where 
\[ 
	U = D_1 \ldots D_m
\]
is the circuit in RSF $((k_1, l_1),\ldots, (k_m,l_m))$ with diagonals $D_i$ satisfying the conditions in Theorem~\ref{thm:RSF_optimal}, and $\ket {b} = \ket{b_1\ldots b_n}$ a computational basis state. Let $V$ be an MGC generating $\ket{\psi}$, that has minimal matchgate count for generating $\ket\psi$ when acting on any computational basis state. Due to the absorption algorithm~\cite{MoLa25} (Algorithm~\ref{alg:absorption}), we can assume w.l.o.g. $V$ to be in RSF $((h_1,r_1),\ldots,(h_p,r_p))$. That is, 
\[
    V=E_1 \ldots E_p
\]
where $E_i$ are the diagonals of $V$. We furthermore decompose the first diagonals $D_1= D_\text{R}G$ ($E_1 = E_\text{R} H$) into the first gate $D_\text{F}$ ($E_\text{F}$) and remaining gates $D_\text{R}$ ($E_\text{R}$) respectively.

To prove Theorem~\ref{thm:RSF_optimal}, we show that $U$ and $V$ have the same number of gates (and therefore, $U$ also has minimal matchgate count). To this end, we show that $U$ and $V$ have the same number of diagonals $p=m$, and for all $i=1,\ldots,m$, that the positions $h_i = k_i$ and lengths $r_i = l_i$ coincide. From this, the claim follows immediately.

First, we use Lemma~\ref{lemma:diagonals_coincide} and obtain that $h_1 = k_1$, $r_1 = l_1$ and that $E_\text{R} = D_\text{R} T$, where $T$ is a tensor product of local gates. We have seen in step~2 of the proof of Lemma~\ref{lemma:diagonals_coincide} how $T$ can be absorbed into the RSF circuit $E_2\ldots E_p$ and $\ket\cbs$ without changing the number of gates. That is, $T E_\text{F} E_2\ldots E_p \ket\cbs = E_F' E_2'\ldots E_p' \ket{\cbs'}$, with suitable $E'_i$, $E_\text{F}'$ and $\ket{\cbs'}$. Putting everything together, we have
\begin{equation} \label{eq:thmrsfoptimalproof1}
    \begin{tikzpicture}[baseline = 0.7cm,scale = 0.8]
        \node at (-1, 0.825) {$\ket \psi = $};
        \node at (4.25, 0.825) {$ = $};
        \node at (9.25, 0.825) {$ = $};
        \node at (14.25, 0.825) {$ = $};
        \node at (18.7, 0.825) {,};
        \node[anchor = south] at (1.75,1.75) {$D_\text{R}D_\text{F} D_2 \ldots D_m \ket{b}$};
        \node[anchor = south] at (6.75,1.75) {$E_\text{R}E_\text{F} E_2 \ldots E_p \ket{\cbs}$};
        \node[anchor = south] at (11.75,1.75) {$D_\text{R} TE_\text{F} E_2 \ldots E_p \ket{\cbs}$};
        \node[anchor = south] at (16.75,1.75) {$D_\text{R} E'_\text{F} E'_2 \ldots E'_p \ket{\cbs'}$};
        \begin{scope}[xshift=0cm, rotate=90]
            \helpdrawCircuitLines{0}{7}{0}{1.75}
            \circuitInitXAt{0.25}
            \helpdrawMGColorToFaint
            \helpdrawRSFdiagonal{2}{5}
            \helpdrawRSFdiagonal{5}{2}
            \helpdrawMGColorReset
            \helpdrawRSFdiagonal{0}{1}
            \circuitAdvanceXBy{\flatgateadvance}
            \circuitMultiGate[\flatgatesize][white][black][]{}{-1}{-2}{}
            \circuitAdvanceXBy{\flatgateadvance}
            \circuitMultiGate[\flatgatesize][white][black][]{}{-2}{-3}{}
            \circuitAdvanceXBy{\flatgateadvance}
            \circuitMultiGate[\flatgatesize][white][black][]{}{-3}{-4}{}
        \end{scope}
        \begin{scope}[xshift=5cm, rotate=90]
            \helpdrawCircuitLines{0}{7}{0}{1.75}
            \circuitInitXAt{0.25}
            \helpdrawMGColorToFaint
            \helpdrawRSFdiagonal{3}{2}
            \helpdrawRSFdiagonal{5}{2}
            \helpdrawMGColorReset
            \helpdrawRSFdiagonal{0}{1}
            \circuitAdvanceXBy{\flatgateadvance}
            \circuitMultiGate[\flatgatesize][othergatecolorbg][othergatecolorfg][]{}{-1}{-2}{}
            \circuitAdvanceXBy{\flatgateadvance}
            \circuitMultiGate[\flatgatesize][othergatecolorbg][othergatecolorfg][]{}{-2}{-3}{}
            \circuitAdvanceXBy{\flatgateadvance}
            \circuitMultiGate[\flatgatesize][othergatecolorbg][othergatecolorfg][]{}{-3}{-4}{}
        \end{scope}
        \begin{scope}[xshift=10cm, rotate=90]
            \def\circuitGateHeight{\circuitlinespacing*0.5}
            \helpdrawCircuitLines{0}{7}{0}{1.75}
            \circuitInitXAt{0.25}
            \helpdrawMGColorToFaint
            \helpdrawRSFdiagonal{3}{2}
            \helpdrawRSFdiagonal{5}{2}
            \helpdrawMGColorReset
            \helpdrawRSFdiagonal{0}{1}
            \circuitAdvanceXBy{\flatgateadvance}
            \circuitSingleGate[\flatgatesize][othergatecolorbg][othergatecolorfg]{}{-1}{}
            \circuitSingleGate[\flatgatesize][othergatecolorbg][othergatecolorfg]{}{-2}{}
            \circuitAdvanceXBy{\flatgateadvance}
            \circuitMultiGate[\flatgatesize][white][black][]{}{-1}{-2}{}
            \circuitSingleGate[\flatgatesize][othergatecolorbg][othergatecolorfg]{}{-3}{}
            \circuitAdvanceXBy{\flatgateadvance}
            \circuitMultiGate[\flatgatesize][white][black][]{}{-2}{-3}{}
            \circuitSingleGate[\flatgatesize][othergatecolorbg][othergatecolorfg]{}{-4}{}
            \circuitAdvanceXBy{\flatgateadvance}
            \circuitMultiGate[\flatgatesize][white][black][]{}{-3}{-4}{}
        \end{scope}
        \begin{scope}[xshift=15cm, rotate=90]
            \helpdrawCircuitLines{0}{7}{0}{1.75}
            \circuitInitXAt{0.25}
            \helpdrawMGColorToFaint
            \helpdrawRSFdiagonal{3}{2}
            \helpdrawRSFdiagonal{5}{2}
            \helpdrawMGColorReset
            \helpdrawRSFdiagonal{0}{1}
            \circuitAdvanceXBy{\flatgateadvance}
            \circuitMultiGate[\flatgatesize][white][black][]{}{-1}{-2}{}
            \circuitAdvanceXBy{\flatgateadvance}
            \circuitMultiGate[\flatgatesize][white][black][]{}{-2}{-3}{}
            \circuitAdvanceXBy{\flatgateadvance}
            \circuitMultiGate[\flatgatesize][white][black][]{}{-3}{-4}{}
        \end{scope}
    \end{tikzpicture}
\end{equation}
where the framed and red gates constitute $D_\text{R}$ and $E_\text{R}$ respectively.

It remains to show is that $T$ cannot be absorbed into $E_\text{F}E_2\ldots E_p \ket\cbs$ such that $TE_\text{F}E_2\ldots E_p \ket\cbs = V'\ket{\cbs'}$, where $V'$ is another MGC with fewer gates and $\ket{\cbs'}$ another computational basis state. This can be done by contradiction. Suppose $T E_\text{F}E_2\ldots E_p \ket\cbs = V'\ket{\cbs'}$, where $V'$ has fewer gates than $E_2 \ldots E_p$. Then, 
\[
    \ket\psi = E_\text{R} E_\text{F} E_2 \ldots E_p \ket{\cbs} = E_\text{R} T^\dagger T E_\text{F} E_2 \ldots E_p \ket\cbs = D_\text{R} V'\ket{\cbs'}.
\]
The circuit $D_\text{R} V'$ would generate $\ket\psi$ with fewer gates than $V$ and we have obtained the contradiction. 

We can now multiply with $D_\text{R}^\dagger$ in Eq.~\eqref{eq:thmrsfoptimalproof1}, and disregard $D_\text{F}$ and $E_\text{F}'$ since they act on disjoint qubits compared to all the other gates. This way, the problem is reduced to a similar instance with one fewer diagonal, since $D_2\ldots D_m$ satisfies the condition of Theorem~\ref{thm:RSF_optimal}, and $E_2'\ldots E_p'$ has minimal gate count. Applying the same argument hence another $m-1$ times completes the proof.
\end{proof}

\renewcommand{\cbs}{b}

\twocolumngrid

\section{Proof of Th.~\ref{thm:banded_cm_shallow_mgc}} \label{app:proof_shallow_banded}

In Theorem~\ref{thm:banded_cm_shallow_mgc}, we state that if an FGS has a $\bdns$-banded CM, then it can be generated by an MGC of depth $\mathcal{O}(\bdns)$. Likewise, the converse direction holds. The proofs are straightforward, but cumbersome since several cases need to be analyzed. We present this analysis here.

\begin{proof}[Proof of Theorem~\ref{thm:banded_cm_shallow_mgc}]
Consider a CM of a pure FGS which is $\bdns$-banded. To show the existence of a circuit of depth $\mathcal{O}(\bdns)$, we follow the action of Algorithm~\ref{alg:symmetric_euler_decomposition}. Note that it is sufficient to show that the CMs in all the intermediate steps are also $\bdns$-banded. Indeed, suppose that this is the case. Eliminating entries in a given column can be done with at most $\bdns-1$ Givens rotations. Hence, disentangling a single qubit from the others requires a diagonal of length at most $\lceil \frac{\bdns-1}{2} \rceil$. To obtain a circuit in RSF, two such diagonals starting on neighboring qubits are combined in a single diagonal of length at most $\max\{\lceil \frac{\bdns-1}{2} \rceil+1, \lceil \frac{\bdns-1}{2} \rceil\} = \lceil \frac{\bdns+1}{2} \rceil$. Since the diagonals obtained this way start on non-overlapping qubits, the output circuit depth can be bounded by $\lceil \frac{\bdns+1}{2} \rceil$.

What is left to do is to show that all the intermediate CMs, denoted by $\Gamma^{(q,k)}$ as in Algorithm~\ref{alg:symmetric_euler_decomposition}, are still $\bdns$-banded. For simplicity, we illustrate the case $q=1$, cases for the remaining values of $q$ can be treated analogously. We also omit writing the index `$q$' in the following, so we write $\Gamma^{(k)}_{ij}$, where the subindices indicate the row and column of the matrix. First, note that, while $\bdns-1$ Euler rotations are sufficient to disentangle the first qubit, fewer are necessary in case there are more zeros in the column. Let $R^{(k)}$ be the rotations as defined in the algorithm and, likewise, $r$ be the maximum row index for which the first column of the covariance matrix is nonzero.

Consider $\Gamma^{(k)}_{ij}$ for some $i,j,k$ with $|i-j|>b$. W.l.o.g. $i>j$, since the covariance matrix is antisymmetric, and $i>\bdns+1$ (since otherwise $|i-j|\leq\bdns$). From this conditions, we have $i>k,r$. Now, we shall consider the different values that $j$ may take and prove that for each of them the matrix element is zero.

\begin{enumerate}
    \item $j\leq k-1$: if $k=r$ then $\Gamma_{ij}^{(r)}=0$ by definition of $r$. Instead if $k<r$, recall how $\Gamma^{(k)}$ is defined by the algorithm in terms of $\Gamma^{(k+1)}$, $\Gamma_{ij}^{(k)}=\sum_{\mu\nu}R_{i\mu}^{(k+1)}\Gamma_{\mu\nu}^{(k+1)}R_{j\nu}^{(k+1)}=\Gamma_{ij}^{(k+1)}$, where we use that for the given $i,j$, $R_{i\mu}^{(k+1)}=\delta_{i\mu}$ and $R_{j\nu}^{(k+1)}=\delta_{j\nu}$, since $R^{(k+1)}$ is only nontrivial in $k,k+1$. Therefore, we can repeat this argument until reaching $r$, at which point we know that the matrix element is 0: $\Gamma_{ij}^{(k)}=\Gamma_{ij}^{(k+1)}=\dots=\Gamma_{ij}^{(r)}=0$.
    \item $j=k$: Suppose that $\Gamma_{ik}^{(k)}\neq 0$, and let us find a contradiction. We have that $\Gamma_{i,k-1}^{(k-1)}=\sum_{\mu\nu} R_{i,\mu}^{(k)}\Gamma_{\mu\nu}^{(k)}R_{k-1,\nu}^{(k)}$, where only the terms with $\nu=k,k-1$ are nonzero. This yields two terms, but since $\Gamma_{i,k-1}^{(k)}R_{k-1,k-1}^{(k)}$ is zero by case 1, we get $\Gamma_{i,k-1}^{(k-1)}=\Gamma_{i,k}^{(k)}R_{k-1,k}^{(k)}$. Similarly, we can prove $\Gamma_{i,k-2}^{(k-2)}=\Gamma_{i,k-1}^{(k-1)}R_{k-2,k-1}^{(k-1)}$. Iterating this argument we get to $\Gamma_{i,2}^{(2)}=R_{2,3}^{(3)}\dots R_{k-1,k}^{(k)}\Gamma_{i,k}^{(k)}\neq 0$. However, we know that $\Gamma_{i,2}^{(2)}=0$ by the way the CM is modified by the algorithm. Therefore, since all rotations are nonzero by construction, $\Gamma_{i,k}^{(k)}=0$.
    \item $j=k+1$: We have that $\Gamma_{i,k+1}^{(k)}=\sum_{\mu\nu} R_{i,\mu}^{(k+1)}\Gamma_{\mu\nu}^{(k+1)}R_{k+1,\nu}^{(k+1)}=\Gamma_{i,k}^{(k+1)}R_{k+1,k}^{(k+1)}+\Gamma_{i,k+1}^{(k+1)}R_{k+1,k+1}^{(k+1)}$, where the first term is 0 due to case 1 and the second term is 0 due to case 2. Therefore, $\Gamma_{i,k+1}^{(k)}=0$.
    \item $j>k+1$: Similarly as for case 1, we have $\Gamma_{ij}^{(k)}=\Gamma_{ij}^{(k+1)}=\dots =\Gamma_{ij}^{(r)}=0$.
\end{enumerate}

The converse direction can easily be seen with an argument using a lightcone. Suppose $\ket{\psi} = U \ket{\cbs}$ is generated with an MGC $U$ of depth $\depth$. We show that the corresponding CM $\Gamma$ is $4\depth+5$-banded. Consider the entry $\Gamma_{kl} = \ii/2 \bra{\cbs} U^\dagger [c_{k}, c_l] U \ket{\cbs}$, with $\vert k-l\vert > 4\depth +5$. When conjugating a Majorana operator $c_k$ with a circuit of depth~$\depth$, one can only obtain a sum of $2\depth+2$ neighboring Majorana operators. The entry $\Gamma_{kl}$ can hence be written as a sum over contributions of the form $\bra{\cbs} [c_i, c_j]\ket{\cbs}$, with $|i-k| \leq 2\depth+2$ and $|j-l| \leq 2\depth+2$. For any pair of indices $(i,j)$ appearing in the sum it holds that $4\depth+5  < |k - l| \leq |k-i|  + |i-j| + |j-l| \leq |i-j| + 4\depth + 4$. Thus, $|i-j| > 1$, and any of the $c_i c_j$ cannot be proportional to any Pauli-$Z$ operator. Those are, however, the only product of two Majorana operators which give non-zero expectation value for computational basis states, and hence $\Gamma_{kl} = 0$.
\end{proof}

\section{Cutting algorithm with degenerate Williamson eigenvalues} \label{app:cutting_alg_degenerate_exact}

In the main text, we have described an algorithm that takes as an input a $\bdns$-banded covariance matrix $\Gamma$ of a pure FGS $\ket{\psi}$, and outputs a circuit of depth $\mathcal{O}(\bdns)$ that generates $\ket{\psi}$. In one crucial step of the algorithm, the qubits are divided into three parties of consecutive qubits $\partyA$, $\partyB$ and $\partyC$. So far, we have assumed that the Williamson eigenvalues $\lambda^{\partyB}_i$ of $\Gamma_{\partyB\partyB}$ with $|\lambda^{\partyB}_i| < 1$ are non-degenerate. With this assumption. One can then find a rotation acting only on the reduced CM $\Gamma_{\partyB\partyB}$ of $\partyB$, that brings $\Gamma_{\partyB\partyB}$ to Williamson normal form, and sorts the Williamson eigenvalues such that the first $l$, and the last $r$ eigenvalues correspond to entangled pairs shared between $\partyB$ and $\partyA$, and $\partyB$ and $\partyC$ respectively. Here, we aim to modify the algorithm to work without any assumptions on the Williamson eigenvalues. Note that it is of course possible to add an additional few layers of random gates first (which then need to be added to the output circuit as well) in order to break the degeneracy of the eigenvalues. However, one can also address the problem of degenerate eigenvalues directly, as we will explain here.

We consider the CM
\[
\Gamma = \begin{pmatrix} 
		\Gamma_{\partyA\partyA} & -\Gamma_{\partyB\partyA}^\transpose & 0 \\
		\Gamma_{\partyB\partyA} & \Gamma_{\partyB\partyB} & -\Gamma_{\partyC\partyB}^\transpose \\
		0 &  \Gamma_{\partyC\partyB} & \Gamma_{\partyC\partyC}
	\end{pmatrix},
\]
where $\Gamma_{kk}$ is the reduced CM corresponding to party~$k\in\{\partyA,\partyB,\partyC\}$. Since $\Gamma$ is $\bdns$-banded, by assumption $\partyB$ must hold more than $\bdns +2$ qubits, and hence the size of $\Gamma_{\partyB\partyB}$ is at least $2\bdns+2$. For this reason, the lower left and upper right blocks of $\Gamma$ vanish. Denote by $\lambda_i^{k}$ the Williamson eigenvalues of each block.

Determining a MGC $U_\partyB$ acting only on $\partyB$ such that
\[
    \id \otimes U_\partyB \otimes \id \ket{\psi} = \ket{\psi_{\partyA\partyB}} \ket{\psi_{\partyB\partyC}}
\]
can be done as follows. The first step is to determine rotations $R_k$ that take $\Gamma_{kk}$ to Williamson normal form~\footnote{Note that computing $R_\partyA$ and $R_\partyC$ is a convenient step for the following, but to disentangle $\ket{\psi}$ as above it is not necessary to apply any operations on $\partyA$ and $\partyC$.} (note that those are not unique), and define a CM $\tilde \Gamma = (R_\partyA \oplus R_\partyB \oplus R_\partyC) \, \Gamma \, (R_\partyA \oplus R_\partyB \oplus R_\partyC)^\transpose$. The second step is to determine an additional rotation to apply on $\partyB$ which resolves the degeneracies and sorts the qubits depending on whether they are entangled to $\partyA$ or $\partyC$. This can be achieved as follows: Consider the matrix $M$ obtained by removing all vanishing rows from the matrix $\begin{pmatrix} \tilde \Gamma_{\partyB\partyA} & - \tilde \Gamma^\transpose_{\partyC\partyB} \end{pmatrix}$, where $\tilde\Gamma_{ij}$ label the blocks of $\tilde \Gamma$ corresponding to parties $i,j \in \{\partyA, \partyB, \partyC\}$ (note that the empty rows correspond to non-entangled qubits). Then, compute a QR decomposition of $M$, $QD = M$, where $Q$ is orthogonal and $D$ is upper triangular. The desired transformation is given by $Q^\transpose$ acting on the non-vanishing rows. We demonstrate this in the following.

The main idea is find another rotation $\tilde Q$ that separates and sorts the qubits 
and then argue that this $\tilde Q$ also performs a QR decomposition $M = \tilde Q \tilde D$, where $\tilde D$ is diagonal and has full rank. The statement then follows since QR decomposition of invertible matrices are unique apart from multiplication with a diagonal matrix with entries $\pm1$. Note that the procedure we will describe is an alternative to the algorithm based on computing the QR decomposition. We find numerically that an algorithm based on this procedure performs better in some cases when applied to approximately banded covariance matrices. Our approach is based on the constructions given in Ref.~\cite{BoRe04}, adapted to the three-party case as considered here.

Denote by $m_i^k$ the degeneracy of the Williamson eigenvalue $\lambda_i^k$. Consider an eigenvalue $\lambda_i^\partyB$ with $|\lambda_i^\partyB|<1$ of $\Gamma_{\partyB\partyB}$. Then, by Observation~\ref{observation:threepartite_bore}, $\lambda_i^\partyB$ is also an eigenvalue of $\Gamma_{\partyA\partyA}$ and/or $\Gamma_{\partyC\partyC}$, with the sum of both multiplicities being $m_i^\partyA + m_i^\partyC = m^\partyB_i$. The interesting case is when both $m_i^\partyA \neq 0$ and  $ m_i^\partyC\neq0$. Consider the matrix $M_{\lambda_i^\partyB}$ obtained when selecting rows and columns of $\tilde \Gamma$ in which this eigenvalue appears in the blocks on the diagonal, i.e.,
\[
    M_{\lambda_i^\partyB} = \begin{pmatrix} 
		\lambda^\partyB_i J_{2m_i^\partyA} & -G_\partyA^T & 0 \\
		G_\partyA & \lambda^\partyB_i J_{2m_i^\partyB} & G_\partyC \\
		0 & -G_\partyC^T & \lambda^\partyB_i J_{2m_i^\partyC}
	\end{pmatrix}
\]
with $J_{2m} = \bigoplus_{k=1}^m \covblock$. Combine the blocks $G_\partyA$ and $G_\partyC$ into a matrix $G = \begin{pmatrix}G_\partyA & G_\partyC \end{pmatrix}$. We now investigate some properties of $G$. Note that $M_{\lambda_i^\partyB}$ is the CM of some pure FGS. From $M_{\lambda_i^\partyB}M_{\lambda_i^\partyB}^\transpose = \id$, we get 
\begin{align*}
    G_\partyA^\transpose G_\partyA &= (1-(\lambda^\partyB_i)^2) \ID_{2m^\partyA_i}, \\
    G_\partyC^\transpose G_\partyC &= (1-(\lambda^\partyB_i)^2) \ID_{2m^\partyC_i}, \\
    G_\partyA^\transpose G_\partyC &= 0, \\
    \lambda^\partyB_i (G_\partyA J_{2m^\partyA_i} + J_{2m^\partyB_i} G_\partyA) &= 0, \\
    \lambda^\partyB_i (G_\partyC J_{2m^\partyC_i} + J_{2m^\partyB_i} G_\partyC) &= 0.
\end{align*}
Hence $G^\transpose G = (1-(\lambda^\partyB_i)^2) (\ID_{2m^\partyA_i} \oplus \ID_{2m^\partyC_i}) = (1-(\lambda^\partyB_i)^2) \ID_{2m^\partyB_i}$. In case $\lambda^\partyB_i = 0$, $G$ is an orthogonal transformation that, when applied on $\partyB$, separates state into pairs shared between parties $\partyA$, $\partyB$, and parties $\partyB$, $\partyC$. Suppose now that $\lambda^\partyB_i \neq 0$. Multiplying the fourth equation from the left with $G_C^\transpose$ or $G_\partyA^\transpose$, one gets $G_\partyC^\transpose J_{2m^\partyB_i} G_\partyA = 0$ or 
$G_\partyA^\transpose J_{2m^\partyB_i} G_\partyA = - (1-(\lambda^\partyB_i)^2) J_{2m^\partyA_i}$
respectively. Similar identities can be obtained from the fifth equation, and one obtains 
\[
    G^\transpose J_{2m^\partyB_i} G = - (1-(\lambda_i^\partyB)^2) J_{2m^\partyB_i}.
\]
Hence, the desired transformation is given by $\tilde G=(1-(\lambda_i^\partyB)^2)^{-1} G$.

To summarize, in order to find the MGC acting on party $\partyB$ that disentangles $\ket{\psi}$ somewhere within $\partyB$, the following steps are necessary. First, diagonalize all the reduced CMs, and compute the transformed CM of the whole state. Next, for each of the $m$ unique Williamson eigenvalue, select the offdiagonal blocks corresponding to the eigenvalue and compute the transformation $\tilde G_j$ as explained above. Finally, compute a permutation matrix $S$ that swaps pairs entangled to party $\partyA$ ($\partyC$) to the first (last) positions. The MGC acting on party $\partyB$ is the the one corresponding to $S\, (\tilde G_1 \oplus \ldots \oplus \tilde G_m) \, R_\partyB$.

\section{Variations of the cutting algorithm in the approximately banded case} \label{app:cutting_alg_degenerate_approx}

In the main text, we described several modifications to the entanglement cutting algorithm that allow it to be applied to states with approximately banded CMs. There are many possible approaches one could take. Our aim, however, is to ensure that any modification preserves the property that the algorithm outputs the exact circuit whenever the CM is exactly banded. The modifications introduced so far include the incorporation of a block-size parameter $s$, which by construction determines the circuit depth, and a procedure for determining whether a given pair is shared between the block of interest and its preceding or succeeding qubits. Here, we describe two additional modifications that become particularly important when the Williamson eigenvalues of a block are degenerate. They build on tools described in Appendix~\ref{app:cutting_alg_degenerate_exact}.

As before, let us consider a CM $\Gamma$ on three parties $\partyA$, $\partyB$, and $\partyC$. The first step of applying the rotations that diagonalize the reduced CMs of each party remains the same as before. The rotated CM becomes
\[
    \tilde \Gamma = \begin{pmatrix} 
		\bigoplus_i \lambda_i^\partyA J_2 & -\tilde\Gamma_{\partyB\partyA}^\transpose & -\mathcal{E}^\transpose \\
		\tilde\Gamma_{\partyB\partyA} & \bigoplus_i \lambda_i^\partyB J_2& -\tilde\Gamma_{\partyC\partyB}^\transpose \\
		\mathcal{E} &  \tilde\Gamma_{\partyC\partyB} & \bigoplus_i \lambda_i^\partyC J_2
	\end{pmatrix},
\]
where $\lambda_{i}^{k}$ are the Williamson eigenvalues corresponding to party $k\in\{\partyA, \partyB,\partyC\}$, and $\mathcal{E}$ is a matrix that vanishes in the exactly banded case. For our numerical studies, we consider a precision $\varepsilon_\lambda$ for the Williamson eigenvalues, i.e., we assume $\lambda=1$ if $\vert1-\lambda \vert < \varepsilon_\lambda$. Such eigenvalues correspond to qubits that are almost in a computational basis state. In the following, we describe two algorithms for finding an additional rotation to apply to $\partyB$ that approximately disentangles the state.

In some cases we find it to be sufficient to compute a QR decomposition of the the matrix $\begin{pmatrix} \tilde\Gamma_{\partyB\partyA} & -\tilde\Gamma_{\partyC\partyB}^\transpose \end{pmatrix} = QD$, where $Q$ is a rotation and $D$ is upper diagonal with non-negative coefficients on the diagonal. In the previous appendix, we have explained why this works in the exact case, where $D$ becomes diagonal. In order to determine the exact position of the splitting in which the state is disentangled after applying $Q^\transpose$, an additional step is to determine how many entangled pairs are shared with parties $\partyA$ and $\partyC$. This can be done as follows: First, sort the eigenvalues of $\partyB$ and, using the cutoff parameter $\varepsilon_\lambda$, determine the number $m$ of non-trivial eigenvalues of $\partyB$ (i.e., those that do not correspond to computational basis states). Second, sort the combined eigenvalues of $\partyA$ and $\partyC$, while keeping track of which party an eigenvalue belongs to. Third, pick the smallest $m$ of those eigenvalues. Those eigenvalues correspond to the pairs that contribute the most to the entanglement. Among those eigenvalues, count how many of them correspond to $\partyA$ and $\partyC$ respectively.

Another algorithm to find the disentangling rotation is based on a variation of the method described in Appendix~\ref{app:cutting_alg_degenerate_exact}. The idea is to first determine the blocks corresponding to each possibly degenerate Williamson eigenvalue, then find a rotation that separates the pair in this block, and finally apply a permutation to disentangle the state. Note that, in the non-exact case, it is not guaranteed that an eigenvalues belonging to $\partyB$ correspond exactly to those of $\partyA$ or $\partyC$. Hence, eigenvalues would need to be matched based on how close they are. However, such a procedure carries the risk of leaving unpaired eigenvalues that contribute a lot to entanglement. Instead, a variation of the counting procedure described above can be used to perform the matching. To this end, the first three steps are the same as above, and the matching is then given by pairing up the sorted list of eigenvalues of $\partyB$ with the combined sorted list of $\partyA$ and $\partyC$. A second cutoff parameter $\varepsilon_\text{deg}$ is used to decide which consecutive eigenvalues are considered to be degenerate. This can either be done using the list of values of $\partyB$ or the combined list of $\partyA$ and $\partyC$. In our studies, we get slightly better results by doing the latter. Once all the blocks have been determined, one performs a QR decomposition of each block. Finally, as in Appendix~\ref{app:cutting_alg_degenerate_exact}, a permutation is used to map pairs shared with $\partyA$ to one end of $\partyB$, and those shared with $\partyC$ to the other end.

\section{Further numerical examples of the approximate cutting algorithm} \label{app:further_numerical_examples}

In the main text, we applied the approximate cutting algorithm to the Ising model in the disordered phase, where correlations between qubits decay exponentially with the distance. One might, hence, ask whether the cutting algorithm still could perform well if correlations decay much slower, e.g., algebraically. To this end, we apply here the algorithm to the ground state of two additional models. First, we study a long-range Kitaev model, whose ground state exhibits algebraically decaying correlations with a tuneable exponent~\cite{VoLe14,VoLe16}. In this particular model, our results show that the required circuit depths depend on the decay exponent, with a scaling better than $\mathcal{O}(n)$ in some cases. For the second example, we show how our algorithm can be extended to study higher-dimensional systems. We consider the ground state of a Hamiltonian in two dimensions by mapping it to a one-dimensional chain through relabeling of modes. Our results indicate that it is not necessary to have exponentially decaying correlations in the CM in order for the approximate cutting algorithm to output circuits with a better depth scaling than standard methods.

\begin{figure*}[ht!]
    \centering
    \includegraphics[width=\textwidth]{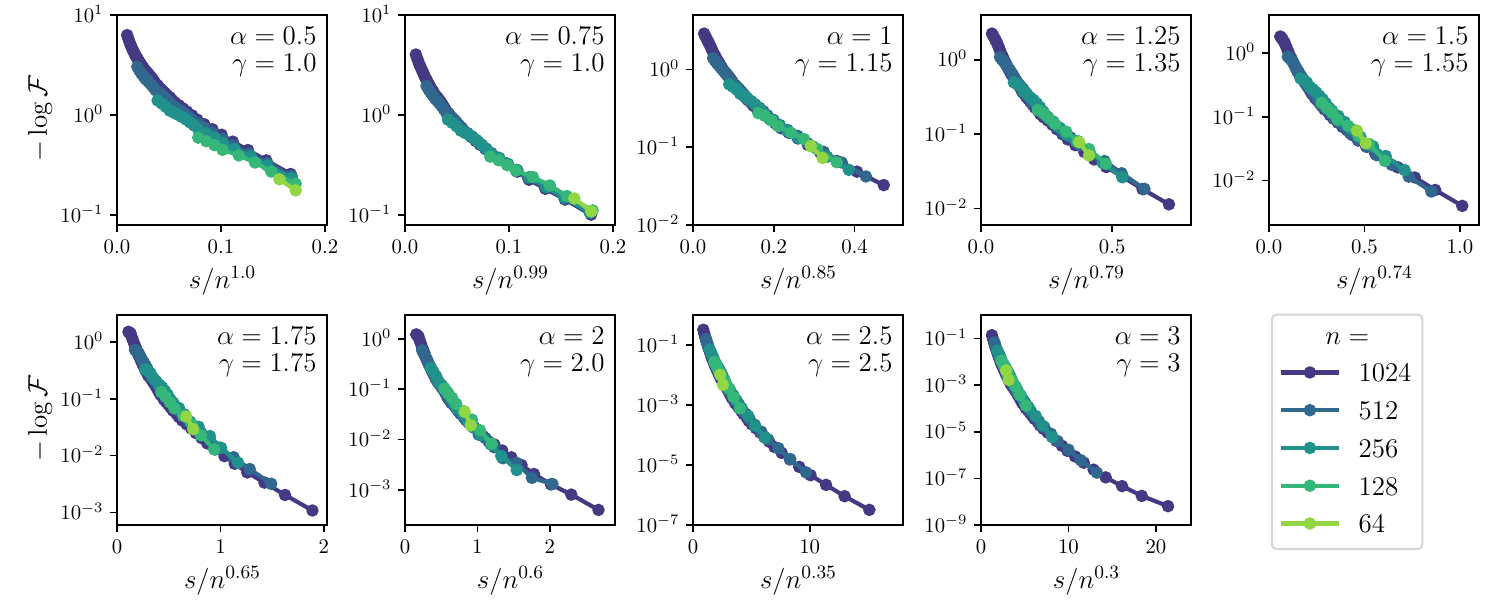}
    \caption{
    Logarithm of the fidelity of the reconstructed ground state for the long range Kitaev model in a regime where correlations decay algebraically (with an exponent of $\gamma$), as a function of block size $s$ rescaled by some function of the system size $n$. Data is presented for various system sizes and hopping decay exponents $\alpha$, at a fixed magnetic field strength of $g=2$.
    }
    \label{fig:algebraic_decay_1}
\end{figure*}

We start considering the Hamiltonian of the long-range Kitaev model~\cite{VoLe14,VoLe16}, described by
\[
H_\mathrm{LRK} = 
-\sum_{i<j} \frac{1}{\vert i - j\vert^\alpha}\majo{2i}\majo{2j+1} - g\sum_i \majo{2i}\majo{2i+1}. 
\]
Note that in the limit $\alpha\to\infty$, one obtains the Ising model. For $g=2$, the ground-state correlations decay algebraically for different values of $\alpha$, with an exponent $\gamma$ determined by $\alpha$. We estimate these exponents from the first two rows of the ground state CM via the relation
\[\sqrt{\sum_{i,j=1}^{2}\Gamma_{i,2x+j}^2 }\sim x^{-\gamma}.\]

In Fig.~\ref{fig:algebraic_decay_1}, we show the fidelity of the reconstructed state for various values of $\alpha$ (and the corresponding $\gamma$) and system sizes, plotted as a function of the block size~$s$. The latter is rescaled by some power of system size~$n^\nu$, in order to obtain a data collapse for different system sizes. This collapse implies that, to reach a given target fidelity for system size $n$, the block size must scale as $n^\nu$. The exponent $\nu$ is estimated by choosing the value that yields the best data collapse.  The data show that for $\gamma\lesssim1$ our methods produce circuits with depth $\mathcal{O}(n)$ (as occurs at the critical point of the Ising model). By contrast, for $\gamma>1$ we get circuits whose depth scales with an exponent $\nu<1$, thereby achieving better scaling than the standard approach. This demonstrates that entanglement-cutting algorithms can be effective even for systems with algebraically decaying correlations.

\newcommand\subA{\mathrm{A}}
\newcommand\subB{\mathrm{B}}

\begin{figure}
    \centering
    \includegraphics[width=\linewidth]{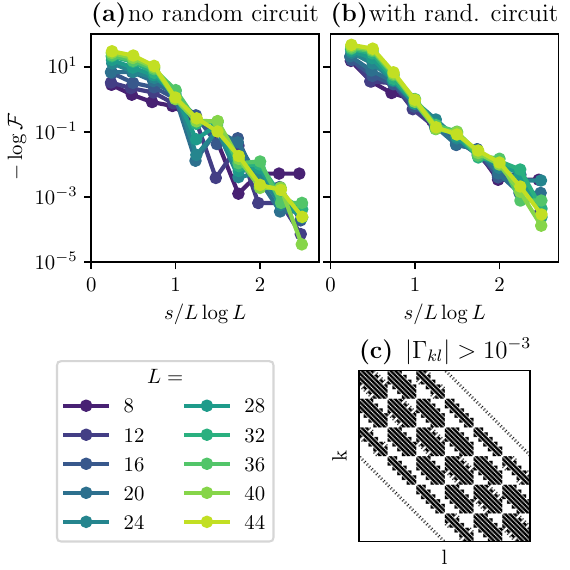}
    \caption{(a) Logarithm of the fidelity of the reconstructed ground state of $H_\text{2d}$ for $t=2$, $\mu=1$, and $\Delta=4$, prepared on $n=2L^2$ qubits, shown as a function of the block size $s$ rescaled by $L\log L$. Data are presented for various system sizes $L$. (b) Same as (a), but with a random MGC of depth $L\log L$ prepended. (c) Illustration of the tile structure of the ground-state CM $\Gamma$ (shown for $L=6$). Dark points indicate matrix elements $|\Gamma_{kl}|$ larger than a fixed cutoff.}
    \label{fig:2d_model}
\end{figure}

For the second example, we study a two-dimensional model of spin-$\frac12$ fermions defined on an $L\times L$ lattice, which arises in mean field approximations in the BCS theory of superconductivity~\cite{BCS57,HiMa89}. We label the two Majorana operators that correspond to the site at position $\vec x$ with spin $\sigma \in\{\uparrow, \downarrow\}$ by $\majo{\vec x, \sigma,\subA}$ and $\majo{\vec x, \sigma,\subB}$. The Hamiltonian reads 
\begin{align*}
H_\text{2d} = &- t \sum_{ \langle \vec x, \vec y \rangle, \sigma} \majo{\vec x,\sigma,\subA} \majo{\vec y,\sigma,\subB} 
 -  \mu \sum_{\vec x, \sigma} \majo{\vec x, \sigma, \subA} \majo{\vec x, \sigma, \subB} \\
&+ \Delta \sum_{\vec x} \Big(
 \majo{\vec x, \uparrow, \subA} \majo{\vec x, \downarrow, \subA} 
- \majo{\vec x, \uparrow, \subB} \majo{\vec x, \downarrow, \subB} \\
&\quad\quad\quad\quad- \majo{\vec x, \uparrow, \subA} \majo{\vec x, \downarrow, \subB}
+ \majo{\vec x, \downarrow, \subB} \majo{\vec x, \uparrow, \subA} \Big),
\end{align*}
where the first sum is carried out over all bonds in the lattice. For the parameters $t=2$, $\mu=1$, and $\Delta=4$, the ground state of $H_\text{2d}$ exhibits correlations that decay exponentially with the distance $|\vec x - \vec y|$ between lattice sites $\vec x$ and $\vec y$.  We now relabel lattice sites using the index $i(\vec x, \sigma) = 4 x_1 + 2 x_2 + \sigma$ (interpreting the spins $\uparrow$ and $\downarrow$ as $0$ and $1$ respectively). This mapping yields a one-dimensional chain of length $n=2L^2$, in which modes that are separated by a distance $\mathcal{O}(L)$ can still have non-vanishing correlations. Therefore, when applying our algorithm one would expects the circuit depth required to reproduce the state to scale at least with the linear system size $L$.

In Fig.~\ref{fig:2d_model}a, we again show the fidelity $\mathcal{F}$ of the reconstructed state as a function of block size $s$ and system size $L$. Here, the block size is rescaled by $L\log{L}$. One observes that the quantity $\log(-\log\mathcal{F}) \approx \log(1-\mathcal{F})$ decays approximately linearly with $s/L\log{L}$. However, one also finds some strong irregularities for several values of $s$ and $L$. This behavior can be explained to some extent due to a certain (approximate) tile structure (see Fig.~\ref{fig:2d_model}c) that the CMs of such ground states have. For some values of $s$ and $L$, the blocks determined by the cutting algorithm align more favorably with this structure, which gives a much lower fidelity.

To avoid this behavior, we apply the algorithm to states of the form $U_\text{r} \ket{\psi_\text{gs}}$, where $\ket{\psi_\text{gs}}$ is the ground state of $H_\text{2d}$, and $U_\text{r}$ is a random brickwall MGC of depth $L\log{L}$ (each gate is independently sampled from the matchgate Haar measure). This has the effect of smearing out the tile structure to some degree. In Fig.~\ref{fig:2d_model}b, we show the fidelity averaged over 10 realizations for each $L$ (we find that the fluctuations are negligible on the scale of the plot). Although the overall fidelity becomes slightly worse, the irregularities are suppressed at least when $s\gtrsim L\log{L}$. The ground states $\ket{\psi_\text{gs}}$ can then, of course, be approximated by applying $U_\text{r}^\dagger$ to the reconstructed states. To obtain a roughly equal fidelity for different $L$, the depth of the reconstructed circuit needs to scale with $L\log{L}$, improving the naïve approach which scales as $L^2$.

\bibliography{references}

\begin{thebibliography}{101}%
\makeatletter
\providecommand \@ifxundefined [1]{%
 \@ifx{#1\undefined}
}%
\providecommand \@ifnum [1]{%
 \ifnum #1\expandafter \@firstoftwo
 \else \expandafter \@secondoftwo
 \fi
}%
\providecommand \@ifx [1]{%
 \ifx #1\expandafter \@firstoftwo
 \else \expandafter \@secondoftwo
 \fi
}%
\providecommand \natexlab [1]{#1}%
\providecommand \enquote  [1]{``#1''}%
\providecommand \bibnamefont  [1]{#1}%
\providecommand \bibfnamefont [1]{#1}%
\providecommand \citenamefont [1]{#1}%
\providecommand \href@noop [0]{\@secondoftwo}%
\providecommand \href [0]{\begingroup \@sanitize@url \@href}%
\providecommand \@href[1]{\@@startlink{#1}\@@href}%
\providecommand \@@href[1]{\endgroup#1\@@endlink}%
\providecommand \@sanitize@url [0]{\catcode `\\12\catcode `\$12\catcode `\&12\catcode `\#12\catcode `\^12\catcode `\_12\catcode `\%12\relax}%
\providecommand \@@startlink[1]{}%
\providecommand \@@endlink[0]{}%
\providecommand \url  [0]{\begingroup\@sanitize@url \@url }%
\providecommand \@url [1]{\endgroup\@href {#1}{\urlprefix }}%
\providecommand \urlprefix  [0]{URL }%
\providecommand \Eprint [0]{\href }%
\providecommand \doibase [0]{https://doi.org/}%
\providecommand \selectlanguage [0]{\@gobble}%
\providecommand \bibinfo  [0]{\@secondoftwo}%
\providecommand \bibfield  [0]{\@secondoftwo}%
\providecommand \translation [1]{[#1]}%
\providecommand \BibitemOpen [0]{}%
\providecommand \bibitemStop [0]{}%
\providecommand \bibitemNoStop [0]{.\EOS\space}%
\providecommand \EOS [0]{\spacefactor3000\relax}%
\providecommand \BibitemShut  [1]{\csname bibitem#1\endcsname}%
\let\auto@bib@innerbib\@empty
\bibitem [{\citenamefont {McArdle}\ \emph {et~al.}(2020)\citenamefont {McArdle}, \citenamefont {Endo}, \citenamefont {Aspuru-Guzik}, \citenamefont {Benjamin},\ and\ \citenamefont {Yuan}}]{McEn20}%
  \BibitemOpen
  \bibfield  {author} {\bibinfo {author} {\bibfnamefont {S.}~\bibnamefont {McArdle}}, \bibinfo {author} {\bibfnamefont {S.}~\bibnamefont {Endo}}, \bibinfo {author} {\bibfnamefont {A.}~\bibnamefont {Aspuru-Guzik}}, \bibinfo {author} {\bibfnamefont {S.~C.}\ \bibnamefont {Benjamin}},\ and\ \bibinfo {author} {\bibfnamefont {X.}~\bibnamefont {Yuan}},\ }\bibfield  {title} {\bibinfo {title} {Quantum computational chemistry},\ }\bibfield  {journal} {\bibinfo  {journal} {Reviews of Modern Physics}\ }\textbf {\bibinfo {volume} {92}},\ \href {https://doi.org/10.1103/revmodphys.92.015003} {10.1103/revmodphys.92.015003} (\bibinfo {year} {2020})\BibitemShut {NoStop}%
\bibitem [{\citenamefont {Echenique}\ and\ \citenamefont {Alonso}(2007)}]{EcAl07}%
  \BibitemOpen
  \bibfield  {author} {\bibinfo {author} {\bibfnamefont {P.}~\bibnamefont {Echenique}}\ and\ \bibinfo {author} {\bibfnamefont {J.~L.}\ \bibnamefont {Alonso}},\ }\bibfield  {title} {\bibinfo {title} {A mathematical and computational review of hartree–fock scf methods in quantum chemistry},\ }\href {https://doi.org/10.1080/00268970701757875} {\bibfield  {journal} {\bibinfo  {journal} {Molecular Physics}\ }\textbf {\bibinfo {volume} {105}},\ \bibinfo {pages} {3057} (\bibinfo {year} {2007})},\ \Eprint {https://arxiv.org/abs/https://doi.org/10.1080/00268970701757875} {https://doi.org/10.1080/00268970701757875} \BibitemShut {NoStop}%
\bibitem [{\citenamefont {Bardeen}\ \emph {et~al.}(1957)\citenamefont {Bardeen}, \citenamefont {Cooper},\ and\ \citenamefont {Schrieffer}}]{BCS57}%
  \BibitemOpen
  \bibfield  {author} {\bibinfo {author} {\bibfnamefont {J.}~\bibnamefont {Bardeen}}, \bibinfo {author} {\bibfnamefont {L.~N.}\ \bibnamefont {Cooper}},\ and\ \bibinfo {author} {\bibfnamefont {J.~R.}\ \bibnamefont {Schrieffer}},\ }\bibfield  {title} {\bibinfo {title} {Theory of superconductivity},\ }\href {https://doi.org/10.1103/PhysRev.108.1175} {\bibfield  {journal} {\bibinfo  {journal} {Phys. Rev.}\ }\textbf {\bibinfo {volume} {108}},\ \bibinfo {pages} {1175} (\bibinfo {year} {1957})}\BibitemShut {NoStop}%
\bibitem [{\citenamefont {Valiant}(2001)}]{Va01}%
  \BibitemOpen
  \bibfield  {author} {\bibinfo {author} {\bibfnamefont {L.~G.}\ \bibnamefont {Valiant}},\ }\bibfield  {title} {\bibinfo {title} {Quantum computers that can be simulated classically in polynomial time},\ }in\ \href {https://doi.org/10.1145/380752.380785} {\emph {\bibinfo {booktitle} {Proceedings of the Thirty-Third Annual ACM Symposium on Theory of Computing}}},\ \bibinfo {series and number} {STOC '01}\ (\bibinfo  {publisher} {Association for Computing Machinery},\ \bibinfo {address} {New York, NY, USA},\ \bibinfo {year} {2001})\ p.\ \bibinfo {pages} {114–123}\BibitemShut {NoStop}%
\bibitem [{\citenamefont {Terhal}\ and\ \citenamefont {DiVincenzo}(2002)}]{TeDi02}%
  \BibitemOpen
  \bibfield  {author} {\bibinfo {author} {\bibfnamefont {B.~M.}\ \bibnamefont {Terhal}}\ and\ \bibinfo {author} {\bibfnamefont {D.~P.}\ \bibnamefont {DiVincenzo}},\ }\bibfield  {title} {\bibinfo {title} {Classical simulation of noninteracting-fermion quantum circuits},\ }\href {https://doi.org/10.1103/PhysRevA.65.032325} {\bibfield  {journal} {\bibinfo  {journal} {Phys. Rev. A}\ }\textbf {\bibinfo {volume} {65}},\ \bibinfo {pages} {032325} (\bibinfo {year} {2002})}\BibitemShut {NoStop}%
\bibitem [{\citenamefont {Knill}(2001)}]{Kn01}%
  \BibitemOpen
  \bibfield  {author} {\bibinfo {author} {\bibfnamefont {E.}~\bibnamefont {Knill}},\ }\href {https://arxiv.org/abs/quant-ph/0108033} {\bibinfo {title} {Fermionic linear optics and matchgates}} (\bibinfo {year} {2001}),\ \Eprint {https://arxiv.org/abs/quant-ph/0108033} {arXiv:quant-ph/0108033 [quant-ph]} \BibitemShut {NoStop}%
\bibitem [{\citenamefont {Jozsa}\ and\ \citenamefont {Miyake}(2008)}]{JoMi08}%
  \BibitemOpen
  \bibfield  {author} {\bibinfo {author} {\bibfnamefont {R.}~\bibnamefont {Jozsa}}\ and\ \bibinfo {author} {\bibfnamefont {A.}~\bibnamefont {Miyake}},\ }\bibfield  {title} {\bibinfo {title} {Matchgates and classical simulation of quantum circuits},\ }\href {https://doi.org/10.1098/rspa.2008.0189} {\bibfield  {journal} {\bibinfo  {journal} {Proceedings of the Royal Society A: Mathematical, Physical and Engineering Sciences}\ }\textbf {\bibinfo {volume} {464}},\ \bibinfo {pages} {3089–3106} (\bibinfo {year} {2008})}\BibitemShut {NoStop}%
\bibitem [{\citenamefont {Bravyi}(2005)}]{Br05Grassmann}%
  \BibitemOpen
  \bibfield  {author} {\bibinfo {author} {\bibfnamefont {S.}~\bibnamefont {Bravyi}},\ }\bibfield  {title} {\bibinfo {title} {Lagrangian representation for fermionic linear optics},\ }\href {https://doi.org/10.26421/QIC5.3-3} {\bibfield  {journal} {\bibinfo  {journal} {Quantum Information and Computation}\ }\textbf {\bibinfo {volume} {5}},\ \bibinfo {pages} {216 } (\bibinfo {year} {2005})}\BibitemShut {NoStop}%
\bibitem [{\citenamefont {Bravyi}\ and\ \citenamefont {König}(2012)}]{BrKo12}%
  \BibitemOpen
  \bibfield  {author} {\bibinfo {author} {\bibfnamefont {S.}~\bibnamefont {Bravyi}}\ and\ \bibinfo {author} {\bibfnamefont {R.}~\bibnamefont {König}},\ }\bibfield  {title} {\bibinfo {title} {Disorder-assisted error correction in {Majorana} chains},\ }\href {https://doi.org/10.1007/s00220-012-1606-9} {\bibfield  {journal} {\bibinfo  {journal} {Communications in Mathematical Physics}\ }\textbf {\bibinfo {volume} {316}},\ \bibinfo {pages} {641–692} (\bibinfo {year} {2012})}\BibitemShut {NoStop}%
\bibitem [{\citenamefont {Gluza}\ \emph {et~al.}(2018)\citenamefont {Gluza}, \citenamefont {Kliesch}, \citenamefont {Eisert},\ and\ \citenamefont {Aolita}}]{GlKl18}%
  \BibitemOpen
  \bibfield  {author} {\bibinfo {author} {\bibfnamefont {M.}~\bibnamefont {Gluza}}, \bibinfo {author} {\bibfnamefont {M.}~\bibnamefont {Kliesch}}, \bibinfo {author} {\bibfnamefont {J.}~\bibnamefont {Eisert}},\ and\ \bibinfo {author} {\bibfnamefont {L.}~\bibnamefont {Aolita}},\ }\bibfield  {title} {\bibinfo {title} {Fidelity witnesses for fermionic quantum simulations},\ }\href {https://doi.org/10.1103/PhysRevLett.120.190501} {\bibfield  {journal} {\bibinfo  {journal} {Phys. Rev. Lett.}\ }\textbf {\bibinfo {volume} {120}},\ \bibinfo {pages} {190501} (\bibinfo {year} {2018})}\BibitemShut {NoStop}%
\bibitem [{\citenamefont {O'Gorman}(2022)}]{Br22}%
  \BibitemOpen
  \bibfield  {author} {\bibinfo {author} {\bibfnamefont {B.}~\bibnamefont {O'Gorman}},\ }\href {https://arxiv.org/abs/2207.14787} {\bibinfo {title} {Fermionic tomography and learning}} (\bibinfo {year} {2022}),\ \Eprint {https://arxiv.org/abs/2207.14787} {arXiv:2207.14787 [quant-ph]} \BibitemShut {NoStop}%
\bibitem [{\citenamefont {Aaronson}\ and\ \citenamefont {Grewal}(2023)}]{AaGe23}%
  \BibitemOpen
  \bibfield  {author} {\bibinfo {author} {\bibfnamefont {S.}~\bibnamefont {Aaronson}}\ and\ \bibinfo {author} {\bibfnamefont {S.}~\bibnamefont {Grewal}},\ }\href {https://arxiv.org/abs/2102.10458} {\bibinfo {title} {Efficient tomography of non-interacting fermion states}} (\bibinfo {year} {2023}),\ \Eprint {https://arxiv.org/abs/2102.10458} {arXiv:2102.10458 [quant-ph]} \BibitemShut {NoStop}%
\bibitem [{\citenamefont {Mele}\ and\ \citenamefont {Herasymenko}(2025)}]{MeHe25}%
  \BibitemOpen
  \bibfield  {author} {\bibinfo {author} {\bibfnamefont {A.~A.}\ \bibnamefont {Mele}}\ and\ \bibinfo {author} {\bibfnamefont {Y.}~\bibnamefont {Herasymenko}},\ }\bibfield  {title} {\bibinfo {title} {Efficient learning of quantum states prepared with few fermionic non-gaussian gates},\ }\href {https://doi.org/10.1103/PRXQuantum.6.010319} {\bibfield  {journal} {\bibinfo  {journal} {PRX Quantum}\ }\textbf {\bibinfo {volume} {6}},\ \bibinfo {pages} {010319} (\bibinfo {year} {2025})}\BibitemShut {NoStop}%
\bibitem [{\citenamefont {Bittel}\ \emph {et~al.}(2025)\citenamefont {Bittel}, \citenamefont {Mele}, \citenamefont {Eisert},\ and\ \citenamefont {Leone}}]{LeMe25}%
  \BibitemOpen
  \bibfield  {author} {\bibinfo {author} {\bibfnamefont {L.}~\bibnamefont {Bittel}}, \bibinfo {author} {\bibfnamefont {A.~A.}\ \bibnamefont {Mele}}, \bibinfo {author} {\bibfnamefont {J.}~\bibnamefont {Eisert}},\ and\ \bibinfo {author} {\bibfnamefont {L.}~\bibnamefont {Leone}},\ }\bibfield  {title} {\bibinfo {title} {Optimal trace-distance bounds for free-fermionic states: Testing and improved tomography},\ }\href {https://doi.org/10.1103/pzx6-nkfb} {\bibfield  {journal} {\bibinfo  {journal} {PRX Quantum}\ }\textbf {\bibinfo {volume} {6}},\ \bibinfo {pages} {030341} (\bibinfo {year} {2025})}\BibitemShut {NoStop}%
\bibitem [{\citenamefont {Carrasco}\ \emph {et~al.}(2024)\citenamefont {Carrasco}, \citenamefont {Langer}, \citenamefont {Neven},\ and\ \citenamefont {Kraus}}]{CaLa25}%
  \BibitemOpen
  \bibfield  {author} {\bibinfo {author} {\bibfnamefont {J.}~\bibnamefont {Carrasco}}, \bibinfo {author} {\bibfnamefont {M.}~\bibnamefont {Langer}}, \bibinfo {author} {\bibfnamefont {A.}~\bibnamefont {Neven}},\ and\ \bibinfo {author} {\bibfnamefont {B.}~\bibnamefont {Kraus}},\ }\bibfield  {title} {\bibinfo {title} {Gaining confidence on the correct realization of arbitrary quantum computations},\ }\bibfield  {journal} {\bibinfo  {journal} {Physical Review Research}\ }\textbf {\bibinfo {volume} {6}},\ \href {https://doi.org/10.1103/physrevresearch.6.l032074} {10.1103/physrevresearch.6.l032074} (\bibinfo {year} {2024})\BibitemShut {NoStop}%
\bibitem [{\citenamefont {Brod}\ and\ \citenamefont {Galv\~ao}(2011)}]{BrGa11}%
  \BibitemOpen
  \bibfield  {author} {\bibinfo {author} {\bibfnamefont {D.~J.}\ \bibnamefont {Brod}}\ and\ \bibinfo {author} {\bibfnamefont {E.~F.}\ \bibnamefont {Galv\~ao}},\ }\bibfield  {title} {\bibinfo {title} {Extending matchgates into universal quantum computation},\ }\href {https://doi.org/10.1103/PhysRevA.84.022310} {\bibfield  {journal} {\bibinfo  {journal} {Phys. Rev. A}\ }\textbf {\bibinfo {volume} {84}},\ \bibinfo {pages} {022310} (\bibinfo {year} {2011})}\BibitemShut {NoStop}%
\bibitem [{\citenamefont {Bravyi}(2006)}]{bravyi}%
  \BibitemOpen
  \bibfield  {author} {\bibinfo {author} {\bibfnamefont {S.}~\bibnamefont {Bravyi}},\ }\bibfield  {title} {\bibinfo {title} {Universal quantum computation with the $\ensuremath{\nu}=5/2$ fractional quantum hall state},\ }\href {https://doi.org/10.1103/PhysRevA.73.042313} {\bibfield  {journal} {\bibinfo  {journal} {Phys. Rev. A}\ }\textbf {\bibinfo {volume} {73}},\ \bibinfo {pages} {042313} (\bibinfo {year} {2006})}\BibitemShut {NoStop}%
\bibitem [{\citenamefont {Hebenstreit}\ \emph {et~al.}(2019)\citenamefont {Hebenstreit}, \citenamefont {Jozsa}, \citenamefont {Kraus}, \citenamefont {Strelchuk},\ and\ \citenamefont {Yoganathan}}]{HeJo19}%
  \BibitemOpen
  \bibfield  {author} {\bibinfo {author} {\bibfnamefont {M.}~\bibnamefont {Hebenstreit}}, \bibinfo {author} {\bibfnamefont {R.}~\bibnamefont {Jozsa}}, \bibinfo {author} {\bibfnamefont {B.}~\bibnamefont {Kraus}}, \bibinfo {author} {\bibfnamefont {S.}~\bibnamefont {Strelchuk}},\ and\ \bibinfo {author} {\bibfnamefont {M.}~\bibnamefont {Yoganathan}},\ }\bibfield  {title} {\bibinfo {title} {All pure fermionic non-gaussian states are magic states for matchgate computations},\ }\href {https://doi.org/10.1103/PhysRevLett.123.080503} {\bibfield  {journal} {\bibinfo  {journal} {Phys. Rev. Lett.}\ }\textbf {\bibinfo {volume} {123}},\ \bibinfo {pages} {080503} (\bibinfo {year} {2019})}\BibitemShut {NoStop}%
\bibitem [{\citenamefont {Wick}(1950)}]{Wi50}%
  \BibitemOpen
  \bibfield  {author} {\bibinfo {author} {\bibfnamefont {G.~C.}\ \bibnamefont {Wick}},\ }\bibfield  {title} {\bibinfo {title} {The evaluation of the collision matrix},\ }\href {https://doi.org/10.1103/PhysRev.80.268} {\bibfield  {journal} {\bibinfo  {journal} {Phys. Rev.}\ }\textbf {\bibinfo {volume} {80}},\ \bibinfo {pages} {268} (\bibinfo {year} {1950})}\BibitemShut {NoStop}%
\bibitem [{\citenamefont {Bach}\ \emph {et~al.}(1994)\citenamefont {Bach}, \citenamefont {Lieb},\ and\ \citenamefont {Solovej}}]{BaLi94}%
  \BibitemOpen
  \bibfield  {author} {\bibinfo {author} {\bibfnamefont {V.}~\bibnamefont {Bach}}, \bibinfo {author} {\bibfnamefont {E.~H.}\ \bibnamefont {Lieb}},\ and\ \bibinfo {author} {\bibfnamefont {J.~P.}\ \bibnamefont {Solovej}},\ }\bibfield  {title} {\bibinfo {title} {Generalized hartree-fock theory and the hubbard model},\ }\href {https://doi.org/10.1007/bf02188656} {\bibfield  {journal} {\bibinfo  {journal} {Journal of Statistical Physics}\ }\textbf {\bibinfo {volume} {76}},\ \bibinfo {pages} {3–89} (\bibinfo {year} {1994})}\BibitemShut {NoStop}%
\bibitem [{\citenamefont {Peschel}\ and\ \citenamefont {Chung}(1999)}]{PeCh99}%
  \BibitemOpen
  \bibfield  {author} {\bibinfo {author} {\bibfnamefont {I.}~\bibnamefont {Peschel}}\ and\ \bibinfo {author} {\bibfnamefont {M.-C.}\ \bibnamefont {Chung}},\ }\bibfield  {title} {\bibinfo {title} {Density matrices for a chain of oscillators},\ }\href {https://doi.org/10.1088/0305-4470/32/48/305} {\bibfield  {journal} {\bibinfo  {journal} {Journal of Physics A: Mathematical and General}\ }\textbf {\bibinfo {volume} {32}},\ \bibinfo {pages} {8419–8428} (\bibinfo {year} {1999})}\BibitemShut {NoStop}%
\bibitem [{\citenamefont {Peschel}(2003)}]{Pe03}%
  \BibitemOpen
  \bibfield  {author} {\bibinfo {author} {\bibfnamefont {I.}~\bibnamefont {Peschel}},\ }\bibfield  {title} {\bibinfo {title} {Calculation of reduced density matrices from correlation functions},\ }\href {https://doi.org/10.1088/0305-4470/36/14/101} {\bibfield  {journal} {\bibinfo  {journal} {Journal of Physics A: Mathematical and General}\ }\textbf {\bibinfo {volume} {36}},\ \bibinfo {pages} {L205–L208} (\bibinfo {year} {2003})}\BibitemShut {NoStop}%
\bibitem [{\citenamefont {Botero}\ and\ \citenamefont {Reznik}(2004)}]{BoRe04}%
  \BibitemOpen
  \bibfield  {author} {\bibinfo {author} {\bibfnamefont {A.}~\bibnamefont {Botero}}\ and\ \bibinfo {author} {\bibfnamefont {B.}~\bibnamefont {Reznik}},\ }\bibfield  {title} {\bibinfo {title} {Bcs-like modewise entanglement of fermion gaussian states},\ }\href {https://doi.org/https://doi.org/10.1016/j.physleta.2004.08.037} {\bibfield  {journal} {\bibinfo  {journal} {Physics Letters A}\ }\textbf {\bibinfo {volume} {331}},\ \bibinfo {pages} {39} (\bibinfo {year} {2004})}\BibitemShut {NoStop}%
\bibitem [{\citenamefont {L\"owdin}(1955)}]{Loe55}%
  \BibitemOpen
  \bibfield  {author} {\bibinfo {author} {\bibfnamefont {P.-O.}\ \bibnamefont {L\"owdin}},\ }\bibfield  {title} {\bibinfo {title} {Quantum theory of many-particle systems. i. physical interpretations by means of density matrices, natural spin-orbitals, and convergence problems in the method of configurational interaction},\ }\href {https://doi.org/10.1103/PhysRev.97.1474} {\bibfield  {journal} {\bibinfo  {journal} {Phys. Rev.}\ }\textbf {\bibinfo {volume} {97}},\ \bibinfo {pages} {1474} (\bibinfo {year} {1955})}\BibitemShut {NoStop}%
\bibitem [{\citenamefont {Bravyi}\ and\ \citenamefont {Gosset}(2017)}]{BrGo17}%
  \BibitemOpen
  \bibfield  {author} {\bibinfo {author} {\bibfnamefont {S.}~\bibnamefont {Bravyi}}\ and\ \bibinfo {author} {\bibfnamefont {D.}~\bibnamefont {Gosset}},\ }\bibfield  {title} {\bibinfo {title} {Complexity of quantum impurity problems},\ }\href {https://doi.org/10.1007/s00220-017-2976-9} {\bibfield  {journal} {\bibinfo  {journal} {Communications in Mathematical Physics}\ }\textbf {\bibinfo {volume} {356}},\ \bibinfo {pages} {451–500} (\bibinfo {year} {2017})}\BibitemShut {NoStop}%
\bibitem [{\citenamefont {Hebenstreit}\ \emph {et~al.}(2020)\citenamefont {Hebenstreit}, \citenamefont {Jozsa}, \citenamefont {Kraus},\ and\ \citenamefont {Strelchuk}}]{HeJo20}%
  \BibitemOpen
  \bibfield  {author} {\bibinfo {author} {\bibfnamefont {M.}~\bibnamefont {Hebenstreit}}, \bibinfo {author} {\bibfnamefont {R.}~\bibnamefont {Jozsa}}, \bibinfo {author} {\bibfnamefont {B.}~\bibnamefont {Kraus}},\ and\ \bibinfo {author} {\bibfnamefont {S.}~\bibnamefont {Strelchuk}},\ }\bibfield  {title} {\bibinfo {title} {Computational power of matchgates with supplementary resources},\ }\href {https://doi.org/10.1103/PhysRevA.102.052604} {\bibfield  {journal} {\bibinfo  {journal} {Phys. Rev. A}\ }\textbf {\bibinfo {volume} {102}},\ \bibinfo {pages} {052604} (\bibinfo {year} {2020})}\BibitemShut {NoStop}%
\bibitem [{\citenamefont {Hakkaku}\ \emph {et~al.}(2022)\citenamefont {Hakkaku}, \citenamefont {Tashima}, \citenamefont {Mitarai}, \citenamefont {Mizukami},\ and\ \citenamefont {Fujii}}]{HaTa22}%
  \BibitemOpen
  \bibfield  {author} {\bibinfo {author} {\bibfnamefont {S.}~\bibnamefont {Hakkaku}}, \bibinfo {author} {\bibfnamefont {Y.}~\bibnamefont {Tashima}}, \bibinfo {author} {\bibfnamefont {K.}~\bibnamefont {Mitarai}}, \bibinfo {author} {\bibfnamefont {W.}~\bibnamefont {Mizukami}},\ and\ \bibinfo {author} {\bibfnamefont {K.}~\bibnamefont {Fujii}},\ }\bibfield  {title} {\bibinfo {title} {Quantifying fermionic nonlinearity of quantum circuits},\ }\href {https://doi.org/10.1103/PhysRevResearch.4.043100} {\bibfield  {journal} {\bibinfo  {journal} {Phys. Rev. Res.}\ }\textbf {\bibinfo {volume} {4}},\ \bibinfo {pages} {043100} (\bibinfo {year} {2022})}\BibitemShut {NoStop}%
\bibitem [{\citenamefont {Mocherla}\ \emph {et~al.}(2024)\citenamefont {Mocherla}, \citenamefont {Lao},\ and\ \citenamefont {Browne}}]{MoLa24}%
  \BibitemOpen
  \bibfield  {author} {\bibinfo {author} {\bibfnamefont {A.}~\bibnamefont {Mocherla}}, \bibinfo {author} {\bibfnamefont {L.}~\bibnamefont {Lao}},\ and\ \bibinfo {author} {\bibfnamefont {D.~E.}\ \bibnamefont {Browne}},\ }\href {https://arxiv.org/abs/2302.02654} {\bibinfo {title} {Extending matchgate simulation methods to universal quantum circuits}} (\bibinfo {year} {2024}),\ \Eprint {https://arxiv.org/abs/2302.02654} {arXiv:2302.02654 [quant-ph]} \BibitemShut {NoStop}%
\bibitem [{\citenamefont {Miller}\ \emph {et~al.}(2025)\citenamefont {Miller}, \citenamefont {Favre}, \citenamefont {Holmes}, \citenamefont {Salehi}, \citenamefont {Chakraborty}, \citenamefont {Nyk{\"a}nen}, \citenamefont {Zimbor{\'a}s}, \citenamefont {Glos},\ and\ \citenamefont {Garc{\'i}a-P{\'e}rez}}]{MiFa25}%
  \BibitemOpen
  \bibfield  {author} {\bibinfo {author} {\bibfnamefont {A.}~\bibnamefont {Miller}}, \bibinfo {author} {\bibfnamefont {J.}~\bibnamefont {Favre}}, \bibinfo {author} {\bibfnamefont {Z.}~\bibnamefont {Holmes}}, \bibinfo {author} {\bibfnamefont {{\"O}.}~\bibnamefont {Salehi}}, \bibinfo {author} {\bibfnamefont {R.}~\bibnamefont {Chakraborty}}, \bibinfo {author} {\bibfnamefont {A.}~\bibnamefont {Nyk{\"a}nen}}, \bibinfo {author} {\bibfnamefont {Z.}~\bibnamefont {Zimbor{\'a}s}}, \bibinfo {author} {\bibfnamefont {A.}~\bibnamefont {Glos}},\ and\ \bibinfo {author} {\bibfnamefont {G.}~\bibnamefont {Garc{\'i}a-P{\'e}rez}},\ }\href {https://arxiv.org/abs/2503.18939} {\bibinfo {title} {Simulation of fermionic circuits using majorana propagation}} (\bibinfo {year} {2025}),\ \Eprint {https://arxiv.org/abs/2503.18939} {arXiv:2503.18939 [quant-ph]} \BibitemShut {NoStop}%
\bibitem [{\citenamefont {Dias}\ and\ \citenamefont {Koenig}(2024)}]{DiKo24}%
  \BibitemOpen
  \bibfield  {author} {\bibinfo {author} {\bibfnamefont {B.}~\bibnamefont {Dias}}\ and\ \bibinfo {author} {\bibfnamefont {R.}~\bibnamefont {Koenig}},\ }\bibfield  {title} {\bibinfo {title} {Classical simulation of non-{G}aussian fermionic circuits},\ }\href {https://doi.org/10.22331/q-2024-05-21-1350} {\bibfield  {journal} {\bibinfo  {journal} {{Quantum}}\ }\textbf {\bibinfo {volume} {8}},\ \bibinfo {pages} {1350} (\bibinfo {year} {2024})}\BibitemShut {NoStop}%
\bibitem [{\citenamefont {Reardon-Smith}\ \emph {et~al.}(2024)\citenamefont {Reardon-Smith}, \citenamefont {Oszmaniec},\ and\ \citenamefont {Korzekwa}}]{ReOs24}%
  \BibitemOpen
  \bibfield  {author} {\bibinfo {author} {\bibfnamefont {O.}~\bibnamefont {Reardon-Smith}}, \bibinfo {author} {\bibfnamefont {M.}~\bibnamefont {Oszmaniec}},\ and\ \bibinfo {author} {\bibfnamefont {K.}~\bibnamefont {Korzekwa}},\ }\bibfield  {title} {\bibinfo {title} {Improved simulation of quantum circuits dominated by free fermionic operations},\ }\href {https://doi.org/10.22331/q-2024-12-04-1549} {\bibfield  {journal} {\bibinfo  {journal} {{Quantum}}\ }\textbf {\bibinfo {volume} {8}},\ \bibinfo {pages} {1549} (\bibinfo {year} {2024})}\BibitemShut {NoStop}%
\bibitem [{\citenamefont {Guaita}\ \emph {et~al.}(2024)\citenamefont {Guaita}, \citenamefont {Hackl},\ and\ \citenamefont {Quella}}]{GuHa24}%
  \BibitemOpen
  \bibfield  {author} {\bibinfo {author} {\bibfnamefont {T.}~\bibnamefont {Guaita}}, \bibinfo {author} {\bibfnamefont {L.}~\bibnamefont {Hackl}},\ and\ \bibinfo {author} {\bibfnamefont {T.}~\bibnamefont {Quella}},\ }\href {https://arxiv.org/abs/2409.11628} {\bibinfo {title} {Representation theory of gaussian unitary transformations for bosonic and fermionic systems}} (\bibinfo {year} {2024}),\ \Eprint {https://arxiv.org/abs/2409.11628} {arXiv:2409.11628 [quant-ph]} \BibitemShut {NoStop}%
\bibitem [{\citenamefont {Jiang}\ \emph {et~al.}(2018)\citenamefont {Jiang}, \citenamefont {Sung}, \citenamefont {Kechedzhi}, \citenamefont {Smelyanskiy},\ and\ \citenamefont {Boixo}}]{JiSu18}%
  \BibitemOpen
  \bibfield  {author} {\bibinfo {author} {\bibfnamefont {Z.}~\bibnamefont {Jiang}}, \bibinfo {author} {\bibfnamefont {K.~J.}\ \bibnamefont {Sung}}, \bibinfo {author} {\bibfnamefont {K.}~\bibnamefont {Kechedzhi}}, \bibinfo {author} {\bibfnamefont {V.~N.}\ \bibnamefont {Smelyanskiy}},\ and\ \bibinfo {author} {\bibfnamefont {S.}~\bibnamefont {Boixo}},\ }\bibfield  {title} {\bibinfo {title} {Quantum algorithms to simulate many-body physics of correlated fermions},\ }\href {https://doi.org/10.1103/PhysRevApplied.9.044036} {\bibfield  {journal} {\bibinfo  {journal} {Phys. Rev. Appl.}\ }\textbf {\bibinfo {volume} {9}},\ \bibinfo {pages} {044036} (\bibinfo {year} {2018})}\BibitemShut {NoStop}%
\bibitem [{\citenamefont {Kivlichan}\ \emph {et~al.}(2018)\citenamefont {Kivlichan}, \citenamefont {McClean}, \citenamefont {Wiebe}, \citenamefont {Gidney}, \citenamefont {Aspuru-Guzik}, \citenamefont {Chan},\ and\ \citenamefont {Babbush}}]{KiMc18}%
  \BibitemOpen
  \bibfield  {author} {\bibinfo {author} {\bibfnamefont {I.~D.}\ \bibnamefont {Kivlichan}}, \bibinfo {author} {\bibfnamefont {J.}~\bibnamefont {McClean}}, \bibinfo {author} {\bibfnamefont {N.}~\bibnamefont {Wiebe}}, \bibinfo {author} {\bibfnamefont {C.}~\bibnamefont {Gidney}}, \bibinfo {author} {\bibfnamefont {A.}~\bibnamefont {Aspuru-Guzik}}, \bibinfo {author} {\bibfnamefont {G.~K.-L.}\ \bibnamefont {Chan}},\ and\ \bibinfo {author} {\bibfnamefont {R.}~\bibnamefont {Babbush}},\ }\bibfield  {title} {\bibinfo {title} {Quantum simulation of electronic structure with linear depth and connectivity},\ }\href {https://doi.org/10.1103/PhysRevLett.120.110501} {\bibfield  {journal} {\bibinfo  {journal} {Phys. Rev. Lett.}\ }\textbf {\bibinfo {volume} {120}},\ \bibinfo {pages} {110501} (\bibinfo {year} {2018})}\BibitemShut {NoStop}%
\bibitem [{\citenamefont {Dallaire-Demers}\ \emph {et~al.}(2019)\citenamefont {Dallaire-Demers}, \citenamefont {Romero}, \citenamefont {Veis}, \citenamefont {Sim},\ and\ \citenamefont {Aspuru-Guzik}}]{DaDe19}%
  \BibitemOpen
  \bibfield  {author} {\bibinfo {author} {\bibfnamefont {P.-L.}\ \bibnamefont {Dallaire-Demers}}, \bibinfo {author} {\bibfnamefont {J.}~\bibnamefont {Romero}}, \bibinfo {author} {\bibfnamefont {L.}~\bibnamefont {Veis}}, \bibinfo {author} {\bibfnamefont {S.}~\bibnamefont {Sim}},\ and\ \bibinfo {author} {\bibfnamefont {A.}~\bibnamefont {Aspuru-Guzik}},\ }\bibfield  {title} {\bibinfo {title} {Low-depth circuit ansatz for preparing correlated fermionic states on a quantum computer},\ }\href {https://doi.org/10.1088/2058-9565/ab3951} {\bibfield  {journal} {\bibinfo  {journal} {Quantum Science and Technology}\ }\textbf {\bibinfo {volume} {4}},\ \bibinfo {pages} {045005} (\bibinfo {year} {2019})}\BibitemShut {NoStop}%
\bibitem [{\citenamefont {Oszmaniec}\ \emph {et~al.}(2022)\citenamefont {Oszmaniec}, \citenamefont {Dangniam}, \citenamefont {Morales},\ and\ \citenamefont {Zimbor\'as}}]{OsDa22}%
  \BibitemOpen
  \bibfield  {author} {\bibinfo {author} {\bibfnamefont {M.}~\bibnamefont {Oszmaniec}}, \bibinfo {author} {\bibfnamefont {N.}~\bibnamefont {Dangniam}}, \bibinfo {author} {\bibfnamefont {M.~E.}\ \bibnamefont {Morales}},\ and\ \bibinfo {author} {\bibfnamefont {Z.}~\bibnamefont {Zimbor\'as}},\ }\bibfield  {title} {\bibinfo {title} {Fermion sampling: A robust quantum computational advantage scheme using fermionic linear optics and magic input states},\ }\href {https://doi.org/10.1103/PRXQuantum.3.020328} {\bibfield  {journal} {\bibinfo  {journal} {PRX Quantum}\ }\textbf {\bibinfo {volume} {3}},\ \bibinfo {pages} {020328} (\bibinfo {year} {2022})}\BibitemShut {NoStop}%
\bibitem [{\citenamefont {Verstraete}\ \emph {et~al.}(2009)\citenamefont {Verstraete}, \citenamefont {Cirac},\ and\ \citenamefont {Latorre}}]{VeCi09}%
  \BibitemOpen
  \bibfield  {author} {\bibinfo {author} {\bibfnamefont {F.}~\bibnamefont {Verstraete}}, \bibinfo {author} {\bibfnamefont {J.~I.}\ \bibnamefont {Cirac}},\ and\ \bibinfo {author} {\bibfnamefont {J.~I.}\ \bibnamefont {Latorre}},\ }\bibfield  {title} {\bibinfo {title} {Quantum circuits for strongly correlated quantum systems},\ }\href {https://doi.org/10.1103/PhysRevA.79.032316} {\bibfield  {journal} {\bibinfo  {journal} {Phys. Rev. A}\ }\textbf {\bibinfo {volume} {79}},\ \bibinfo {pages} {032316} (\bibinfo {year} {2009})}\BibitemShut {NoStop}%
\bibitem [{\citenamefont {Fishman}\ and\ \citenamefont {White}(2015)}]{FiWh15}%
  \BibitemOpen
  \bibfield  {author} {\bibinfo {author} {\bibfnamefont {M.~T.}\ \bibnamefont {Fishman}}\ and\ \bibinfo {author} {\bibfnamefont {S.~R.}\ \bibnamefont {White}},\ }\bibfield  {title} {\bibinfo {title} {Compression of correlation matrices and an efficient method for forming matrix product states of fermionic gaussian states},\ }\href {https://doi.org/10.1103/PhysRevB.92.075132} {\bibfield  {journal} {\bibinfo  {journal} {Phys. Rev. B}\ }\textbf {\bibinfo {volume} {92}},\ \bibinfo {pages} {075132} (\bibinfo {year} {2015})}\BibitemShut {NoStop}%
\bibitem [{\citenamefont {Thoenniss}\ \emph {et~al.}(2023)\citenamefont {Thoenniss}, \citenamefont {Lerose},\ and\ \citenamefont {Abanin}}]{ThLe23}%
  \BibitemOpen
  \bibfield  {author} {\bibinfo {author} {\bibfnamefont {J.}~\bibnamefont {Thoenniss}}, \bibinfo {author} {\bibfnamefont {A.}~\bibnamefont {Lerose}},\ and\ \bibinfo {author} {\bibfnamefont {D.~A.}\ \bibnamefont {Abanin}},\ }\bibfield  {title} {\bibinfo {title} {Nonequilibrium quantum impurity problems via matrix-product states in the temporal domain},\ }\href {https://doi.org/10.1103/PhysRevB.107.195101} {\bibfield  {journal} {\bibinfo  {journal} {Phys. Rev. B}\ }\textbf {\bibinfo {volume} {107}},\ \bibinfo {pages} {195101} (\bibinfo {year} {2023})}\BibitemShut {NoStop}%
\bibitem [{\citenamefont {Wu}\ \emph {et~al.}(2025)\citenamefont {Wu}, \citenamefont {Kloss}, \citenamefont {Krinitsin}, \citenamefont {Fishman}, \citenamefont {Pixley},\ and\ \citenamefont {Stoudenmire}}]{WuKl25}%
  \BibitemOpen
  \bibfield  {author} {\bibinfo {author} {\bibfnamefont {A.-K.}\ \bibnamefont {Wu}}, \bibinfo {author} {\bibfnamefont {B.}~\bibnamefont {Kloss}}, \bibinfo {author} {\bibfnamefont {W.}~\bibnamefont {Krinitsin}}, \bibinfo {author} {\bibfnamefont {M.~T.}\ \bibnamefont {Fishman}}, \bibinfo {author} {\bibfnamefont {J.~H.}\ \bibnamefont {Pixley}},\ and\ \bibinfo {author} {\bibfnamefont {E.~M.}\ \bibnamefont {Stoudenmire}},\ }\bibfield  {title} {\bibinfo {title} {Disentangling interacting systems with fermionic gaussian circuits: Application to quantum impurity models},\ }\href {https://doi.org/10.1103/PhysRevB.111.035119} {\bibfield  {journal} {\bibinfo  {journal} {Phys. Rev. B}\ }\textbf {\bibinfo {volume} {111}},\ \bibinfo {pages} {035119} (\bibinfo {year} {2025})}\BibitemShut {NoStop}%
\bibitem [{\citenamefont {Wong}\ and\ \citenamefont {Potter}(2025)}]{WoPo25}%
  \BibitemOpen
  \bibfield  {author} {\bibinfo {author} {\bibfnamefont {S.~L.}\ \bibnamefont {Wong}}\ and\ \bibinfo {author} {\bibfnamefont {A.~C.}\ \bibnamefont {Potter}},\ }\href {https://arxiv.org/abs/2506.04200} {\bibinfo {title} {Entanglement renormalization circuits for $2d$ gaussian fermion states}} (\bibinfo {year} {2025}),\ \Eprint {https://arxiv.org/abs/2506.04200} {arXiv:2506.04200 [quant-ph]} \BibitemShut {NoStop}%
\bibitem [{\citenamefont {Malz}\ \emph {et~al.}(2024)\citenamefont {Malz}, \citenamefont {Styliaris}, \citenamefont {Wei},\ and\ \citenamefont {Cirac}}]{MaSt24}%
  \BibitemOpen
  \bibfield  {author} {\bibinfo {author} {\bibfnamefont {D.}~\bibnamefont {Malz}}, \bibinfo {author} {\bibfnamefont {G.}~\bibnamefont {Styliaris}}, \bibinfo {author} {\bibfnamefont {Z.-Y.}\ \bibnamefont {Wei}},\ and\ \bibinfo {author} {\bibfnamefont {J.~I.}\ \bibnamefont {Cirac}},\ }\bibfield  {title} {\bibinfo {title} {Preparation of matrix product states with log-depth quantum circuits},\ }\href {https://doi.org/10.1103/PhysRevLett.132.040404} {\bibfield  {journal} {\bibinfo  {journal} {Phys. Rev. Lett.}\ }\textbf {\bibinfo {volume} {132}},\ \bibinfo {pages} {040404} (\bibinfo {year} {2024})}\BibitemShut {NoStop}%
\bibitem [{\citenamefont {Morral-Yepes}\ \emph {et~al.}(2025)\citenamefont {Morral-Yepes}, \citenamefont {Langer}, \citenamefont {Gammon-Smith}, \citenamefont {Kraus},\ and\ \citenamefont {Pollmann}}]{MoLa25}%
  \BibitemOpen
  \bibfield  {author} {\bibinfo {author} {\bibfnamefont {R.}~\bibnamefont {Morral-Yepes}}, \bibinfo {author} {\bibfnamefont {M.}~\bibnamefont {Langer}}, \bibinfo {author} {\bibfnamefont {A.}~\bibnamefont {Gammon-Smith}}, \bibinfo {author} {\bibfnamefont {B.}~\bibnamefont {Kraus}},\ and\ \bibinfo {author} {\bibfnamefont {F.}~\bibnamefont {Pollmann}},\ }\href {https://arxiv.org/abs/2507.05055} {\bibinfo {title} {Disentangling strategies and entanglement transitions in unitary circuit games with matchgates}} (\bibinfo {year} {2025}),\ \Eprint {https://arxiv.org/abs/2507.05055} {arXiv:2507.05055 [quant-ph]} \BibitemShut {NoStop}%
\bibitem [{\citenamefont {Morral-Yepes}\ \emph {et~al.}(2024)\citenamefont {Morral-Yepes}, \citenamefont {Smith}, \citenamefont {Sondhi},\ and\ \citenamefont {Pollmann}}]{MoSm24}%
  \BibitemOpen
  \bibfield  {author} {\bibinfo {author} {\bibfnamefont {R.}~\bibnamefont {Morral-Yepes}}, \bibinfo {author} {\bibfnamefont {A.}~\bibnamefont {Smith}}, \bibinfo {author} {\bibfnamefont {S.}~\bibnamefont {Sondhi}},\ and\ \bibinfo {author} {\bibfnamefont {F.}~\bibnamefont {Pollmann}},\ }\bibfield  {title} {\bibinfo {title} {Entanglement transitions in unitary circuit games},\ }\bibfield  {journal} {\bibinfo  {journal} {PRX Quantum}\ }\textbf {\bibinfo {volume} {5}},\ \href {https://doi.org/10.1103/prxquantum.5.010309} {10.1103/prxquantum.5.010309} (\bibinfo {year} {2024})\BibitemShut {NoStop}%
\bibitem [{\citenamefont {Camps}\ \emph {et~al.}(2022)\citenamefont {Camps}, \citenamefont {Kökcü}, \citenamefont {Bassman~Oftelie}, \citenamefont {de~Jong}, \citenamefont {Kemper},\ and\ \citenamefont {Van~Beeumen}}]{CaKo22}%
  \BibitemOpen
  \bibfield  {author} {\bibinfo {author} {\bibfnamefont {D.}~\bibnamefont {Camps}}, \bibinfo {author} {\bibfnamefont {E.}~\bibnamefont {Kökcü}}, \bibinfo {author} {\bibfnamefont {L.}~\bibnamefont {Bassman~Oftelie}}, \bibinfo {author} {\bibfnamefont {W.~A.}\ \bibnamefont {de~Jong}}, \bibinfo {author} {\bibfnamefont {A.~F.}\ \bibnamefont {Kemper}},\ and\ \bibinfo {author} {\bibfnamefont {R.}~\bibnamefont {Van~Beeumen}},\ }\bibfield  {title} {\bibinfo {title} {An algebraic quantum circuit compression algorithm for hamiltonian simulation},\ }\href {https://doi.org/10.1137/21m1439298} {\bibfield  {journal} {\bibinfo  {journal} {SIAM Journal on Matrix Analysis and Applications}\ }\textbf {\bibinfo {volume} {43}},\ \bibinfo {pages} {1084–1108} (\bibinfo {year} {2022})}\BibitemShut {NoStop}%
\bibitem [{\citenamefont {Kökcü}\ \emph {et~al.}(2022)\citenamefont {Kökcü}, \citenamefont {Camps}, \citenamefont {Bassman~Oftelie}, \citenamefont {Freericks}, \citenamefont {de~Jong}, \citenamefont {Van~Beeumen},\ and\ \citenamefont {Kemper}}]{KoCa22}%
  \BibitemOpen
  \bibfield  {author} {\bibinfo {author} {\bibfnamefont {E.}~\bibnamefont {Kökcü}}, \bibinfo {author} {\bibfnamefont {D.}~\bibnamefont {Camps}}, \bibinfo {author} {\bibfnamefont {L.}~\bibnamefont {Bassman~Oftelie}}, \bibinfo {author} {\bibfnamefont {J.~K.}\ \bibnamefont {Freericks}}, \bibinfo {author} {\bibfnamefont {W.~A.}\ \bibnamefont {de~Jong}}, \bibinfo {author} {\bibfnamefont {R.}~\bibnamefont {Van~Beeumen}},\ and\ \bibinfo {author} {\bibfnamefont {A.~F.}\ \bibnamefont {Kemper}},\ }\bibfield  {title} {\bibinfo {title} {Algebraic compression of quantum circuits for hamiltonian evolution},\ }\bibfield  {journal} {\bibinfo  {journal} {Physical Review A}\ }\textbf {\bibinfo {volume} {105}},\ \href {https://doi.org/10.1103/physreva.105.032420} {10.1103/physreva.105.032420} (\bibinfo {year} {2022})\BibitemShut {NoStop}%
\bibitem [{\citenamefont {Aaronson}(2016)}]{Aa16}%
  \BibitemOpen
  \bibfield  {author} {\bibinfo {author} {\bibfnamefont {S.}~\bibnamefont {Aaronson}},\ }\href {https://arxiv.org/abs/1607.05256} {\bibinfo {title} {The complexity of quantum states and transformations: From quantum money to black holes}} (\bibinfo {year} {2016}),\ \Eprint {https://arxiv.org/abs/1607.05256} {arXiv:1607.05256 [quant-ph]} \BibitemShut {NoStop}%
\bibitem [{\citenamefont {Hackl}\ and\ \citenamefont {Myers}(2018)}]{HaMy18}%
  \BibitemOpen
  \bibfield  {author} {\bibinfo {author} {\bibfnamefont {L.}~\bibnamefont {Hackl}}\ and\ \bibinfo {author} {\bibfnamefont {R.~C.}\ \bibnamefont {Myers}},\ }\bibfield  {title} {\bibinfo {title} {Circuit complexity for free fermions},\ }\bibfield  {journal} {\bibinfo  {journal} {Journal of High Energy Physics}\ }\textbf {\bibinfo {volume} {2018}},\ \href {https://doi.org/10.1007/jhep07(2018)139} {10.1007/jhep07(2018)139} (\bibinfo {year} {2018})\BibitemShut {NoStop}%
\bibitem [{\citenamefont {Braccia}\ \emph {et~al.}(2025)\citenamefont {Braccia}, \citenamefont {Diaz}, \citenamefont {Larocca}, \citenamefont {Cerezo},\ and\ \citenamefont {García-Martín}}]{BrDi25}%
  \BibitemOpen
  \bibfield  {author} {\bibinfo {author} {\bibfnamefont {P.}~\bibnamefont {Braccia}}, \bibinfo {author} {\bibfnamefont {N.~L.}\ \bibnamefont {Diaz}}, \bibinfo {author} {\bibfnamefont {M.}~\bibnamefont {Larocca}}, \bibinfo {author} {\bibfnamefont {M.}~\bibnamefont {Cerezo}},\ and\ \bibinfo {author} {\bibfnamefont {D.}~\bibnamefont {García-Martín}},\ }\href {https://arxiv.org/abs/2505.24212} {\bibinfo {title} {Optimal haar random fermionic linear optics circuits}} (\bibinfo {year} {2025}),\ \Eprint {https://arxiv.org/abs/2505.24212} {arXiv:2505.24212 [quant-ph]} \BibitemShut {NoStop}%
\bibitem [{Note1()}]{Note1}%
  \BibitemOpen
  \bibinfo {note} {To see this, consider the application of a gate $U$ to qubits $i$ and $i+1$ of some state $\ket \varphi $. The difference $\protect \operatorname {LSR}_i(U \ket {\phi }) - \protect \operatorname {LSR}_i(\ket \phi )$ can be at most two, since $|\protect \operatorname {LSR}_{i}(\ket \chi ) - \protect \operatorname {LSR}_{i\pm 1}(\ket \chi )| \leq 1$ due to subadditivity for both $\ket \chi = \ket \varphi $ and $\ket \chi =U \ket {\varphi }$. Hence, even if all gates in a circuit generating $\ket {\psi }$ increase the Schmidt rank maximally, this circuit needs to contain at least $K/2$ gates.}\BibitemShut {Stop}%
\bibitem [{\citenamefont {Cudby}\ and\ \citenamefont {Strelchuk}(2025)}]{CuSt25}%
  \BibitemOpen
  \bibfield  {author} {\bibinfo {author} {\bibfnamefont {J.}~\bibnamefont {Cudby}}\ and\ \bibinfo {author} {\bibfnamefont {S.}~\bibnamefont {Strelchuk}},\ }\href {https://arxiv.org/abs/2307.12654} {\bibinfo {title} {Gaussian decomposition of magic states for matchgate computations}} (\bibinfo {year} {2025}),\ \Eprint {https://arxiv.org/abs/2307.12654} {arXiv:2307.12654 [quant-ph]} \BibitemShut {NoStop}%
\bibitem [{\citenamefont {Jozsa}\ and\ \citenamefont {Van Den~Nest}(2014)}]{JoVa14}%
  \BibitemOpen
  \bibfield  {author} {\bibinfo {author} {\bibfnamefont {R.}~\bibnamefont {Jozsa}}\ and\ \bibinfo {author} {\bibfnamefont {M.}~\bibnamefont {Van Den~Nest}},\ }\bibfield  {title} {\bibinfo {title} {Classical simulation complexity of extended clifford circuits},\ }\href@noop {} {\bibfield  {journal} {\bibinfo  {journal} {Quantum Info. Comput.}\ }\textbf {\bibinfo {volume} {14}},\ \bibinfo {pages} {633–648} (\bibinfo {year} {2014})}\BibitemShut {NoStop}%
\bibitem [{\citenamefont {Bravyi}\ \emph {et~al.}(2022)\citenamefont {Bravyi}, \citenamefont {Gosset},\ and\ \citenamefont {Liu}}]{BrGo22}%
  \BibitemOpen
  \bibfield  {author} {\bibinfo {author} {\bibfnamefont {S.}~\bibnamefont {Bravyi}}, \bibinfo {author} {\bibfnamefont {D.}~\bibnamefont {Gosset}},\ and\ \bibinfo {author} {\bibfnamefont {Y.}~\bibnamefont {Liu}},\ }\bibfield  {title} {\bibinfo {title} {How to simulate quantum measurement without computing marginals},\ }\href {https://doi.org/10.1103/PhysRevLett.128.220503} {\bibfield  {journal} {\bibinfo  {journal} {Phys. Rev. Lett.}\ }\textbf {\bibinfo {volume} {128}},\ \bibinfo {pages} {220503} (\bibinfo {year} {2022})}\BibitemShut {NoStop}%
\bibitem [{Note2()}]{Note2}%
  \BibitemOpen
  \bibinfo {note} {This approach requires that the local Hilbert space dimension is finite.}\BibitemShut {Stop}%
\bibitem [{\citenamefont {Soeda}\ \emph {et~al.}(2014)\citenamefont {Soeda}, \citenamefont {Akibue},\ and\ \citenamefont {Murao}}]{SoAk14}%
  \BibitemOpen
  \bibfield  {author} {\bibinfo {author} {\bibfnamefont {A.}~\bibnamefont {Soeda}}, \bibinfo {author} {\bibfnamefont {S.}~\bibnamefont {Akibue}},\ and\ \bibinfo {author} {\bibfnamefont {M.}~\bibnamefont {Murao}},\ }\bibfield  {title} {\bibinfo {title} {Two-party locc convertibility of quadpartite states and kraus–cirac number of two-qubit unitaries},\ }\href {https://doi.org/10.1088/1751-8113/47/42/424036} {\bibfield  {journal} {\bibinfo  {journal} {J. Phys. A: Math. Theor.}\ }\textbf {\bibinfo {volume} {47}},\ \bibinfo {pages} {424036} (\bibinfo {year} {2014})}\BibitemShut {NoStop}%
\bibitem [{\citenamefont {Jordan}\ and\ \citenamefont {Wigner}(1928)}]{JoWi28}%
  \BibitemOpen
  \bibfield  {author} {\bibinfo {author} {\bibfnamefont {P.}~\bibnamefont {Jordan}}\ and\ \bibinfo {author} {\bibfnamefont {E.}~\bibnamefont {Wigner}},\ }\bibfield  {title} {\bibinfo {title} {{\"U}ber das paulische {\"a}quivalenzverbot},\ }\href {https://doi.org/10.1007/BF01331938} {\bibfield  {journal} {\bibinfo  {journal} {Zeitschrift f{\"u}r Physik}\ }\textbf {\bibinfo {volume} {47}},\ \bibinfo {pages} {631} (\bibinfo {year} {1928})}\BibitemShut {NoStop}%
\bibitem [{\citenamefont {Raffenetti}\ and\ \citenamefont {Ruedenberg}(1969)}]{RaRu69}%
  \BibitemOpen
  \bibfield  {author} {\bibinfo {author} {\bibfnamefont {R.~C.}\ \bibnamefont {Raffenetti}}\ and\ \bibinfo {author} {\bibfnamefont {K.}~\bibnamefont {Ruedenberg}},\ }\bibfield  {title} {\bibinfo {title} {Parametrization of an orthogonal matrix in terms of generalized eulerian angles},\ }\href {https://doi.org/10.1002/qua.560040725} {\bibfield  {journal} {\bibinfo  {journal} {International Journal of Quantum Chemistry}\ }\textbf {\bibinfo {volume} {4}},\ \bibinfo {pages} {625} (\bibinfo {year} {1969})}\BibitemShut {NoStop}%
\bibitem [{\citenamefont {Hoffman}\ \emph {et~al.}(1972)\citenamefont {Hoffman}, \citenamefont {Raffenetti},\ and\ \citenamefont {Ruedenberg}}]{HoRa72}%
  \BibitemOpen
  \bibfield  {author} {\bibinfo {author} {\bibfnamefont {D.~K.}\ \bibnamefont {Hoffman}}, \bibinfo {author} {\bibfnamefont {R.~C.}\ \bibnamefont {Raffenetti}},\ and\ \bibinfo {author} {\bibfnamefont {K.}~\bibnamefont {Ruedenberg}},\ }\bibfield  {title} {\bibinfo {title} {Generalization of euler angles to n‐dimensional orthogonal matrices},\ }\href {https://doi.org/10.1063/1.1666011} {\bibfield  {journal} {\bibinfo  {journal} {Journal of Mathematical Physics}\ }\textbf {\bibinfo {volume} {13}},\ \bibinfo {pages} {528} (\bibinfo {year} {1972})}\BibitemShut {NoStop}%
\bibitem [{\citenamefont {Surace}\ and\ \citenamefont {Tagliacozzo}(2022)}]{SuTa20}%
  \BibitemOpen
  \bibfield  {author} {\bibinfo {author} {\bibfnamefont {J.}~\bibnamefont {Surace}}\ and\ \bibinfo {author} {\bibfnamefont {L.}~\bibnamefont {Tagliacozzo}},\ }\bibfield  {title} {\bibinfo {title} {{Fermionic Gaussian states: an introduction to numerical approaches}},\ }\href {https://doi.org/10.21468/SciPostPhysLectNotes.54} {\bibfield  {journal} {\bibinfo  {journal} {SciPost Phys. Lect. Notes}\ ,\ \bibinfo {pages} {54}} (\bibinfo {year} {2022})}\BibitemShut {NoStop}%
\bibitem [{\citenamefont {Yang}(1967)}]{Ya67}%
  \BibitemOpen
  \bibfield  {author} {\bibinfo {author} {\bibfnamefont {C.~N.}\ \bibnamefont {Yang}},\ }\bibfield  {title} {\bibinfo {title} {Some exact results for the many-body problem in one dimension with repulsive delta-function interaction},\ }\href {https://doi.org/10.1103/PhysRevLett.19.1312} {\bibfield  {journal} {\bibinfo  {journal} {Phys. Rev. Lett.}\ }\textbf {\bibinfo {volume} {19}},\ \bibinfo {pages} {1312} (\bibinfo {year} {1967})}\BibitemShut {NoStop}%
\bibitem [{\citenamefont {Baxter}(1972)}]{Ba72}%
  \BibitemOpen
  \bibfield  {author} {\bibinfo {author} {\bibfnamefont {R.~J.}\ \bibnamefont {Baxter}},\ }\bibfield  {title} {\bibinfo {title} {Partition function of the eight-vertex lattice model},\ }\href {https://doi.org/https://doi.org/10.1016/0003-4916(72)90335-1} {\bibfield  {journal} {\bibinfo  {journal} {Annals of Physics}\ }\textbf {\bibinfo {volume} {70}},\ \bibinfo {pages} {193} (\bibinfo {year} {1972})}\BibitemShut {NoStop}%
\bibitem [{\citenamefont {Smirnov}(1992)}]{Sm92}%
  \BibitemOpen
  \bibfield  {author} {\bibinfo {author} {\bibfnamefont {F.~A.}\ \bibnamefont {Smirnov}},\ }\href@noop {} {\emph {\bibinfo {title} {Form Factors in Completely Integrable Models of Quantum Field Theory}}},\ \bibinfo {series} {Advanced Series in Mathematical Physics}, Vol.~\bibinfo {volume} {14}\ (\bibinfo  {publisher} {World Scientific},\ \bibinfo {address} {Singapore},\ \bibinfo {year} {1992})\BibitemShut {NoStop}%
\bibitem [{\citenamefont {Korepin}\ \emph {et~al.}(1993)\citenamefont {Korepin}, \citenamefont {Bogoliubov},\ and\ \citenamefont {Izergin}}]{KoBo93}%
  \BibitemOpen
  \bibfield  {author} {\bibinfo {author} {\bibfnamefont {V.~E.}\ \bibnamefont {Korepin}}, \bibinfo {author} {\bibfnamefont {N.~M.}\ \bibnamefont {Bogoliubov}},\ and\ \bibinfo {author} {\bibfnamefont {A.~G.}\ \bibnamefont {Izergin}},\ }\href@noop {} {\emph {\bibinfo {title} {Quantum Inverse Scattering Method and Correlation Functions}}},\ Cambridge Monographs on Mathematical Physics\ (\bibinfo  {publisher} {Cambridge University Press},\ \bibinfo {year} {1993})\BibitemShut {NoStop}%
\bibitem [{\citenamefont {Rothe}(1800)}]{origin_of_telephone}%
  \BibitemOpen
  \bibfield  {author} {\bibinfo {author} {\bibfnamefont {H.~A.}\ \bibnamefont {Rothe}},\ }\bibfield  {title} {\bibinfo {title} {Über permutationen, in beziehung auf die stellen ihrer elemente. anwendung der daraus abgeleiteten sätze auf das eliminationsproblem},\ }\href@noop {} {\bibfield  {journal} {\bibinfo  {journal} {Sammlung combinatorisch-analytischer Abhandlungen / herausgegeben von Carl Friedrich Hindenburg}\ }\textbf {\bibinfo {volume} {2}},\ \bibinfo {pages} {282} (\bibinfo {year} {1800})}\BibitemShut {NoStop}%
\bibitem [{Note3()}]{Note3}%
  \BibitemOpen
  \bibinfo {note} {Since the covariance matrix does not contain phase information, the circuit prepares $\ket {\psi }$ up to a global phase. The phase can be fixed by, e.g., knowing the inner product with a reference state~\cite {DiKo24,ReOs24}.}\BibitemShut {Stop}%
\bibitem [{Note4()}]{Note4}%
  \BibitemOpen
  \bibinfo {note} {Due to parity preservation, if a matchgate does not generate entanglement when acting on a computational basis state, it simply applies a global phase and possibly flips the bits.}\BibitemShut {Stop}%
\bibitem [{Note5()}]{Note5}%
  \BibitemOpen
  \bibinfo {note} {One can generate any computational basis state from any other one by acting with $\protect \mathcal {O}(n)$ matchgates of the form $G(X,X) = X\otimes X$ arranged in two layers (possibly using an auxiliary qubit for preserving parity). Such circuits, however, do not satisfy the conditions of the theorem.}\BibitemShut {Stop}%
\bibitem [{Note6()}]{Note6}%
  \BibitemOpen
  \bibinfo {note} {Note that a representation of (unnormalized) fermionic Gaussian states via a Gaussian operator $\exp (\DOTSB \sum@ \slimits@ _{kl} c_k^\dagger h_{kl} c_l^\dagger )$, where $c_k^\dagger = (\gamma _{2k-1}-i\gamma _{2k})/2$ are the creation operators, acting on the vacuum state $\ket {0\protect \ldots 0}$ allows to compute amplitudes in the computational basis~\cite {BeSo17} or other product bases~\cite {BaRe24,RaMi25}. One needs to take care, however, to preserve the global phase when transforming any other representation of FGSs to the above one.}\BibitemShut {Stop}%
\bibitem [{\citenamefont {Fannes}\ \emph {et~al.}(1992)\citenamefont {Fannes}, \citenamefont {Nachtergaele},\ and\ \citenamefont {Werner}}]{FaNa92}%
  \BibitemOpen
  \bibfield  {author} {\bibinfo {author} {\bibfnamefont {M.}~\bibnamefont {Fannes}}, \bibinfo {author} {\bibfnamefont {B.}~\bibnamefont {Nachtergaele}},\ and\ \bibinfo {author} {\bibfnamefont {R.~F.}\ \bibnamefont {Werner}},\ }\bibfield  {title} {\bibinfo {title} {Finitely correlated states on quantum spin chains},\ }\href {https://doi.org/10.1007/BF02099178} {\bibfield  {journal} {\bibinfo  {journal} {Communications in Mathematical Physics}\ }\textbf {\bibinfo {volume} {144}},\ \bibinfo {pages} {443} (\bibinfo {year} {1992})}\BibitemShut {NoStop}%
\bibitem [{\citenamefont {\"Ostlund}\ and\ \citenamefont {Rommer}(1995)}]{OsSt95}%
  \BibitemOpen
  \bibfield  {author} {\bibinfo {author} {\bibfnamefont {S.}~\bibnamefont {\"Ostlund}}\ and\ \bibinfo {author} {\bibfnamefont {S.}~\bibnamefont {Rommer}},\ }\bibfield  {title} {\bibinfo {title} {Thermodynamic limit of density matrix renormalization},\ }\href {https://doi.org/10.1103/PhysRevLett.75.3537} {\bibfield  {journal} {\bibinfo  {journal} {Phys. Rev. Lett.}\ }\textbf {\bibinfo {volume} {75}},\ \bibinfo {pages} {3537} (\bibinfo {year} {1995})}\BibitemShut {NoStop}%
\bibitem [{\citenamefont {Rommer}\ and\ \citenamefont {\"Ostlund}(1997)}]{RoOs97}%
  \BibitemOpen
  \bibfield  {author} {\bibinfo {author} {\bibfnamefont {S.}~\bibnamefont {Rommer}}\ and\ \bibinfo {author} {\bibfnamefont {S.}~\bibnamefont {\"Ostlund}},\ }\bibfield  {title} {\bibinfo {title} {Class of ansatz wave functions for one-dimensional spin systems and their relation to the density matrix renormalization group},\ }\href {https://doi.org/10.1103/PhysRevB.55.2164} {\bibfield  {journal} {\bibinfo  {journal} {Phys. Rev. B}\ }\textbf {\bibinfo {volume} {55}},\ \bibinfo {pages} {2164} (\bibinfo {year} {1997})}\BibitemShut {NoStop}%
\bibitem [{\citenamefont {Hauschild}\ and\ \citenamefont {Pollmann}(2018)}]{HaPo18}%
  \BibitemOpen
  \bibfield  {author} {\bibinfo {author} {\bibfnamefont {J.}~\bibnamefont {Hauschild}}\ and\ \bibinfo {author} {\bibfnamefont {F.}~\bibnamefont {Pollmann}},\ }\bibfield  {title} {\bibinfo {title} {Efficient numerical simulations with tensor networks: Tensor network python (tenpy)},\ }\bibfield  {journal} {\bibinfo  {journal} {SciPost Physics Lecture Notes}\ }\href {https://doi.org/10.21468/scipostphyslectnotes.5} {10.21468/scipostphyslectnotes.5} (\bibinfo {year} {2018})\BibitemShut {NoStop}%
\bibitem [{Note7()}]{Note7}%
  \BibitemOpen
  \bibinfo {note} {More generally, we can consider unitaries generated by an exponential of $\kappa = 4$ Majorana operators. Further generalization of our constructions to arbitary even $\kappa = \protect \mathcal {O}(1)$ is likewise possible.}\BibitemShut {Stop}%
\bibitem [{\citenamefont {Vidal}(2004)}]{Vi04}%
  \BibitemOpen
  \bibfield  {author} {\bibinfo {author} {\bibfnamefont {G.}~\bibnamefont {Vidal}},\ }\bibfield  {title} {\bibinfo {title} {Efficient simulation of one-dimensional quantum many-body systems},\ }\href {https://doi.org/10.1103/PhysRevLett.93.040502} {\bibfield  {journal} {\bibinfo  {journal} {Phys. Rev. Lett.}\ }\textbf {\bibinfo {volume} {93}},\ \bibinfo {pages} {040502} (\bibinfo {year} {2004})}\BibitemShut {NoStop}%
\bibitem [{Note8()}]{Note8}%
  \BibitemOpen
  \bibinfo {note} {This effect can be observed, for instance, in states generated by applying a few layers of random matchgates to a product state.}\BibitemShut {Stop}%
\bibitem [{Note9()}]{Note9}%
  \BibitemOpen
  \bibinfo {note} {The backward lightcone for a splitting $1,\protect \ldots ,k\vert k+1,\protect \ldots ,n$ can be defined by two conditions: (i) The last gate acting on qubits $k$ and $k+1$ is in the lightcone, and (ii) whenever for two two-qubit gates $G$, $G'$, one of the structures $G G'$, $(G\otimes \protect \mathds {1}_2)(\protect \mathds {1}_2 \otimes G')$, or $(\protect \mathds {1}_2\otimes G)(G' \otimes \protect \mathds {1}_2)$ is present in the circuit, if $G$ is in the lightcone, then $G'$ is also in the lightcone.}\BibitemShut {Stop}%
\bibitem [{\citenamefont {Horn}\ and\ \citenamefont {Johnson}(2013)}]{HoJo13}%
  \BibitemOpen
  \bibfield  {author} {\bibinfo {author} {\bibfnamefont {R.~A.}\ \bibnamefont {Horn}}\ and\ \bibinfo {author} {\bibfnamefont {C.~R.}\ \bibnamefont {Johnson}},\ }\href@noop {} {\emph {\bibinfo {title} {Matrix Analysis}}},\ \bibinfo {edition} {2nd}\ ed.\ (\bibinfo  {publisher} {Cambridge University Press},\ \bibinfo {address} {Cambridge; New York},\ \bibinfo {year} {2013})\BibitemShut {NoStop}%
\bibitem [{Note10()}]{Note10}%
  \BibitemOpen
  \bibinfo {note} {Here, there is generally no well defined notion of entangled pairs anymore.}\BibitemShut {Stop}%
\bibitem [{\citenamefont {Pfeuty}(1970)}]{Pf70}%
  \BibitemOpen
  \bibfield  {author} {\bibinfo {author} {\bibfnamefont {P.}~\bibnamefont {Pfeuty}},\ }\bibfield  {title} {\bibinfo {title} {The one-dimensional ising model with a transverse field},\ }\href {https://doi.org/https://doi.org/10.1016/0003-4916(70)90270-8} {\bibfield  {journal} {\bibinfo  {journal} {Annals of Physics}\ }\textbf {\bibinfo {volume} {57}},\ \bibinfo {pages} {79} (\bibinfo {year} {1970})}\BibitemShut {NoStop}%
\bibitem [{Note11()}]{Note11}%
  \BibitemOpen
  \bibinfo {note} {For $g<1$, the CM formalism yields as a ground state an FGS with long-range correlations, for which our algorithm is not well-suited.}\BibitemShut {Stop}%
\bibitem [{\citenamefont {Jobst}\ \emph {et~al.}(2022)\citenamefont {Jobst}, \citenamefont {Smith},\ and\ \citenamefont {Pollmann}}]{JoSm22}%
  \BibitemOpen
  \bibfield  {author} {\bibinfo {author} {\bibfnamefont {B.}~\bibnamefont {Jobst}}, \bibinfo {author} {\bibfnamefont {A.}~\bibnamefont {Smith}},\ and\ \bibinfo {author} {\bibfnamefont {F.}~\bibnamefont {Pollmann}},\ }\bibfield  {title} {\bibinfo {title} {Finite-depth scaling of infinite quantum circuits for quantum critical points},\ }\href {https://doi.org/10.1103/PhysRevResearch.4.033118} {\bibfield  {journal} {\bibinfo  {journal} {Phys. Rev. Res.}\ }\textbf {\bibinfo {volume} {4}},\ \bibinfo {pages} {033118} (\bibinfo {year} {2022})}\BibitemShut {NoStop}%
\bibitem [{\citenamefont {Vodola}\ \emph {et~al.}(2014)\citenamefont {Vodola}, \citenamefont {Lepori}, \citenamefont {Ercolessi}, \citenamefont {Gorshkov},\ and\ \citenamefont {Pupillo}}]{VoLe14}%
  \BibitemOpen
  \bibfield  {author} {\bibinfo {author} {\bibfnamefont {D.}~\bibnamefont {Vodola}}, \bibinfo {author} {\bibfnamefont {L.}~\bibnamefont {Lepori}}, \bibinfo {author} {\bibfnamefont {E.}~\bibnamefont {Ercolessi}}, \bibinfo {author} {\bibfnamefont {A.~V.}\ \bibnamefont {Gorshkov}},\ and\ \bibinfo {author} {\bibfnamefont {G.}~\bibnamefont {Pupillo}},\ }\bibfield  {title} {\bibinfo {title} {Kitaev chains with long-range pairing},\ }\href {https://doi.org/10.1103/PhysRevLett.113.156402} {\bibfield  {journal} {\bibinfo  {journal} {Phys. Rev. Lett.}\ }\textbf {\bibinfo {volume} {113}},\ \bibinfo {pages} {156402} (\bibinfo {year} {2014})}\BibitemShut {NoStop}%
\bibitem [{\citenamefont {Vodola}\ \emph {et~al.}(2016)\citenamefont {Vodola}, \citenamefont {Lepori}, \citenamefont {Ercolessi},\ and\ \citenamefont {Pupillo}}]{VoLe16}%
  \BibitemOpen
  \bibfield  {author} {\bibinfo {author} {\bibfnamefont {D.}~\bibnamefont {Vodola}}, \bibinfo {author} {\bibfnamefont {L.}~\bibnamefont {Lepori}}, \bibinfo {author} {\bibfnamefont {E.}~\bibnamefont {Ercolessi}},\ and\ \bibinfo {author} {\bibfnamefont {G.}~\bibnamefont {Pupillo}},\ }\bibfield  {title} {\bibinfo {title} {Long-range ising and kitaev models: phases, correlations and edge modes},\ }\href {https://doi.org/10.1088/1367-2630/18/1/015001} {\bibfield  {journal} {\bibinfo  {journal} {New Journal of Physics}\ }\textbf {\bibinfo {volume} {18}},\ \bibinfo {pages} {015001} (\bibinfo {year} {2016})}\BibitemShut {NoStop}%
\bibitem [{\citenamefont {Casas}\ \emph {et~al.}(2026)\citenamefont {Casas}, \citenamefont {Braccia}, \citenamefont {Élie Gouzien}, \citenamefont {Cerezo},\ and\ \citenamefont {García-Martín}}]{CaBr26}%
  \BibitemOpen
  \bibfield  {author} {\bibinfo {author} {\bibfnamefont {B.}~\bibnamefont {Casas}}, \bibinfo {author} {\bibfnamefont {P.}~\bibnamefont {Braccia}}, \bibinfo {author} {\bibnamefont {Élie Gouzien}}, \bibinfo {author} {\bibfnamefont {M.}~\bibnamefont {Cerezo}},\ and\ \bibinfo {author} {\bibfnamefont {D.}~\bibnamefont {García-Martín}},\ }\href {https://arxiv.org/abs/2602.05425} {\bibinfo {title} {Matchgate synthesis via clifford matchgates and $t$ gates}} (\bibinfo {year} {2026}),\ \Eprint {https://arxiv.org/abs/2602.05425} {arXiv:2602.05425 [quant-ph]} \BibitemShut {NoStop}%
\bibitem [{\citenamefont {Gottlieb}\ and\ \citenamefont {Mauser}(2005)}]{GoMa05}%
  \BibitemOpen
  \bibfield  {author} {\bibinfo {author} {\bibfnamefont {A.~D.}\ \bibnamefont {Gottlieb}}\ and\ \bibinfo {author} {\bibfnamefont {N.~J.}\ \bibnamefont {Mauser}},\ }\bibfield  {title} {\bibinfo {title} {New measure of electron correlation},\ }\href {https://doi.org/10.1103/PhysRevLett.95.123003} {\bibfield  {journal} {\bibinfo  {journal} {Phys. Rev. Lett.}\ }\textbf {\bibinfo {volume} {95}},\ \bibinfo {pages} {123003} (\bibinfo {year} {2005})}\BibitemShut {NoStop}%
\bibitem [{\citenamefont {Gottlieb}\ and\ \citenamefont {Mauser}(2006)}]{GoMa06}%
  \BibitemOpen
  \bibfield  {author} {\bibinfo {author} {\bibfnamefont {A.~D.}\ \bibnamefont {Gottlieb}}\ and\ \bibinfo {author} {\bibfnamefont {N.~J.}\ \bibnamefont {Mauser}},\ }\href {https://arxiv.org/abs/quant-ph/0608171} {\bibinfo {title} {Properties of nonfreeness: an entropy measure of electron correlation}} (\bibinfo {year} {2006}),\ \Eprint {https://arxiv.org/abs/quant-ph/0608171} {arXiv:quant-ph/0608171 [quant-ph]} \BibitemShut {NoStop}%
\bibitem [{\citenamefont {Lumia}\ \emph {et~al.}(2024)\citenamefont {Lumia}, \citenamefont {Tirrito}, \citenamefont {Fazio},\ and\ \citenamefont {Collura}}]{LuTi24}%
  \BibitemOpen
  \bibfield  {author} {\bibinfo {author} {\bibfnamefont {L.}~\bibnamefont {Lumia}}, \bibinfo {author} {\bibfnamefont {E.}~\bibnamefont {Tirrito}}, \bibinfo {author} {\bibfnamefont {R.}~\bibnamefont {Fazio}},\ and\ \bibinfo {author} {\bibfnamefont {M.}~\bibnamefont {Collura}},\ }\bibfield  {title} {\bibinfo {title} {Measurement-induced transitions beyond gaussianity: A single particle description},\ }\href {https://doi.org/10.1103/PhysRevResearch.6.023176} {\bibfield  {journal} {\bibinfo  {journal} {Phys. Rev. Res.}\ }\textbf {\bibinfo {volume} {6}},\ \bibinfo {pages} {023176} (\bibinfo {year} {2024})}\BibitemShut {NoStop}%
\bibitem [{\citenamefont {Lyu}\ and\ \citenamefont {Bu}(2024)}]{LyBu24}%
  \BibitemOpen
  \bibfield  {author} {\bibinfo {author} {\bibfnamefont {X.}~\bibnamefont {Lyu}}\ and\ \bibinfo {author} {\bibfnamefont {K.}~\bibnamefont {Bu}},\ }\href {https://arxiv.org/abs/2409.08180} {\bibinfo {title} {Fermionic gaussian testing and non-gaussian measures via convolution}} (\bibinfo {year} {2024}),\ \Eprint {https://arxiv.org/abs/2409.08180} {arXiv:2409.08180 [quant-ph]} \BibitemShut {NoStop}%
\bibitem [{\citenamefont {Coffman}\ \emph {et~al.}(2025)\citenamefont {Coffman}, \citenamefont {Smith},\ and\ \citenamefont {Gao}}]{CoSm25}%
  \BibitemOpen
  \bibfield  {author} {\bibinfo {author} {\bibfnamefont {L.}~\bibnamefont {Coffman}}, \bibinfo {author} {\bibfnamefont {G.}~\bibnamefont {Smith}},\ and\ \bibinfo {author} {\bibfnamefont {X.}~\bibnamefont {Gao}},\ }\href {https://arxiv.org/abs/2501.06179} {\bibinfo {title} {Measuring non-gaussian magic in fermions: Convolution, entropy, and the violation of wick's theorem and the matchgate identity}} (\bibinfo {year} {2025}),\ \Eprint {https://arxiv.org/abs/2501.06179} {arXiv:2501.06179 [quant-ph]} \BibitemShut {NoStop}%
\bibitem [{\citenamefont {Sierant}\ \emph {et~al.}(2026)\citenamefont {Sierant}, \citenamefont {Stornati},\ and\ \citenamefont {Turkeshi}}]{SiSt26}%
  \BibitemOpen
  \bibfield  {author} {\bibinfo {author} {\bibfnamefont {P.}~\bibnamefont {Sierant}}, \bibinfo {author} {\bibfnamefont {P.}~\bibnamefont {Stornati}},\ and\ \bibinfo {author} {\bibfnamefont {X.}~\bibnamefont {Turkeshi}},\ }\bibfield  {title} {\bibinfo {title} {Fermionic magic resources of quantum many-body systems},\ }\href {https://doi.org/10.1103/3yx4-1j27} {\bibfield  {journal} {\bibinfo  {journal} {PRX Quantum}\ }\textbf {\bibinfo {volume} {7}},\ \bibinfo {pages} {010302} (\bibinfo {year} {2026})}\BibitemShut {NoStop}%
\bibitem [{\citenamefont {Langer}\ \emph {et~al.}(2026)\citenamefont {Langer}, \citenamefont {Morral-Yepes}, \citenamefont {Gammon-Smith}, \citenamefont {Pollmann},\ and\ \citenamefont {Kraus}}]{zenodo}%
  \BibitemOpen
  \bibfield  {author} {\bibinfo {author} {\bibfnamefont {M.}~\bibnamefont {Langer}}, \bibinfo {author} {\bibfnamefont {R.}~\bibnamefont {Morral-Yepes}}, \bibinfo {author} {\bibfnamefont {A.}~\bibnamefont {Gammon-Smith}}, \bibinfo {author} {\bibfnamefont {F.}~\bibnamefont {Pollmann}},\ and\ \bibinfo {author} {\bibfnamefont {B.}~\bibnamefont {Kraus}},\ }\href {https://doi.org/10.5281/zenodo.18875164} {\bibinfo {title} {Matchgate circuit representation of fermionic gaussian states: optimal preparation, approximation, and classical simulation}} (\bibinfo {year} {2026})\BibitemShut {NoStop}%
\bibitem [{\citenamefont {Shepherd}\ and\ \citenamefont {Bremner}(2009)}]{ShBr09}%
  \BibitemOpen
  \bibfield  {author} {\bibinfo {author} {\bibfnamefont {D.}~\bibnamefont {Shepherd}}\ and\ \bibinfo {author} {\bibfnamefont {M.~J.}\ \bibnamefont {Bremner}},\ }\bibfield  {title} {\bibinfo {title} {Temporally unstructured quantum computation},\ }\href {https://doi.org/10.1098/rspa.2008.0443} {\bibfield  {journal} {\bibinfo  {journal} {Proceedings of the Royal Society A: Mathematical, Physical and Engineering Sciences}\ }\textbf {\bibinfo {volume} {465}},\ \bibinfo {pages} {1413} (\bibinfo {year} {2009})},\ \Eprint {https://arxiv.org/abs/https://royalsocietypublishing.org/rspa/article-pdf/465/2105/1413/753599/rspa.2008.0443.pdf} {https://royalsocietypublishing.org/rspa/article-pdf/465/2105/1413/753599/rspa.2008.0443.pdf} \BibitemShut {NoStop}%
\bibitem [{\citenamefont {Bremner}\ \emph {et~al.}(2010)\citenamefont {Bremner}, \citenamefont {Jozsa},\ and\ \citenamefont {Shepherd}}]{BrJo10}%
  \BibitemOpen
  \bibfield  {author} {\bibinfo {author} {\bibfnamefont {M.~J.}\ \bibnamefont {Bremner}}, \bibinfo {author} {\bibfnamefont {R.}~\bibnamefont {Jozsa}},\ and\ \bibinfo {author} {\bibfnamefont {D.~J.}\ \bibnamefont {Shepherd}},\ }\bibfield  {title} {\bibinfo {title} {Classical simulation of commuting quantum computations implies collapse of the polynomial hierarchy},\ }\href {https://doi.org/10.1098/rspa.2010.0301} {\bibfield  {journal} {\bibinfo  {journal} {Proceedings of the Royal Society A: Mathematical, Physical and Engineering Sciences}\ }\textbf {\bibinfo {volume} {467}},\ \bibinfo {pages} {459} (\bibinfo {year} {2010})},\ \Eprint {https://arxiv.org/abs/https://royalsocietypublishing.org/rspa/article-pdf/467/2126/459/789839/rspa.2010.0301.pdf} {https://royalsocietypublishing.org/rspa/article-pdf/467/2126/459/789839/rspa.2010.0301.pdf} \BibitemShut {NoStop}%
\bibitem [{\citenamefont {Serafini}(2017)}]{Se17qcv}%
  \BibitemOpen
  \bibfield  {author} {\bibinfo {author} {\bibfnamefont {A.}~\bibnamefont {Serafini}},\ }\href {https://doi.org/10.1201/9781315118727} {\emph {\bibinfo {title} {Quantum Continuous Variables: A Primer of Theoretical Methods}}}\ (\bibinfo  {publisher} {CRC Press},\ \bibinfo {year} {2017})\BibitemShut {NoStop}%
\bibitem [{\citenamefont {van Luijk}\ \emph {et~al.}(2024)\citenamefont {van Luijk}, \citenamefont {Galke}, \citenamefont {Hahn},\ and\ \citenamefont {Burgarth}}]{VaGa24}%
  \BibitemOpen
  \bibfield  {author} {\bibinfo {author} {\bibfnamefont {L.}~\bibnamefont {van Luijk}}, \bibinfo {author} {\bibfnamefont {N.}~\bibnamefont {Galke}}, \bibinfo {author} {\bibfnamefont {A.}~\bibnamefont {Hahn}},\ and\ \bibinfo {author} {\bibfnamefont {D.}~\bibnamefont {Burgarth}},\ }\bibfield  {title} {\bibinfo {title} {Error bounds for lie group representations in quantum mechanics},\ }\href {https://doi.org/10.1088/1751-8121/ad288b} {\bibfield  {journal} {\bibinfo  {journal} {Journal of Physics A: Mathematical and Theoretical}\ }\textbf {\bibinfo {volume} {57}},\ \bibinfo {pages} {105301} (\bibinfo {year} {2024})}\BibitemShut {NoStop}%
\bibitem [{Note12()}]{Note12}%
  \BibitemOpen
  \bibinfo {note} {The aim of this permutation is to have the two linearly independent vectors in the first two columns. Any permutation acting only on party $\protect \mathbf {A}$ cannot change the entanglement w.r.t. to the bipartitions of interest.}\BibitemShut {Stop}%
\bibitem [{Note13()}]{Note13}%
  \BibitemOpen
  \bibinfo {note} {Note that computing $R_\protect \mathbf {A}$ and $R_\protect \mathbf {C}$ is a convenient step for the following, but to disentangle $\ket {\psi }$ as above it is not necessary to apply any operations on $\protect \mathbf {A}$ and $\protect \mathbf {C}$.}\BibitemShut {Stop}%
\bibitem [{\citenamefont {Hirsch}\ and\ \citenamefont {Marsiglio}(1989)}]{HiMa89}%
  \BibitemOpen
  \bibfield  {author} {\bibinfo {author} {\bibfnamefont {J.~E.}\ \bibnamefont {Hirsch}}\ and\ \bibinfo {author} {\bibfnamefont {F.}~\bibnamefont {Marsiglio}},\ }\bibfield  {title} {\bibinfo {title} {Superconducting state in an oxygen hole metal},\ }\href {https://doi.org/10.1103/PhysRevB.39.11515} {\bibfield  {journal} {\bibinfo  {journal} {Phys. Rev. B}\ }\textbf {\bibinfo {volume} {39}},\ \bibinfo {pages} {11515} (\bibinfo {year} {1989})}\BibitemShut {NoStop}%
\bibitem [{\citenamefont {Becca}\ and\ \citenamefont {Sorella}(2017)}]{BeSo17}%
  \BibitemOpen
  \bibfield  {author} {\bibinfo {author} {\bibfnamefont {F.}~\bibnamefont {Becca}}\ and\ \bibinfo {author} {\bibfnamefont {S.}~\bibnamefont {Sorella}},\ }\href@noop {} {\emph {\bibinfo {title} {Quantum Monte Carlo Approaches for Correlated Systems}}}\ (\bibinfo  {publisher} {Cambridge University Press},\ \bibinfo {year} {2017})\BibitemShut {NoStop}%
\bibitem [{\citenamefont {Tarighi}\ \emph {et~al.}(2024)\citenamefont {Tarighi}, \citenamefont {Khasseh},\ and\ \citenamefont {Rajabpour}}]{BaRe24}%
  \BibitemOpen
  \bibfield  {author} {\bibinfo {author} {\bibfnamefont {B.}~\bibnamefont {Tarighi}}, \bibinfo {author} {\bibfnamefont {R.}~\bibnamefont {Khasseh}},\ and\ \bibinfo {author} {\bibfnamefont {M.~A.}\ \bibnamefont {Rajabpour}},\ }\bibfield  {title} {\bibinfo {title} {Efficient representation of gaussian fermionic pure states in noncomputational bases},\ }\href {https://doi.org/10.1103/PhysRevA.109.062214} {\bibfield  {journal} {\bibinfo  {journal} {Phys. Rev. A}\ }\textbf {\bibinfo {volume} {109}},\ \bibinfo {pages} {062214} (\bibinfo {year} {2024})}\BibitemShut {NoStop}%
\bibitem [{\citenamefont {Rajabpour}\ \emph {et~al.}(2025)\citenamefont {Rajabpour}, \citenamefont {Mirjafarlou},\ and\ \citenamefont {Khasseh}}]{RaMi25}%
  \BibitemOpen
  \bibfield  {author} {\bibinfo {author} {\bibfnamefont {M.~A.}\ \bibnamefont {Rajabpour}}, \bibinfo {author} {\bibfnamefont {M.~A.~S.}\ \bibnamefont {Mirjafarlou}},\ and\ \bibinfo {author} {\bibfnamefont {R.}~\bibnamefont {Khasseh}},\ }\bibfield  {title} {\bibinfo {title} {Explicit pfaffian formula for amplitudes of fermionic gaussian pure states in arbitrary pauli bases},\ }\href {https://doi.org/10.1103/PhysRevB.111.235102} {\bibfield  {journal} {\bibinfo  {journal} {Phys. Rev. B}\ }\textbf {\bibinfo {volume} {111}},\ \bibinfo {pages} {235102} (\bibinfo {year} {2025})}\BibitemShut {NoStop}%
\end{thebibliography}%

\end{document}